\documentclass[sigconf,nonacm]{acmart}

\usepackage{amsmath,amsfonts}

\usepackage{array}

 \usepackage{tabularx}

\usepackage{textcomp}
\usepackage{stfloats}
\usepackage{verbatim}
\usepackage{graphicx}
\usepackage{color, colortbl}
\usepackage{enumitem}

\usepackage{amssymb}
\usepackage{framed}
\usepackage{multirow}
\usepackage{booktabs}
\usepackage{siunitx}
\usepackage{todonotes}
\usepackage{makecell}
\usepackage{longtable}
\usepackage{xspace}
\usepackage{pgfplots}
\usepackage{pgfplots, pgfplotstable}
\usepackage{bm}
\usepackage{adjustbox}
\usepackage{multicol}
\usepackage{boxedminipage}
\usepackage{multirow}
\usepackage{adjustbox}
\usepackage{multicol}
\usepackage{amsthm}

\newcommand{\spc}{\ensuremath{\text{SPC}}\xspace}
\usepackage{pgfplots}

\usepackage{xspace}

\usepackage{multirow}
\usepackage[mathscr]{eucal}
\usepackage{bm}
\definecolor{Gray}{gray}{0.9}
\usepackage{adjustbox}
\usepgfplotslibrary{groupplots}
\usepgfplotslibrary{colorbrewer}

\usepackage{subcaption}

\usepgfplotslibrary{colorbrewer}

\usepackage{colortbl}
\usepackage[skins]{tcolorbox}
\newtcolorbox{mybox}[2][]{%
  attach boxed title to top center
               = {yshift=-11pt},
  colframe     =black,
  colbacktitle = black,
  title        = #2,#1,
  enhanced,
}

\newcommand{\et}{\textit{et al.}\xspace}
\usepackage{amsfonts}

\newcommand{\ux}{\ensuremath{u_{\st 1}}\xspace}
\newcommand{\uy}{\ensuremath{u_{\st 2}}\xspace}
\usepackage{enumitem}
\usepackage{lipsum}
\usepackage{xcolor}
 \newcommand{\st}{\scriptscriptstyle}

\usepackage{adjustbox}

\usepackage{multicol}

\usepackage{bm}
\usepackage{multicol}

\newcommand{\cper}{\ensuremath{\bar{\mathtt{\pi}}}\xspace}

\newcommand{\ses}{\ensuremath\mathtt{SS}\xspace}

\newcommand{\keygen}{\ensuremath\mathtt{KGen}\xspace}
\newcommand{\enc}{\ensuremath\mathtt{Enc}\xspace}
\newcommand{\dec}{\ensuremath\mathtt{Dec}\xspace}

\newcommand{\hadd}{\ensuremath\stackrel{\st H}+\xspace}
\newcommand{\hmul}{\ensuremath\stackrel{\st H}{\times}\xspace}

\newcommand{\se}{\ensuremath{{S}}\xspace}
\newcommand{\re}{\ensuremath{{R}}\xspace}
\newcommand{\p}{\ensuremath{{P}}\xspace}

\newcommand{\tp}{\ensuremath{{T}}\xspace}

\newcommand{\parse}{\ensuremath{\mathtt{parse}}\xspace}

\newcommand{\adv}{\ensuremath{\mathcal{A}}\xspace}
\newcommand{\simm}{\ensuremath{\mathtt{Sim}}\xspace}
\newcommand{\view}{\ensuremath{\mathtt{View}}\xspace}
\newcommand{\secsh}{\ensuremath{\mathtt{SS}^{\st(t,n)}}\xspace}
\newcommand{\empt}{\ensuremath{\epsilon}\xspace}

\usepackage{mathtools}
\usepackage{adjustbox}

\usepackage{framed}
\usepackage{tikz}
\newcommand{\h}{\ensuremath\mathtt{H}\xspace}
\newcommand{\g}{\ensuremath\mathtt{G}\xspace}

\newcommand{\ot}{\ensuremath{\mathcal{OT}^{\st 2}_{\st 1}}\xspace}
\newcommand{\dqot}{\ensuremath{\mathcal{DQ\text{--}OT}^{\st 2}_{\st 1}}\xspace}
\newcommand{\dqothf}{\ensuremath{\mathcal{DQ^{\st MR}\text{--}OT}^{\st 2}_{\st 1}}\xspace}
\newcommand{\rdqothf}{\ensuremath{{\text{DQ}^{\st \text{MR}}\text{--}\text{OT}}}\xspace}
\newcommand{\duqot}{\ensuremath{\mathcal{DUQ\text{--}OT}^{\st 2}_{\st 1}}\xspace}
\newcommand{\duqothf}{\ensuremath{\mathcal{DUQ^{\st MR}\text{--}OT}^{\st 2}_{\st 1}}\xspace}
\newcommand{\rduqothf}{\ensuremath{{\text{DUQ}^{\st \text{MR}}\text{--}\text{OT}}}\xspace}
\newcommand{\onenot}{\ensuremath{\mathcal{OT}^{\st n}_{\st 1}}\xspace}

\usetikzlibrary{arrows,decorations.markings}

\theoremstyle{theorem}
\theoremstyle{definition}

\theoremstyle{theorem}

\begin{document}

\title[Delegated-Query OT]{Oblivis: A Framework for\\ Delegated and Efficient Oblivious Transfer}\titlenote{This paper combines and supersedes our two earlier ePrint manuscripts~\cite{AbadiD24,DesmedtA24}, and will appear at PETS 2026.}



\author{Aydin Abadi}
\orcid{0000-0002-1414-8351}
\affiliation{%
  \institution{Newcastle University}
  \city{} 
  \state{} 
  \country{} 
}
\email{aydin.abadi@ncl.ac.uk}

\author{Yvo Desmedt}
\orcid{0000-0002-6679-7484}
\affiliation{%
  \institution{The University of Texas at Dallas}
  \city{} 
  \state{} 
  \country{} 
}
\email{y.desmedt@cs.ucl.ac.uk}



\begin{abstract}
As database deployments shift toward cloud platforms and edge devices, thin clients need to securely retrieve sensitive records without leaking their query intent or metadata to the proxies that mediate access. Oblivious Transfer (OT) is a core tool for private retrieval, yet existing OTs assume direct client–database interaction and lack support for delegated querying or lightweight clients. 

We present Oblivis, a modular framework of new OT protocols that enable delegated, privacy-preserving query execution. Oblivis allows clients to retrieve database records without direct access, protects against leakage to both databases and proxies, and is designed with practical efficiency in mind. Its components include: (1) Delegated-Query OT, which permits secure outsourcing of query generation; (2) Multi-Receiver OT for merged, cloud-hosted databases; (3) a compiler producing constant-size responses suitable for thin clients; and (4) Supersonic OT, a proxy-based, information-theoretic, and highly efficient 1-out-of-2 OT. The protocols are formally defined and proven secure in the simulation-based paradigm, under non-colluding assumption. We implement and empirically evaluate Supersonic OT.
It achieves at least a $92\times$ speedup over a highly efficient 1-out-of-2 OT, and a $2.6\times$–$106\times$ speedup over a standard OT extension across 200–100,000 invocations. Our implementation further shows that Supersonic OT remains efficient even on constrained hardware, e.g., it completes an end-to-end transfer in 1.36~ms on a Raspberry Pi~4.

\end{abstract}



\maketitle





\section{Introduction}\label{sec:intro}



Modern data-driven services increasingly operate in cloud-hosted environments, where many clients are thin and lack substantial local resources. In these settings, clients often rely on intermediaries (both technical and organisational) to reach the underlying service or database. At the technical level, platforms such as Cloudflare~\cite{CloudflareHowItWorks2025} sit between users and origin servers, handling client traffic and therefore observing request structure and metadata; recent compromises of such proxies highlight the associated risks~\cite{CloudflareInvestigationOkta2022}. 
At the organisational level, financial advisors routinely issue queries to proprietary data providers on behalf of clients, creating opportunities for misuse or leakage of sensitive query intent, as shown in prior incidents~\cite{leigh2015hsbc,JPMorgan}. 

In these settings, a recurring challenge emerges: clients, often resource-constrained, cannot communicate directly with the database, yet must retrieve information without revealing their query intent or exposing metadata to the proxies that mediate access.

Oblivious Transfer (OT) is a foundational primitive for private data retrieval. Classical OT ensures that a client retrieves exactly one database item without exposing its index and without learning anything about the remaining items. Existing OT constructions fall short in several fundamental ways, which we analyse below. 

\subsection{Motivating Scenarios}\label{sec::motivation}
To illustrate these challenges, we highlight representative scenarios.

\vspace{-2mm}

\subsubsection{Proxy-Assisted Access in Financial Systems.} 
In financial settings, advisors often act on behalf of clients, querying proprietary databases (e.g., containing real estate market information, market trends, and capital flows) offered by third-party providers such as Bloomberg Terminal~\cite{Bloomberg}, CoStar~\cite{costar},  or Multiple Listing Service~\cite{mls},  without granting clients direct access.  
This proxy model introduces privacy risks, especially when insider threats arise, e.g., ``Swiss Leaks'' and JPMorgan Chase incidents~\cite{leigh2015hsbc,JPMorgan}.  
In such cases, neither the advisor nor the database should learn the content or intent of a client's query. Existing OTs do not account for this kind of secure delegation. Encrypting a classical OT query would hide the client's index from the proxy advisor, but it does not support delegation: the receiver must still compute the full OT query locally, leaving the proxy unable to assist with the query-generation task. Delegation is valuable in these settings because it allows clients to remain lightweight and improves scalability by shifting the expensive steps away from the client.


\vspace{-2mm}

\subsubsection{Merged Cloud-Hosted Databases.} Cloud providers often merge datasets belonging to different organisations into a single logical database to reduce operational cost and simplify indexing. A recent incident, for example, showed that data belonging to British Airways, Boots, and the BBC was stored within the same cloud~\cite{BBB-hack}. When staff from different organizations query only their own records, direct privacy-preserving retrieval reveals sensitive metadata about the merged dataset. A similar situation arises even within a single organisation (e.g., RMIT University) when the cloud provider merges datasets from multiple internal departments, to enable organisation-wide analytics~\cite{slalom2023rmit}.



\vspace{-2mm}

\subsubsection{Hidden Query Components.}
In domains such as finance and healthcare, clients cannot always be informed about the exact conditions used to compute their outputs. An investment recommendation may depend on a proprietary internal flag associated with an account~\cite{Arora24,starlit}; a patient may be shielded from certain diagnoses~\cite{exchange-of-health-info,Withholding-Information}; and some systems assign each client a private access policy~\cite{JoshiMJF17,CamenischDEN12}.  Secure retrieval must therefore succeed without revealing the underlying index to the receiver.

\vspace{-2mm}
\subsubsection{Performance Constraints in Lightweight Clients.} Many systems rely on lightweight or long-lived clients (e.g., mobile applications or autonomous devices) that cannot run heavy machinery, download/store large encrypted datasets. These clients constantly interact with cloud services, including in settings where long-term security is important or where hardness assumptions are undesirable, since quantum algorithms show that quantum computers can break certain hardness assumptions. Sectors such as financial services and healthcare may operate under such requirements.

\subsection{Limitations of Existing Techniques}

We now discuss key research gaps our work addresses.


\vspace{-1mm}
\subsubsection{Lack of Support for Proxy-Based Delegation.} 

 Existing OT schemes assume that the receiver can compute the query index locally and interact directly with the sender, which is incompatible with settings that rely on proxy-based query generation. This restriction is inherent to classical OTs, since the receiver's choice bit is bound to the protocol execution to prevent any untrusted intermediary from inferring the selection.


%


\vspace{-1.2mm}
\subsubsection{Absence of Multi-Receiver Privacy and Hidden-Query Support.}

OT variants designed for multi-user or access-controlled databases reveal non-trivial metadata in multi-tenant cloud environments (e.g., the total database size or receiver access patterns). These constructions further assume that the receiver knows the query index, which conflicts with hidden-query settings. 


\vspace{-1.2mm}
\subsubsection{Constant-Size Responses Without Large Local Storage.}

Protocols achieving constant-size responses usually require substantial receiver-side storage.  Existing OTs offer no generic way to obtain constant-size responses without imposing heavy local storage costs.

\vspace{-1.2mm}

\subsubsection{Information-Theoretic OT} Information-theoretic OT avoids
hardness assumptions, but existing constructions depend on multiple senders, noisy channels, or a fully trusted setup. These requirements limit their practicality and prevent deployment in lightweight or latency-sensitive environments.

\vspace{-2mm}
\subsection{Our Contributions}

We introduce Oblivis\footnote{Oblivis is derived from the Latin verb oblivisci (to forget).
We use it to reflect that any information seen by parties other than the receiver is uninformative and can be regarded as forgettable since it reveals nothing about the message or the query.}, a framework organized into two families. The first family is setting-oriented and provides three key OT primitives. The second one is performance-oriented. We formally define and prove the proposed solutions in the simulation-based paradigm. 
Our main contributions are as follows.

\begin{itemize}[leftmargin=4mm]
\item Setting-Oriented Contributions:

\begin{enumerate}[leftmargin=5.2mm]
\item \textbf{Delegated-Query OT (DQ-OT):} Our construction introduces a new technique in which a receiver's choice index is secret-shared, jointly processed, and ordered by two proxies whose sequential computations yield a query structurally identical to that of the classical Naor–Pinkas OT. This technique enables secure, fully delegated query generation without revealing the index to either proxy.

\item \textbf{Delegate-Unknown-Query OT (DUQ-OT):} 
An extension of DQ-OT that supports hidden query indices, where the receiver securely obtains the correct record without learning the index used to retrieve it. This enables access policies and private conditioning that conventional OTs cannot capture.

\item \textbf{Delegated-Query Multi-Receiver OT (\rdqothf):} 
An extension of DQ-OT that supports privately querying merged cloud datasets, where receivers learn nothing about the total number of records or unrelated fields, and the sender learns nothing about which receiver queried which record.  

\end{enumerate}
\item Performance-Oriented Contributions: 
\begin{enumerate}[leftmargin=5.2mm]
\item \textbf{A Generic Compiler for Constant-Size Responses}: It transforms any $1$-out-of-$n$ OT  into one with constant-size receiver size download complexity.
\item \textbf{Supersonic OT}: We introduce Supersonic OT, a highly efficient proxy-based, information-theoretic 1-out-of-2 OT that requires neither multiple senders, nor noisy channels, nor a trusted setup. The construction employs a controlled-swap mechanism, which, to our knowledge, has not been previously reported in the OT literature. Supersonic OT enables constant-size responses without relying on pre-stored database encryptions or heavy preprocessing. It achieves:
\begin{itemize}[leftmargin=4.1mm]
\item constant-size responses (i.e., receiver download: $O(1)$);
\item  0.53~ms per transfer on a modern laptop;
\item $92\times$ speedup over the efficient base OT of~\cite{ChouO15}, and a $2.6\times$–$106\times$ speedup over the OT extension of ~\cite{IshaiKNP03};  
\item end-to-end execution on a Raspberry Pi 4, completing each transfer in 1.36 ms, showing smooth scaling
(e.g., 3.26 ms at 1,000 invocations).
\end{itemize}

\end{enumerate}
\end{itemize}



\begin{table}[t]
\centering
\small
\setlength{\tabcolsep}{4pt}
\renewcommand{\arraystretch}{1.1}
\begin{tabularx}{\columnwidth}{l l X X}
\toprule
\textbf{OT variant} & \textbf{Party} & \textbf{Input} & \textbf{Output} \\
\midrule


\multirow{4}{*}{DQ--OT} 
& $\se$   & $(m_{\st 0}, m_{\st 1})$ & $\epsilon$ \\
& $\p_{\st 1}$ & $\epsilon$  & $\epsilon$ \\
& $\p_{\st 2}$ & $\epsilon$  & $\epsilon$ \\
& $\re$   & $s \in \{0,1\}$ & $m_{\st s}$ \\
\midrule

\multirow{5}{*}{DUQ--OT} 
& $\se$   & $(m_{\st 0}, m_{\st 1})$ & $\epsilon$ \\
& $\tp$   & $s \in \{0,1\}$ & $\epsilon$ \\
& $\p_{\st 1}$ & $\epsilon$  & $\epsilon$ \\
& $\p_{\st 2}$ & $\epsilon$  & $\epsilon$ \\
& $\re$   & $\epsilon$  & $m_{\st s}$ \\
\midrule

\multirow{4}{*}{\rdqothf} 
& $\se$   & $\vec{m}$ & $\epsilon$ \\
& $\p_{\st 1}$ & $v \in \{0, \ldots, z-1\}$ & $z$ \\
& $\p_{\st 2}$ & $\epsilon$ & $\epsilon$ \\
& $\re$ & $s \in \{0,1\}$ & $m_{\st s, v}$ \\
\midrule

\multirow{5}{*}{\rduqothf} 
& $\se$   & $\vec{m}$ & $\epsilon$ \\
& $\tp$   & $(v,s,z)$ & $\epsilon$ \\
& $\p_{\st 1}$ & $\epsilon$ & $z$ \\
& $\p_{\st 2}$ & $\epsilon$ & $\epsilon$ \\
& $\re$ & $\epsilon$ & $m_{\st s, v}$ \\
\midrule

\multirow{2}{*}{Compiler} 
& $\se$ & $(m_{\st 0}, \ldots, m_{\st n-1})$ & $\epsilon$ \\
& $\re$ & $s \in \{0, \ldots, n-1\}$ & $m_{\st s}$ \\
\midrule

\multirow{3}{*}{Supersonic OT} 
& $\se$ & $(m_{\st 0}, m_{\st 1})$ & $\epsilon$ \\
& $\p$ & $\epsilon$ & $\epsilon$ \\
& $\re$ & $s \in \{0,1\}$ & $m_{\st s}$ \\
\bottomrule
\end{tabularx}
\caption{Summary of parties, inputs, and outputs for the OT variants in \emph{Oblivis} (ideal-functionality view). Here, $\vec{m} = [(m_{\st 0, 0}, m_{\st 1, 0}), \ldots, (m_{\st 0, z-1}, m_{\st 1, z-1})]$, $\epsilon$ denotes an empty input/output. In multi-receiver variants, $v$ identifies the receiver (or dataset slice) among $z$ slices.}
\label{tab:ot-io-summary}
\end{table}

\subsection{Context and Applications}
Oblivis targets proxy-mediated access where intermediaries are unavoidable. Consider the following  use cases. (a) Financial advisory platform: A client wants to retrieve a proprietary report via a brokerage platform, which today learns the client's intent; DQ-OT enables mediation without revealing the requested record. (b) Multi-tenant SaaS behind a gateway: SaaS APIs often sit behind an edge gateway (e.g., Cloudflare) for routing/access control. Oblivis uses the gateway for delegation, such that the tenant sends a request, proxies generate OT queries and contact the backend, and only the intended encrypted response reaches the tenant. (c) Travel booking via a portal:  Corporate travel portals mediate access to proprietary offer databases. Oblivis retrieves the chosen offer without exposing the selection to the portal.

\section{Preliminaries}

\subsection{Notations and Assumptions}\label{sec::notations}
By $\empt$ we mean an empty string. When $y$ represents a single value, $|y|$ refers to the bit length of $y$; when $y$ is a tuple, $|y|$ denotes the number of elements contained within $y$. We denote a sender by $\se$ and a receiver by $\re$. We assume parties interact with each other through a secure channel. 
We define a parse function as $\parse(\lambda, y)\rightarrow (\ux, \uy)$, which takes as input a value $\lambda$ and a value $y$ of length at least  $\lambda$-bit. It parses $y$ into two values  $\ux$ and $\uy$ and returns $(\ux, \uy)$ where the bit length of $\ux$ is $|y|-\lambda$ and the bit length of $\uy$ is  $\lambda$. 
Also, $U$ denotes a universe of messages $m_{\st 1},\ldots, m_{\st t}$. We define $\sigma$ as the maximum size of messages in $U$, i.e., $\sigma=Max(|m_{\st 1}|,\ldots, |m_{\st t}|)$. 
We use two hash functions $\h:\{0, 1\}^{\st *}\rightarrow \{0, 1\}^{\sigma}$ and $\g:\{0, 1\}^{\st *}\rightarrow \{0, 1\}^{\st\sigma+\lambda}$ modelled as random oracles \cite{Canetti97}. 

Table~\ref{tab:ot-io-summary} summarizes our protocols' input and output values for each party involved.  
In this paper, we use the simulation-based paradigm of secure multi-party computation \cite{DBLP:books/cu/Goldreich2004} to define and prove the protocols. We consider semi-honest adversaries that do not collude with each other.  
Appendix \ref{sec::sec-model} (in the supplemental material file) restates the formal definition in this model.


  \vspace{-2mm}
  \subsection{Random Permutation}
A random permutation $\pi(e_{\st 0} ,..., e_{\st n})\rightarrow(e'_{\st 0},..., e'_{\st n})$ is a probabilistic function that takes a set $A=\{e_{\st 0}, ..., e_{\st n}\}$ and returns the same set of elements in a permuted order $B=\{e'_{\st 0},..., e'_{\st n}\}$. The security of $\pi$ requires that given a set $B$, the probability that one can find the original index of an element $e'_{\st i}\in B$ is $\frac{1}{n}$.  In practice, the Fisher-Yates shuffle algorithm \cite{Knuth81} can permute a set of $n$ elements in time $O(n)$. We will use $\pi$ in the protocols presented in Figures \ref{fig::DQ-OT-with-unknown-query} and \ref{fig::DQ-OT-with-unknown-query-and-HF-first-eight-phases}.


\subsection{Controlled Swap}\label{Customised-Random-Swap}
The idea of controlled swap was introduced by Fredkin and Toffoli~\cite{FredkinT02a}. It can be defined as function $\cper$ which takes two inputs: a binary value $s$ and a pair $(c_{\st 0}, c_{\st 1})$. When $s=0$, it returns the input pair $(c_{\st 0}, c_{\st 1})$, i.e., it does not swap the elements. When $s=1$, it returns $(c_{\st 1}, c_{\st 0})$, effectively swapping the elements. 
It is clear that if $s$ is uniformly chosen at random, then $\cper$ represents a random permutation, implying that the probability of swapping or not swapping is $\frac{1}{2}$. We will use $\bar{\pi}$ in the Supersonic OT, presented in Figure~\ref{fig::Ultrasonic-OT}.

\subsection{The Original OT of Naor and Pinkas}\label{sec::OT-of-Naor-Pinkas}


As the DQ-OT protocol in Section~\ref{sec::DQ-OT} builds on the OT by Naor and Pinkas~\cite[pp.~450–451]{Efficient-OT-Naor}, we restate their construction in Figure~\ref{fig::Noar-OT}.


\begin{figure}[!h]
\setlength{\fboxsep}{.9pt}
\begin{center}
    \begin{tcolorbox}[enhanced,right=1mm, 
    drop fuzzy shadow southwest,
    colframe=black,colback=white]
\begin{enumerate}[leftmargin=2.2mm]

\item \underline{\textit{$\se$-side Initialization:}} 
$\mathtt{\se.Init}(1^{\st\lambda})\rightarrow pk$

\begin{enumerate}

\item  choose a random large prime number $p$ (where $\log_{\st 2}p=\lambda$), a random element
$C \stackrel{\st \$}\leftarrow \mathbb{Z}_{\st p}$ and random generator $g$. It is only important that the receiver \re will not know $a$, which is the discrete logarithm of $C$ to the base $g$, i.e., $C=g^{\st a}$.
\item publish $pk=(C, g, p)$. 

\end{enumerate}

\item \underline{\textit{$\re$-side Query Generation:}} 
$\mathtt{\re.GenQuery}(pk, s)\rightarrow (q, sp)$

\begin{enumerate}

\item pick a random value $r \stackrel{\st \$}\leftarrow \mathbb{Z}_p\setminus\{0\}$ and sets $sp=r$. 
\item set $\beta_{\st s}=g^{\st r}$ and $\beta_{\st 1-s}=\frac{C}{\beta_{\st s}}$.

\item send $q=\beta_{\st 0}$ to \se and locally stores $sp$.

\end{enumerate}

\item \underline{\textit{$\se$-side Response Generation:}} 
$\mathtt{\se.GenRes}(m_{\st 0}, m_{\st 1}, pk, q)\rightarrow res$

\begin{enumerate}

\item compute $\beta_{\st 1}=\frac{C}{\beta_{\st 0}}$.

\item choose two random values, $y_{\st 0}, y_{\st 1}\stackrel{\st \$}\leftarrow \mathbb{Z}_p$. 

\item encrypt the elements of the pair $(m_{\st 0}, m_{\st 1})$ as follows: 
$$e_{\st 0}:=(e_{\st 0,0}, e_{\st 0,1})=(g^{\st y_0}, \h(\beta_{\st 0}^{\st y_0}) \oplus m_{\st 0})$$
$$e_{\st 1}:=(e_{\st 1,0}, e_{\st 1,1})=(g^{\st y_1}, \h(\beta_{\st 1}^{\st y_1}) \oplus m_{\st 1})$$

\item send $res=(e_{\st 0}, e_{\st 1})$ to \re.
\end{enumerate}

\item\underline{\textit{\re-side Message Extraction:}} 
$\mathtt{\re.Retrieve}(res, sp, pk, s)\rightarrow m_{\st s}$

\begin{itemize}
\item retrieve the related message $m_{\st s}$ by computing 
$$m_{\st s}=\h((e_{\st s,0})^{\st r})\oplus e_{\st s, 1}$$

\end{itemize}

\end{enumerate}
\end{tcolorbox}
\end{center}
\vspace{-3mm}
\caption{Original OT proposed by  Naor and Pinkas \cite[pp. 450, 451]{Efficient-OT-Naor}. In this protocol, the input of \re is a private binary index $s$ and the input of \se is a pair of private messages $(m_{\st 0}, m_{\st 1})$.} 
\vspace{-4mm}
\label{fig::Noar-OT}
\end{figure}

\vspace{-2mm}
\subsection{Diffie-Hellman Assumption}\label{sec::dh-assumption}
Let $G$ be a group-generator scheme, which on input $1^{\st\lambda}$ outputs $(\mathbb{G}, p, g)$ where $\mathbb{G}$ is the description of a group, $p$ is the order of the group which is always a prime number, $\log_{\st 2}(p)=\lambda$ is a security parameter and $g$ is a generator of the group. In this paper, $g$ and $p$ can be selected by sender \se (in the context of OT). 

\subsubsection*{Computational Diffie-Hellman (CDH) Assumption} 
We say that $G$ is hard under the CDH assumption if for any probabilistic polynomial time (PPT) adversary $\mathcal{A}$, given $(g^{\st a_1}, g^{\st a_2})$, it has only negligible probability to correctly compute $g^{\st a_1\cdot a_2}$. More formally, it holds that $Pr[\mathcal{A}(\mathbb{G}, p, g, g^{\st a_1}, g^{\st a_2})\rightarrow g^{\st a_1\cdot a_2}]\leq \mu(\lambda)$, where $(\mathbb{G}, p, g) \stackrel{\st\$}\leftarrow G(1^{\st\lambda})$, $a_{\st 1}, a_{\st 2} \stackrel{\st\$}\leftarrow \mathbb{Z}_{\st p}$, and $\mu$ is a negligible function~\cite{DiffieH76}. 
%

\subsubsection*{Convention} Exponents (e.g., $x, r_{\st 1}, r_{\st 2}, a$) are elements of $\mathbb{Z}_{\st p}$ and all arithmetic on them  is modulo $p$; we omit $\bmod\ p$ for readability. Group expressions such as $g^{\st r}$ remain elements of the group $\mathbb{G}$.

\subsection{Secret Sharing}\label{sec::secret-haring}

A (threshold) secret sharing $\mathtt{SS}^{\st(t,n)}$  scheme is a cryptographic protocol that enables a dealer to distribute a string $s$, known as the secret, among $n$ parties in a way that the secret $s$ can be recovered when at least a predefined number of shares, say $t$, are combined. If the number of shares in any subset is less than $t$, the secret remains unrecoverable, and the shares divulge no information about $s$. This type of scheme is referred to as $(n, t)$-secret sharing or \secsh for brevity. 
In the case where $t=n$, there exists a highly efficient XOR-based secret sharing \cite{blakley1980one}. In this case, to share the secret $s$, the dealer first picks $n-1$ random bit strings $r_{\st 1}, ..., r_{\st n-1}$ of the same length as the secret. Then, it computes $r_{\st n} = r_{\st 1} \oplus, ..., \oplus  r_{\st n} \oplus s$. It considers each $r_{\st i}\in\{r_{\st 1},..,r_{\st n}\}$ as a share of the secret. To reconstruct the secret, one can easily compute $r_{\st 1}\oplus,..., \oplus r_{\st n}$. Any subset of fewer than $n$ shares reveals no information about the secret. We will use this scheme in this paper. A secret sharing scheme involves two main algorithms; namely, $\ses(1^{\st \lambda}, s, n, t)\rightarrow (r_{\st 1}, ..., r_{\st n})$: to share a secret and $\mathtt{RE}(r_{\st 1}, ..., r_{\st t}, n, t)\rightarrow s$ to reconstruct the secret. 

\vspace{-2mm}
\subsection{Additive Homomorphic Encryption}\label{sec::AHE}

Additive homomorphic encryption involves three algorithms: (1) key generation: $\keygen(1^{\st\lambda})\rightarrow (sk, pk)$, which takes a security parameter as input and outputs a secret and public keys pair, (2) encryption: $\enc(pk, m)\rightarrow c$, that takes public key $pk$ and a plaintext message $m$ as input and returns a ciphertext $c$, and (3) decryption: $\dec(sk, c)\rightarrow m$, which takes secret key $sk$ and ciphertext $c$ as input and returns plaintext message $m$.  It has the following properties: 
 \vspace{-1mm}
 \begin{itemize}
 \item [$\bullet$]  Given two ciphertexts $\enc(pk, m_{\st 1})$ and $\enc(pk, m_{\st 2})$, one can compute the encryption of the sum of related plaintexts: 
  $\dec(sk,\mathtt{Enc}(pk, m_{\st 1})\hadd \enc(pk, m_{\st 2}))= m_{\st 1}+m_{\st 2}$, where $\hadd$ denotes homomorphic addition.

  
 \item [$\bullet$] Given a ciphertext $\mathtt{Enc}(pk, m)$ and a plaintext message $c$, one can compute the encryption of the product of related plaintexts: 
 %
 %
  $\dec(sk, \enc(pk, m)\hmul c) = m\cdot c$, where $\hmul$ denotes homomorphic multiplication.
 \end{itemize}
We require that the encryption scheme satisfy indistinguishability against chosen-plaintext attacks (IND-CPA).  We refer readers to \cite{KatzLindell2014} for a formal definition. One such scheme that meets the above features is the Paillier public key cryptosystem~\cite{Paillier99}. We will use an additive homomorphic encryption only in Sections~\ref{sec::Delegated-Unknown-Query-OT-HF} and \ref{sec::the-compiler}.


\vspace{-2mm}
\section{Related Work}\label{sec::related-work}


The traditional $1$-out-of-$2$ OT (\ot) is a protocol that involves a sender \se and a receiver \re  \cite{Rabin-OT,EvenGL85}.  The sender \se has a pair of messages $(m_{\st 0}, m_{\st 1})$ and \re has an index $s$. Informally, \ot allows \re to obtain $m_{\st s}$, without revealing anything about $s$ to \se, and without allowing \re to learn anything about  $m_{\st 1-s}$. The \ot functionality is defined as $\mathcal{F}_{\st\ot}:((m_{\st 0}, m_{\st 1}), s) \rightarrow (\empt, m_{\st s})$. 

There exist many variants of OT. 
$1$-out-of-$n$ OTs~\cite{NaorP99,Tzeng02,LiuH19} allow \re to pick one entry out of $n$ entries held by \se, $k$-out-of-$n$ OTs~\cite{ChuT05,JareckiL09,ChenCH10} enable \re to pick $k$ entries out of $n$ entries held by \se, where $k\geq 1$. OT extensions~\cite{IshaiKNP03,Henecka013,Nielsen07,AsharovL0Z13} support efficient executions of OT (that mainly rely on symmetric-key operations) in the case where OT must be invoked many times. Distributed OTs~\cite{NaorP00,CorniauxG13,ZhaoSJMZX20} allow the database to be distributed among $m$ servers/senders. 
In the remainder of this section, we discuss several variants of OT  that have extended and enhanced the original OT~\cite{Rabin-OT}.

\subsection{Distributed OT} 

 Naor and Pinkas \cite{NaorP00} proposed several distributed OTs where the role of sender \se (in the original OT) is divided between several servers.  In these schemes, a receiver must contact a threshold of the servers to run the OT. 
 
 The proposed protocols are in the semi-honest model. They use symmetric-key primitives and do not involve any modular exponentiation that can lead to efficient implementations.  These protocols are based on various variants of polynomials (e.g., sparse and bivariate), polynomial evaluation, and pseudorandom functions. In these distributed OTs, the security against the servers holds as long as fewer than a predefined number of these servers collude. 
 Later, various distributed OTs have been proposed\footnote{Distributed OT has also been called proxy OT in \cite{YaoF06}.}. For instance, Corniaux and Ghodosi \cite{CorniauxG13} proposed a verifiable $1$-out-of-$n$ distributed OT that considers the case where a threshold of the servers is potentially active adversaries. The scheme is based on a sparse $n$-variate polynomial, verifiable secret sharing, and error-correcting codes. 
 
 Moreover, Zhao \et  \cite{ZhaoSJMZX20} has proposed a distributed version of OT extension that aims to preserve the efficiency of OT extension while delegating the role of \se to multiple servers, a threshold of which can be potentially semi-honest. The scheme is based on a hash function and an oblivious pseudorandom function.  
 
However, there exists no OT that supports secure delegation of the query computation to third-party servers. 


\subsection{Multi-Receiver OT}\label{sec::multi-rec-ot}

Camenisch \et~\cite{CamenischDN09} proposed a protocol for ``OT with access control''. It involves a set of receivers and a sender that maintains records of the receivers. It offers a set of interesting features; namely, (i) only authorized receivers can access certain records; (ii) the sender does not learn which record a receiver accesses, and (iii) the sender does not learn which roles (or security clearance) the receiver has when it accesses the records. 
In this scheme, during the setup, the sender encrypts all records (along with their field elements) and publishes the encrypted database for the receivers to download. Subsequently, researchers proposed various variants of OT with access control, as seen in \cite{CamenischDNZ11,CamenischDEN12,CamenischDN10}. 

Nevertheless, in all the aforementioned schemes, the size of the entire database is revealed to the receivers.


\subsection{OT with Constant Response Size}

Researchers have proposed several OTs~\cite{CamenischNS07,GreenH08,ZhangLWR13}  that enable a receiver to obtain a constant-size response to its query. To achieve this level of communication efficiency, these OTs require the receiver to locally store the encryption of the entire database in the initialization phase. During the transfer phase, the sender assists the receiver with locally decrypting the message that the receiver is interested in. 
The main limitation of these protocols is that a thin client with limited available storage space cannot use them, as it cannot locally store the encryption of the entire database.


\subsection{Information-Theoretic  OT}\label{sec::uncon-sec-OT} There have been efforts to design (both-sided) information-theoretic OTs. Some schemes~\cite{NaorP00,BlundoDSS07,CorniauxG13} rely on multiple servers/senders that maintain an identical copy of the database.  Other ones~\cite{CrepeauK88,CrepeauMW04,IshaiKOPSW11} rely on a specific network structure, i.e., a noisy channel, to achieve information-theoretic OT.   There is also a scheme~\cite{rivest1999unconditionally} that achieves information-theoretic OT using a fully trusted initializer.  
Hence, there exists no (efficient) information-theoretic OT that does not use noisy channels, multi-server, and a fully trusted initializer. 
%
%

\section{Delegated-Query OT}\label{sec:protocol}
In this section, we present the notion of Delegated-Query $1$-out-of-$2$ OT (\dqot) and a protocol that realizes it. 
 \dqot involves sender \se, receiver \re, and two proxy servers $\p_{\st 1}$ and $\p_{\st 2}$ that assist \re to compute the query.  
 \dqot enables \re to delegate (i) the computation of query and (ii) the interaction with \se to $\p_{\st 1}$ and $\p_{\st 2}$, who jointly compute \re's query and send it to \se.  
\dqot (in addition to offering the basic security of OT) ensures that \re's privacy is preserved from  $\p_{\st 1}$ and $\p_{\st 2}$, in the sense that  $\p_{\st 1}$ and $\p_{\st 2}$ do not learn anything about the actual index (i.e., $s\in\{0,1\}$) that \re is interested in, if they do not collude with each other.

\subsection{Functionality Definition}

Informally, the functionality that \dqot computes takes as input (i) a pair of messages $(m_{\st 0}, m_{\st 1})$ from \se, (ii) empty string  \empt from $\p_{\st 1}$, (iii)  empty string  \empt from $\p_{\st 2}$, and (iv) the index  $s$ (where $s\in \{0, 1\}$) from \re. It outputs an empty string $\empt$ to  \se, $\p_{\st 1}$, and $\p_{\st 2}$, and outputs the message with index $s$, i.e., $m_{\st s}$, to \re.  Formally, we define the functionality as: 
 $\mathcal{F}_{\scriptscriptstyle\dqot}:\big((m_{\st 0}, m_{\st 1}), \empt, \empt, s\big) \rightarrow (\empt, \empt, \empt, m_{\st s})$.

 \subsection{Security Definition}\label{sec::sec-def}
 
 Next, we present a formal definition of \dqot.

\begin{definition}[\dqot]\label{def::DQ-OT-sec-def} Let $\mathcal{F}_{\scriptscriptstyle\dqot}$ be the delegated-query OT functionality defined above. Protocol $\Gamma$ realizes $\mathcal{F}_{\scriptscriptstyle\dqot}$ in the presence of passive adversaries, if for  every non-uniform PPT adversary \adv in the real model, there exists a non-uniform PPT adversary (or simulator) \simm  in
the ideal model, such that:
%
\begin{equation}\label{equ::sender-sim-}
\begin{split}
\Big\{\simm_{\st\se}\big((m_{\st 0}, m_{\st 1}), \empt\big)\Big\}_{\st m_{\st 0}, m_{\st 1}, s}\stackrel{c}{\equiv} \Big\{\view_{\st\se}^{\st \Gamma}\big((m_{\st 0}, m_{\st 1}), \empt,  \empt, s\big) \Big\}_{ m_{\st 0}, m_{\st 1}, s}
\end{split}
\end{equation}
\begin{equation}\label{equ::server-sim-}
\begin{split}
\Big\{\simm_{\st\p_i}(\empt, \empt)\Big\}_{\st m_{\st 0}, m_{\st 1}, s}\stackrel{c}{\equiv}     \Big\{\view_{\st\p_i}^{\st \Gamma}\big((m_{\st 0}, m_{\st 1}), \empt,  \empt, s\big) \Big\}_{\st\st m_{\st 0}, m_{\st 1}, s}
\end{split}
\end{equation}
\begin{equation}\label{equ::reciever-sim-}
\begin{split}
&\Big\{\simm_{\st\re}\Big(s, \mathcal{F}_{\scriptscriptstyle\dqot}\big((m_{\st 0}, m_{\st 1}), \empt,  \empt, s\big)\Big)\Big\}_{\st m_{\st 0}, m_{\st 1}, s}\stackrel{c}{\equiv}\\  &\Big\{\view_{\st\re}^{\st \Gamma}\big((m_{\st 0}, m_{\st 1}), \empt,  \empt, s\big) \Big\}_{\st m_{\st 0}, m_{\st 1}, s}
\end{split}
\end{equation}
for all $i$,  $i\in \{1,2\}$.
\end{definition}
%
%

%
A \dqot scheme meets \textit{efficiency} and \textit{sender-push communication} (\spc). 
Efficiency states that the query generation of the receiver is faster compared to traditional (non-delegated) OT. 
\spc states that the sender sends responses to the receiver without requiring the receiver to directly initiate a query to the sender. 



\begin{definition}[Efficiency]\label{def::efficiency}
A \dqot scheme is efficient if the running time of the receiver-side request (or query) generation algorithm, denoted as $\mathtt{Request(1^{\st \lambda}, s, pk)}$, satisfies two conditions: 

\begin{enumerate}[label=$\bullet$,leftmargin=4.7mm]
    \item The running time is upper-bounded by $poly(|{m}|)$, where $poly$ is a fixed polynomial, $m$ is a tuple of messages that the sender holds, and $|{m}|$ represents the number of elements/messages in tuple $m$.

    \item The running time is asymptotically constant with respect to the security parameter $\lambda$, i.e., it is $O(1)$.  
\end{enumerate}
\end{definition}


\begin{definition}[Sender-push communication]\label{def::Server-push-comm} Let 
(a) $\mathtt{Action}_{\st \re}(t)$ represents the set of actions available to \re at time $t$ (these actions may include sending requests, receiving messages, or any other interactions \re can perform within the scheme's execution), 
(b) $\mathtt{Action}_{\st \se}(t)$ be the set of actions available to \se at time $t$,  
(c) $\mathtt{SendRequest}(\re,\se)$ be the action of \re sending a request to \se, and 
(d) $\mathtt{SendMessage}(\se,\re)$ represents the action of \se sending a message to \re. Then,  a \dqot scheme supports sender-push communication if it meets the following conditions:

\begin{enumerate}[label=$\bullet$,leftmargin=4.7mm]

\item \textit{Receiver-side restricted interaction}: 
For all $t$ in the communication timeline (i.e., within the scheme's execution), the set of actions $\mathtt{Action}_{\st \re}(t)$ available to the receiver \re is restricted such that it does not include direct requests to the sender \se. Formally,
$$\forall t: \mathtt{Action}_{\st \re}(t) \cap \{\mathtt{SendRequest}(\re,\se)\}=\emptyset$$

\item \textit{Sender-side non-restricted interaction}: 
For all $t$ in the communication timeline, the sender \se can push messages to the receiver \re without receiving an explicit request directly from \re. Formally, 
$$\forall t:
\mathtt{Action}_{\st \se}(t) \subseteq \{\mathtt{SendMessage}(\se,\re)\}
$$

\end{enumerate}

\end{definition}

Note that the above two definitions are valid for the variants of \dqot; namely, \duqot, \dqothf, and  \duqothf, with a minor difference being that the receiver-side request generation algorithm in these three variants is denoted as $\mathtt{\re.Request}$.

%

\hspace{-4mm}

\subsection{Protocol}\label{sec::DQ-OT}
Now, we present an efficient 1-out-of-2 OT protocol, called DQ-OT, that realizes \dqot. 
%
We build DQ-OT upon the \ot proposed by  Naor and Pinkas, presented in Figure~\ref{fig::Noar-OT}. Our motivation for this choice is primarily didactic. 

The DQ-OT relies on a key observation (which to our knowledge has not appeared in prior OTs): if the receiver secret-shares its index (see Figure~\ref{fig::Noar-OT}) into two random shares, and the proxies sequentially compute and order their partial queries according to theses shares, the final query pair sent to the sender has an identical structure to the receiver's query in the Naor–Pinkas OT, i.e., as if the receiver had generated both queries $(\beta_{\st 0}, \beta_{\st 1})$ itself.  
Below, we explain how  DQ-OT operates. 

Initially, \re splits its desired index into two binary shares,  $(s_{\st 1}, s_{\st 2})$. Then, it picks two random values, $(r_{\st 1}, r_{\st 2})$, and sends each pair $(s_{\st i}, r_{\st i})$ to each $\p_{\st i}$. 
To compute a partial query, $\p_{\st 2}$ treats $s_{\st 2}$ as the main index that \re is interested in and computes a partial query, $\delta_{\st s_2}=g^{\st r_2}$. Also, $\p_{\st 2}$  generates another query, $\delta_{\st 1-s_2}=\frac{C}{g^{\st r_2}}$, where $C$ is a random public parameter (as defined in \cite{Efficient-OT-Naor}). $\p_{\st 2}$ sorts the two queries in ascending order based on the value of $s_{\st 2}$ and sends the resulting $(\delta_{\st 0}, \delta_{\st 1})$ to $\p_{\st 1}$.

To compute its queries, $\p_{\st 1}$ treats $\delta_{\st 0}$ as the main index (that \re is interested) and computes $\beta_{\st s_1}=\delta_{\st 0}\cdot g^{\st r_1}$. Additionally, it generates another query $\beta_{\st 1-s_1}=\frac{\delta_{\st 1}}{g^{\st r_1}}$. Subsequently, $\p_{\st 1}$ sorts the two queries in ascending order based on the value of ${s_{\st 1}}$ and sends the resulting $(\beta_{\st 0}, \beta_{\st 1})$ to $\se$. 
Given the queries, \se computes the response in the same manner it does in the original OT in \cite{Efficient-OT-Naor} and sends the result to \re, who extracts from it the message that it asked for, with the help of $s_i$ and $r_i$ values.  
The detailed  DQ-OT  is presented in Figure~\ref{fig::DQ-OT}.

\begin{figure}[!htbp]
\setlength{\fboxsep}{.9pt}
\begin{center}
    \begin{tcolorbox}[enhanced,right=1mm,  
    drop fuzzy shadow southwest,
    colframe=black,colback=white]
\begin{enumerate}[leftmargin=2mm]

\item \underline{\textit{$\se$-side Initialization:}} 
$\mathtt{Init}(1^{\st \lambda})\rightarrow pk$
\begin{enumerate}

\item choose a sufficiently large prime number $p$.

\item select a random element $C\hspace{-.5mm} \stackrel{\st \$}\leftarrow\hspace{-.5mm} \mathbb{Z}_p$ and generator $g$.
\item publish $pk=(C, p, g)$. 

\end{enumerate}

\item \underline{\textit{\re-side Delegation:}}
$\mathtt{Request}(1^{\st \lambda}, s, pk)\hspace{-.6mm}\rightarrow\hspace{-.6mm} req=(req_{\st 1}, req_{\st 2})$

\begin{enumerate}

\item split  the private index $s$ into two shares $(s_{\st 1}, s_{\st 2})$ by calling  $\ses(1^{\st \lambda}, s, 2, 2)\rightarrow (s_{\st 1}, s_{\st 2})$.

\item pick two uniformly random values: $r_{\st 1}, r_{\st 2} \stackrel{\st\$}\leftarrow\mathbb{Z}_{\st p}$.

\item send $req_{\st 1}=(s_{\st 1}, r_{\st 1})$ to $\p_{\st 1}$ and $req_{\st 2}=(s_{\st 2}, r_{\st 2})$ to $\p_{\st 2}$.

\end{enumerate}

\item \underline{\textit{$\p_{\st 2}$-side Query Generation:}}
$\mathtt{\p_{\st 2}.GenQuery}(req_{\st 2}, ,pk)\rightarrow q_{\st 2}$

\begin{enumerate}

\item compute a pair of partial queries:
%
%
  $$\delta_{\st s_2}= g^{\st r_2},\ \  \delta_{\st 1-s_{\st 2}} = \frac{C}{g^{\st r_2}}$$
  
\item send $q_{\st 2}=(\delta_{\st 0}, \delta_{\st 1})$ to  $\p_{\st 1}$. 

\end{enumerate}

\item\underline{\textit{$\p_{\st 1}$-side Query Generation:}}
$\mathtt{\p_{\st 1}.GenQuery}(req_{\st 1}, q_{\st 2},pk)\rightarrow q_{\st 1}$

\begin{enumerate}

\item compute a pair of final queries as: 
%
%
$$\beta_{s_{\st 1}}=\delta_{\st 0}\cdot g^{\st r_{\st 1}},\ \ \beta_{\st 1-s_1}=\frac{\delta_{\st 1}} {g^{\st r_{\st 1}}}$$

\item send $q_{\st 1}=(\beta_{\st 0}, \beta_{\st 1})$ to  $\se$.

\end{enumerate}

\item\underline{\textit{\se-side Response Generation:}}  
$\mathtt{GenRes}(m_{\st 0}, m_{\st 1}, pk, q_{\st 1})\rightarrow res$
\begin{enumerate}

\item abort if  $C \neq \beta_{\st 0}\cdot \beta_{\st 1}$.
\item pick two uniformly random values: 
$y_{\st 0}, y_{\st 1}  \stackrel{\st\$}\leftarrow\mathbb{Z}_{\st p}$.
\item compute a response pair $(e_{\st 0}, e_{\st 1})$ as follows:
 $$e_{\st 0} := (e_{\st 0,0}, e_{\st 0,1}) = (g^{\st y_0}, \h(\beta_{\st 0}^{\st y_0}) \oplus m_{\st 0})$$
$$e_{\st 1} := (e_{\st 1,0}, e_{\st 1,1}) = (g^{\st y_1}, \h(\beta_{\st 1}^{\st y_1}) \oplus m_{\st 1})$$
\item send $res=(e_{\st 0}, e_{\st 1})$ to \re. 

\end{enumerate}

\item\underline{\textit{\re-side Message Extraction:}} 
$\mathtt{Retrieve}(res, req, pk)\hspace{-.6mm}\rightarrow\hspace{-.8mm} m_{\st s}$
\begin{enumerate}

\item set $x=r_{\st 2}+r_{\st 1}\cdot(-1)^{\st s_{\st 2}}$

\item retrieve the related message: 

$m_{\st s}=\h((e_{\st s, 0})^{\st x})\oplus e_{\st s, 1}$

\end{enumerate}

\end{enumerate}
\end{tcolorbox}
\end{center}
\vspace{-2mm}
\caption{DQ-OT: Our $1$-out-of-$2$ OT that supports query delegation. The input of \re is a private binary index $s$, and the input of \se is a pair of messages $(m_{\st 0}, m_{\st 1})$. Note, $\ses$ is the share-generation algorithm, $\h$ is a hash function, and $\$$ denotes sampling a value uniformly at random.} 
\vspace{-3mm}
\label{fig::DQ-OT}
\end{figure}

\begin{theorem}\label{theo::DQ-OT-sec}
Let $\mathcal{F}_{\scriptscriptstyle\dqot}$ be the functionality defined in Section~\ref{sec::sec-def}. If  
Discrete Logarithm (DL), Computational Diffie-Hellman (CDH), and Random Oracle (RO) assumptions hold, then DQ-OT (presented in Figure \ref{fig::DQ-OT}) securely computes $\mathcal{F}_{\scriptscriptstyle\dqot}$ in the presence of semi-honest adversaries,
w.r.t. Definition \ref{def::DQ-OT-sec-def}. 

\end{theorem}

%
\begin{proof}[Proof sketch]
We argue security via simulation for each party in~Definition \ref{def::DQ-OT-sec-def}. 
Each proxy receives only a one-time-pad share $s_{\st i}$ of $s$ (uniform on
$\{0,1\}$) and random $r_{\st i}$, and its remaining view consists of group elements whose
exponents are hidden under DL, so its view can be simulated without knowledge of $s$. 
Moreover, \se sees only a random $\beta_{\st 0}$ and $\beta_1=C/\beta_{\st 0}$, so its view is independent of $s$ and is simulatable.   For \re, the ciphertext component for $m_{\st 1-s}$ is a one-time pad under key $\h(\beta_{\st 1-s}^{\st y_{\st 1-s}})$, and computing $\beta_{\st 1-s}^{\st y_{\st 1-s}}$ from $(\beta_{\st 1-s}, g^{\st y_{\st 1-s}})$ breaks CDH; as $\h$ is modelled as a RO, this pad is indistinguishable from random, so the transcript is simulatable given only $m_{\st s}$.
\end{proof}
Appendices \ref{sec::DQ-OT-proof} and \ref{sec::DQ-OT-proof-of-correctness} present the full proof of Theorem \ref{theo::DQ-OT-sec} and the proof of DQ-OT's correctness, respectively. 

Note that in the protocol, we send $\beta_{\st 1}$ to maintain uniformity in the query format across our OT protocols. Equivalently, $\p_{\st 1}$ can send only $\beta_{\st 0}$ and $S$ can reconstruct $\beta_{\st 1}$.

\noindent\textbf{Naïve Approach.} One might instead let the receiver generate a standard OT query, encrypt it, and send the ciphertext through a proxy. Under semantic security, the proxy learns nothing from such a ciphertext, so this naïve design offers no security disadvantage relative to DQ-OT. However, it does not achieve delegation: the receiver must construct the OT query locally, and the proxy performs no part of the query-generation task. In contrast, DQ-OT (and its variants) shifts all heavy operations in the query-generation phase to the proxies and leaves the receiver with only basic operations.



\section{Delegated-Unknown-Query OT}\label{sec::DUQ-OT}

%

In certain cases, the receiver itself may not know the value of query $s$. Instead, the query is issued by a third-party query issuer (\tp).  
In this section, we present a new variant of \dqot, called Delegated-Unknown-Query 1-out-of-2 OT (\duqot). It enables \tp to issue the query while (a) preserving the security of  \dqot and (b) preserving the privacy of query $s$ from \re.

%
%
%
%
\subsection{Security Definition}\label{sec::DUQ-OT-definition}

The functionality that \duqot computes takes as input (a) a pair of messages $(m_{\st 0}, m_{\st 1})$ from \se, (b) empty strings \empt from $\p_{\st 1}$, (c)  \empt from $\p_{\st 2}$, (d)   \empt from $\re$, and (e) the index $s$ (where $s\in \{0, 1\}$) from \tp. It outputs an empty string $\empt$ to  \se, \tp, $\p_{\st 1}$, and $\p_{\st 2}$, and outputs the message with index $s$, i.e., $m_{\st s}$, to \re. More formally, we define the functionality as: 
 $\mathcal{F}_{\scriptscriptstyle\duqot}:\big((m_{\st 0}, m_{\st 1}), \empt, \empt, \empt, s\big) \rightarrow (\empt, \empt, \empt, m_{\st s}, \empt)$. 
Next, we present a formal definition of \duqot.

\begin{definition}[\duqot]\label{def::DUQ-OT-sec-def} Let $\mathcal{F}_{\scriptscriptstyle\duqot}$ be the functionality defined above. A protocol $\Gamma$ realizes $\mathcal{F}_{\scriptscriptstyle\duqot}$ in the presence of passive adversaries, if for  every PPT adversary \adv in the real model, there exists a non-uniform PPT simulator \simm  in the ideal model, such that:
%
\begin{equation}\label{equ::DUQ-OT-sender-sim-}
\begin{split}
\Big\{\simm_{\st\se}\big((m_{\st 0}, m_{\st 1}), \empt\big)\Big\}_{\st m_{\st 0}, m_{\st 1}, s}\stackrel{c}{\equiv}  \Big\{\view_{\st \se}^{\st \Gamma}\big((m_{\st 0}, m_{\st 1}), \empt,  \empt, \empt, s\big) \Big\}_{\st m_{\st 0}, m_{\st 1}, s}
\end{split}
\end{equation}
\begin{equation}\label{equ::DUQ-OT-server-sim-}
\begin{split}
\Big\{\simm_{\st\p_i}(\empt, \empt)\Big\}_{\st m_{\st 0}, m_{\st 1}, s}\stackrel{c}{\equiv} 
\Big\{\view_{\st\p_i}^{\st \Gamma}\big((m_{\st 0}, m_{\st 1}), \empt,  \empt, \empt, s\big) \Big\}_{\st m_0, m_1, s}
\end{split}
\end{equation}
\begin{equation}\label{equ::DUQ-OT-t-sim-}
\begin{split}
\Big\{\simm_{\st\tp}(s, \empt)\Big\}_{\st m_0, m_1, s}\stackrel{c}{\equiv}  \Big\{\view_{\st\tp}^{\st \Gamma}\big((m_{\st 0}, m_{\st 1}), \empt,  \empt, \empt, s\big) \Big\}_{\st m_0, m_1, s}
\end{split}
\end{equation}
\begin{equation}\label{equ::DUQ-OT-reciever-sim-}
\begin{split}
&\Big\{\simm_{\st\re}\Big(\empt, \mathcal{F}_{\scriptscriptstyle\duqot}\big((m_{\st 0}, m_{\st 1}), \empt,  \empt,\empt, s\big)\Big)\Big\}_{\st m_0, m_1, s}\stackrel{c}{\equiv}\\  &\Big\{\view_{\st\re}^{\st \Gamma}\big((m_{\st 0}, m_{\st 1}), \empt,  \empt, \empt, s\big) \Big\}_{\st m_0, m_1, s}
\end{split}
\end{equation}

for all $i$,  $i\in \{1,2\}$. Since \duqot is a variant of \dqot, it also supports efficiency and \spc, as discussed in Section \ref{sec::sec-def}.

\end{definition}


\subsection{Protocol}\label{sec::DUQ-OT-protocol}
In this section, we present DUQ-OT that realizes \duqot.  

\subsubsection{Main Challenge to Overcome} 
One of the primary differences between DUQ-OT and previous OTs in the literature (and DQ-OT) is that in DUQ-OT, \re does not know the secret index $s$. The knowledge of $s$ would help \re pick the suitable element from \se's response; for instance, in the DQ-OT, it picks  $e_s$ from $(e_{\st 0}, e_{\st 1})$. Then, it can extract the message from the chosen element. In DUQ-OT, to enable \re to extract the desirable message from \se's response without the knowledge of $s$, we rely on the following observation and technique. 
We know that (in any OT) after decrypting $e_{\st s-1}$, \re obtains a value indistinguishable from a random value (otherwise, it would learn extra information about $m_{\st s-1}$). 
Thus, if \se imposes a certain publicly known structure on messages $(m_{\st 0}, m_{\st 1})$, then after decrypting \se's response, only $m_{\st s}$ would preserve the same structure. 

In DUQ-OT, \se imposes a publicly known structure to $(m_{\st 0}, m_{\st 1})$  and then computes the response. Given the response, \re tries to decrypt \emph{every} message it received from \se and accepts only the result that has the structure.

\subsubsection{An Overview} 
DUQ-OT operates as follows. First, \re picks two random values and sends each to a $\p_{\st i}$. Also, \tp splits the secret index $s$ into two shares and sends each share to a $\p_{\st i}$. Moreover, \tp selects a random value $r_{\st 3}$ and sends it to \re and \se. Given the messages received from \re and \tp, each $\p_{\st i}$ generates queries the same way they do in DQ-OT. 
Given the final query pair and $r_{\st 3}$,  \se first appends $r_{\st 3}$ to $m_{\st 0}$ and $m_{\st 1}$ and then computes the response the same way it does in DQ-OT, with the difference that it also randomly permutes the elements of the response pair.  Given the response pair and $r_{\st 3}$, \re decrypts each element in the pair and accepts the result that contains $r_{\st 3}$. 
Figure \ref{fig::DQ-OT-with-unknown-query} presents DUQ-OT in more detail. 

\begin{figure}[!htbp]
\setlength{\fboxsep}{.9pt}
\begin{center}
    \begin{tcolorbox}[enhanced, right=1mm,
    drop fuzzy shadow southwest,
    colframe=black,colback=white]
\begin{enumerate}[leftmargin=1mm]

\item \underline{\textit{$\se$-side Initialization:}} 
$\mathtt{Init}(1^{\st \lambda})\rightarrow pk$
\begin{enumerate}

\item choose a sufficiently large prime number $p$.

\item select a random element
$C \stackrel{\st \$}\leftarrow \mathbb{Z}_p$ and generator $g$.
\item publish $pk=(C, p, g)$. 

\end{enumerate}

\item \underline{\textit{\re-side Delegation:}} $\mathtt{\re.Request}( pk)\hspace{-.6mm}\rightarrow\hspace{-.6mm} req=(req_{\st 1}, req_{\st 2})$
\begin{enumerate}

\item pick two uniformly random values: 
$r_{\st 1}, r_{\st 2} \stackrel{\st\$}\leftarrow\mathbb{Z}_{\st p}$.

\item send $req_{\st 1}= r_{\st 1}$ to $\p_{\st 1}$ and $req_{\st 2}=r_{\st 2}$ to $\p_{\st 2}$.

\end{enumerate}

\item \underline{\textit{$\tp$-side Query Generation:}}
$\mathtt{\tp.Request}(1^{\st \lambda}, s, pk)\hspace{-.7mm}\rightarrow\hspace{-.7mm} (req'_{\st 1}, req'_{\st 2}, sp_{\st \se})$

\begin{enumerate}

\item split  the private index $s$ into two shares $(s_{\st 1}, s_{\st 2})$ by calling  $\ses(1^{\st\lambda}, s, 2, 2)\rightarrow (s_{\st 1}, s_{\st 2})$.

\item pick a uniformly random value: $r_{\st 3} \stackrel{\st\$}\leftarrow\{0,1\}^{\st\lambda}$.

\item send $req'_{\st 1} =s_{\st 1}$ to  $\p_{\st 1}$, $req'_{\st 2}=s_{\st 2}$ to  $\p_{\st 2}$. It also sends secret parameter $sp_{\st \se}=r_{\st 3}$ to \se and $sp_{\st \re} = (req'_{\st 2}, sp_{\st \se})$ to \re.

\end{enumerate}

\item \underline{\textit{$\p_{\st 2}$-side Query Generation:}}
$\mathtt{\p_{\st 2}.GenQuery}(req_{\st 2}, req'_{\st 2}, pk)\rightarrow q_{\st 2}$

\begin{enumerate}

\item compute a pair of partial queries:
%
%
  $$\delta_{\st s_2}= g^{\st r_2},\ \  \delta_{\st 1-s_2} = \frac{C}{g^{\st r_2}}$$
  
\item send $q_{\st 2}=(\delta_{\st 0}, \delta_{\st 1})$ to  $\p_{\st 1}$. 

\end{enumerate}

\item\underline{\textit{$\p_{\st 1}$-side Query Generation:}}
$\mathtt{\p_{\st 1}.GenQuery}(req_{\st 1},req'_{\st 1}, q_{\st 2},pk)\rightarrow q_{\st 1}$

\begin{enumerate}

\item compute a pair of final queries as: 
%
%
$$\beta_{\st s_1}=\delta_{\st 0}\cdot g^{\st r_1},\ \ \beta_{\st 1-s_1}=\frac{\delta_{\st 1}} {g^{\st r_1}}$$

\item send $q_{\st 1} =(\beta_{\st 0}, \beta_{\st 1})$ to  $\se$.

\end{enumerate}

\item\underline{\textit{\se-side Response Generation:}} 
$\mathtt{GenRes}(m_{\st 0}, m_{\st 1}, pk, q_{\st 1}, sp_{\st \se})\rightarrow res$

\begin{enumerate}

\item abort if  $C \neq \beta_{\st 0}\cdot \beta_{\st 1}$.
\item pick two uniformly random values: 
$y_{\st 0}, y_{\st 1}  \stackrel{\st\$}\leftarrow\mathbb{Z}_{\st p}$.
    
\item compute a response pair $(e_{\st 0}, e_{\st 1})$ as follows:
$$e_{\st 0} := (e_{\st 0,0}, e_{\st 0,1}) = (g^{\st y_0}, \g(\beta_{\st 0}^{\st y_0}) \oplus (m_{\st 0}||r_{\st 3}))$$
$$e_{\st 1} := (e_{\st 1,0}, e_{\st 1,1}) = (g^{\st y_1}, \g(\beta_{\st 1}^{\st y_1}) \oplus (m_{\st 1}||r_{\st 3}))$$

\item randomly permute the elements of the pair $(e_{\st 0}, e_{\st 1})$ as follows: $\pi(e_{\st 0}, e_{\st 1})\rightarrow ({e}'_{\st 0}, {e}'_{\st 1})$.

\item send $res=(e'_{\st 0}, e'_{\st 1})$ to \re. 

\end{enumerate}

\item\underline{\textit{\re-side Message Extraction:}} 
$\mathtt{Retrieve}(res, req, pk, sp_{\st \re})\hspace{-.6mm}\rightarrow\hspace{-.8mm} m_{\st s}$

\begin{enumerate}

\item set $x=r_{\st 2}+r_{\st 1}\cdot(-1)^{\st s_2}$. 

\item retrieve message $m_{\st s}$ as follows. 
 $\forall i, 0\leq i\leq1:$

\begin{enumerate}

\item set $y=\g(({e}'_{\st i, 0})^{\st x})\oplus {e}'_{\st i, 1}$.

\item  call $\parse(\gamma, y)\rightarrow (\ux, \uy)$.

\item set  $m_{\st s}=\ux$, if $\uy=r_{\st 3}$.

\end{enumerate}

\end{enumerate}

\end{enumerate}
\end{tcolorbox}
\end{center}
\vspace{-2.5mm}
    \caption{DUQ-OT: Our $1$-out-of-$2$ OT that supports query delegation while preserving the privacy of query from \re. In the protocol, $\g$ is a hash function, $\pi$ is a random permutation, and $\$$ denotes picking a value uniformly at random.  
    }
    \label{fig::DQ-OT-with-unknown-query}
    \vspace{-3mm}
\end{figure}

\begin{theorem}\label{theo::DUQ-OT-sec}
%
If  
DL, CDH, and RO assumptions hold and random permutation $\pi$  is secure, DUQ-OT (presented in Figure \ref{fig::DQ-OT-with-unknown-query}) securely computes $\mathcal{F}_{\scriptscriptstyle\duqot}$ (defined in Section \ref{sec::DUQ-OT-definition}) in the presence of semi-honest adversaries, 
%
%
w.r.t. Definition~\ref{def::DUQ-OT-sec-def}. 

\end{theorem}

%
\begin{proof}[Proof sketch]

When \se is corrupted, the received queries are distributed independently of $s$ (under DL); therefore, a simulator can sample a matching transcript without knowing $s$. 
When \re is corrupted, learning anything about $m_{\st 1-s}$ would require computing the
unchosen DH key (CDH); modelling $\g$ as a random oracle, the unchosen pad is
indistinguishable from uniform; moreover, the pad $r_{\st 3}$ is chosen uniformly at random and independent of the index $s$; thus, the view can be simulated given only $m_{\st s}$. If a proxy is corrupted, it sees only uniformly random shares and fresh randomness, plus group elements with hidden exponents; thus, its view is simulatable under DL.  The view of \tp can be easily simulated; it has input $s$; however, it receives no messages from its counterparts (other than the public parameter $pk$) and receives no output from the protocol. 
\end{proof}

Appendix \ref{sec::DUQ-OT-Security-Proof} presents full proof of Theorem \ref{theo::DUQ-OT-sec}.

\section{Delegated-Query Multi-Receiver Oblivious Transfers}\label{sec::Multi-Receiver-OT}

In this section, we present two new variants of \dqot; namely, (1) Delegated-Query Multi-Receiver OT (\dqothf) and\\ (2) Delegated-Unknown-Query Multi-Receiver OT (\duqothf). They are suitable for the \emph{multi-receiver} setting in which the sender maintains a (large) database containing $z$ pairs of messages $\bm m=[(m_{\st 0, 0},m_{\st 1, 0}),...,$ $ (m_{\st 0, z-1},$ $m_{\st 1, z-1})]$. 

In this setting, each pair, say $v$-th pair $(m_{\st 0, v},$ $m_{\st 1, v})\in \bm m$ is related to  a receiver,  $\re$, where $0\leq v\leq z-1$. Both variants (in addition to offering the efficiency, \spc, and security guarantee of \dqot) ensure that (i) a receiver learns nothing about the total number of receivers/pairs (i.e., $z$) and (ii) the sender learns nothing about which receiver is sending the query, i.e., a message pair's index for which a query was generated.  In the remainder of this section, we discuss these new variants.

\vspace{-2mm}
\subsection{Delegated-Query Multi-Receiver OT}\label{sec::DQ-OT-HF}
The first variant \dqothf considers the setting where server $\p_{\st 1}$ or $\p_{\st 2}$ knows a client's related pair's index in the sender's database.


 \subsubsection{Security Definition}\label{sec::DQOT-HF}

The functionality that \dqothf computes takes as input (i)  a vector of messages $\bm{m}=[(m_{\st 0, 0},m_{\st 1, 0}),...,$ $ (m_{\st 0, z-1},$ $m_{\st 1, z-1})]$ from \se, (ii) an index $v$ of a pair in $\bm{m}$ from $\p_{\st 1}$, (iii)  empty string  \empt from $\p_{\st 2}$, and (iv) the index $s$ (where $s\in \{0, 1\}$) from \re. It outputs an empty string $\empt$ to  \se, $z$ to $\p_{\st 1}$, $\empt$ to $\p_{\st 2}$, and outputs to \re $s$-th message from $v$-th pair in the vector, i.e., $m_{\st s, v}$. Formally, we define the functionality as: 
 $\mathcal{F}_{\scriptscriptstyle\dqothf}:\big([(m_{\st 0, 0},m_{\st 1, 0}),...,$ $ (m_{\st 0, z-1},$ $m_{\st 1, z-1})], v, \empt, s\big) \rightarrow (\empt, z, \empt, m_{\st s, v})$, where $v\in\{0,..., z-1\}$. Next, we present a formal definition of \dqothf.

\begin{definition}[\dqothf]\label{def::DQ-OT-HF-sec-def} Let $\mathcal{F}_{\scriptscriptstyle\dqothf}$ be the functionality defined above. We say that protocol $\Gamma$ realizes $\mathcal{F}_{\scriptscriptstyle\dqothf}$ in the presence of passive adversaries, if for  every non-uniform PPT adversary \adv in the real model, there exists a non-uniform PPT simulator \simm  in
the ideal model, such that:
%
\begin{equation}\label{equ::sender-sim-DQ-OT-HF}
\Big\{\simm_{\st\se}\big(\bm{m}, \empt\big)\Big\}_{\st \bm{m}, v, s}\stackrel{c}{\equiv} \Big\{\view_{\st\se}^{\st \Gamma}\big(\bm{m}, v,  \empt, s\big) \Big\}_{\st \bm{m}, v, s}
\end{equation}
\begin{equation}\label{equ::server1-sim-DQ-OT-HF}
\Big\{\simm_{\st\p_1}(v, z)\Big\}_{\st \bm{m}, v, s}\stackrel{c}{\equiv} \Big\{\view_{\st\p_1}^{\st \Gamma}\big(\bm{m}, v,  \empt, s\big) \Big\}_{\st \bm{m}, v, s}
\end{equation}
\begin{equation}\label{equ::server2-sim-DQ-OT-HF}
\Big\{\simm_{\st\p_2}(\empt, \empt)\Big\}_{\st \bm{m}, v, s}\stackrel{c}{\equiv} \Big\{\view_{\st\p_2}^{ \st\Gamma}\big(\bm{m}, v,  \empt, s\big) \Big\}_{\st \bm{m}, v, s}
\end{equation}
\begin{equation}\label{equ::reciever-sim-DQ-OT-HF}
\begin{split}
\Big\{\simm_{\st\re}\Big(s, \mathcal{F}_{\scriptscriptstyle\dqothf}\big(\bm{m}, v,  \empt, s\big)\Big)\Big\}_{\st \bm{m}, v,  s} \hspace{-2mm}
\stackrel{c}{\equiv}    \Big\{\view_{\st\re}^{ \st\Gamma}\big(\bm{m}, v,  \empt, s\big) \Big\}_{\st \bm{m}, v, s}
\end{split}
\end{equation}
where $\bm{m}=[(m_{\st 0, 0},m_{\st 1, 0}),..., (m_{\st 0, z-1}, m_{\st 1, z-1})]$.

\end{definition}


\subsubsection{Strawman Approaches}\label{sec::OT-HFstrawman-aprroach}

One may consider using one of the following ideas in the multi-receiver setting: 

\begin{enumerate}[leftmargin=5mm]
    \item \underline{\textit{Using an existing single-receiver OT, e.g., in \cite{IshaiKNP03}}}, employing one of the following approaches:

    \begin{itemize}
        \item \textit{Approach 1}: receiver $\re$ sends a standard OT query to \se which computes the response for all $z$ pairs of messages. Subsequently, \se sends $z$ pair of responses to receiver $\re$ which discards all pairs from the response except for $v$-th pair. $\re$ extracts its message $m_{\st v}$ from the selected pair, similar to a regular $1$-out-of-$2$ OT.  However, this approach results in the leakage of the entire database size to $\re$.

        \item \textit{Approach 2}: $\re$ sends a standard OT query to \se, along with the index $v$ of its record.  This can be perceived as if \se holds a single record/pair. Accordingly, \se generates a response in the same manner as it does in regular $1$-out-of-$2$ OT. Nevertheless, Approach 2 leaks to $\se$ the index $v$ of the record that $\re$ is interested.
    \end{itemize}
    \item \underline{\textit{Using an existing multi-receiver OT, e.g., in \cite{CamenischDN09}}}. This will also come with a privacy cost. The existing multi-receiver OTs reveal the entire database's size to each receiver (as discussed in Section \ref{sec::multi-rec-ot}). In this scenario, a receiver can learn the number of private records other companies have in the same database. This type of leakage is particularly significant, especially when coupled with specific auxiliary information. 
\end{enumerate}
 Hence, a fully private multi-receiver OT is necessary to ensure user privacy in real-world cloud settings.

\subsubsection{Protocol}\label{sec::rdqothf}
We present \rdqothf that realizes \dqothf. We build \rdqothf upon DQ-OT (presented in Figure \ref{fig::DQ-OT}). \rdqothf relies on our observation that in DQ-OT, given the response of \se, $\p_{\st 1}$ cannot learn anything, e.g., about the plaintext messages $m_{\st i}$ of  \se. Below, we formally state it.

\begin{lemma}\label{lemma::two-pairs-indis-}
Let $g$ be a generator of a group $\mathbb{G}$ (defined in Section~\ref{sec::dh-assumption}) whose order is a prime number $p$ and $\log_{\st 2}(p)=\lambda$ is a security parameter. Also, let $(r_{\st 1}, r_{\st 2}, y_{\st 1}, y_{\st 2})$ be elements of $\mathbb{G}$ picked uniformly at random, $C=g^{\st a}$ be a random public value whose discrete logarithm is unknown, $(m_{\st 0}, m_{\st 1})$ be two arbitrary messages, and $\h$ be a hash function modelled as a RO (as defined in Section \ref{sec::notations}), where its output size is $\delta$-bit.  Let $\gamma=\stackrel{\st+}-r_{\st 1} \stackrel{\st+}-r_{\st 2}$, $\beta_{\st 0}=g^{\st a+\gamma}$, and $\beta_{\st 1}=g^{\st a-\gamma}$.  If DL, RO, and CDH assumptions hold, then given $r_{\st 1}, C, g^{\st r_2}$, and $\frac{C}{g^{\st r_2}}$, a PPT distinguisher cannot distinguish (i) $g^{\st y_0}$ and $g^{\st y_1}$ form random elements of $\mathbb{G}$ and (ii) $ \h(\beta_{\st 0}^{\st y_0}) \oplus m_{\st 0}$ and $ \h(\beta_{\st 1}^{\st y_1}) \oplus m_{\st 1}$  from random elements from $\{0, 1\}^{\st\sigma}$, except for a negligible probability  $\mu(\lambda)$.



\end{lemma}

Appendix \ref{sec::proof-of-thorem} presents the proof of Lemma \ref{lemma::two-pairs-indis-}. 
The main idea behind the design of \rdqothf is as follows. Given a message pair from $\p_{\st 1}$,  \se needs to compute the response for all of the receivers and sends the result to  $\p_{\st 1}$, which picks and sends only one pair in the response to the specific receiver who sent the query and discards the rest of the pairs it received from \se. Therefore, \re receives a single pair (so it cannot learn the total number of receivers or the database size), and the server cannot know which receiver sent the query, as it generates the response for all of them. As we will prove, $\p_{\st 1}$ itself cannot learn the actual query of \re, too.

Consider the case where one of the receivers, say  $\re$, wants to send a query. In this case, within \rdqothf, messages $(s_{\st 1}, r_{\st 1})$, $(s_{\st 2}, r_{\st 2})$ and $(\beta_{\st 0}, \beta_{\st 1})$ are generated the same way as they are computed in DQ-OT.  However, given $(\beta_{\st 0}, \beta_{\st 1})$, \se generates $z$ pairs and sends them to $\p_{\st 1}$ who forwards only $v$-th pair to $\re$ and discards the rest.  Given the pair, $\re$ computes the result the same way a receiver does in DQ-OT. Figure \ref{fig::DQHT-OT} in Appendix \ref{sec::DQ-HF-OT-detailed-protocol} presents \rdqothf in detail.


\subsection{Delegated-Unknown-Query Multi-Receiver OT}\label{sec::Delegated-Unknown-Query-OT-HF}

The second variant    \duqothf can be considered as a variant of  \duqot. It is suitable for the setting where servers $\p_{\st 1}$ and $\p_{\st 2}$ do not (and must not) know a client's related index in the sender's database (as well as the index $s$ of the message that the client is interested in).

\subsubsection{Security Definition}\label{sec::DUQOT-HD-def}

The functionality that \duqothf computes takes as input (i)  a vector of messages $\bm{m}=[(m_{\st 0, 0},m_{\st 1, 0}),...,$ $ (m_{\st 0, z-1},$ $m_{\st 1, z-1})]$ from \se, (ii) an index $v$ of a pair in $\bm{m}$ from $\tp$, (iii) the index $s$ of a message in a pair (where $s\in \{0, 1\}$) from \tp, (iv) the total number $z$ of message pairs from $\tp$, (v)  empty string  \empt from $\p_{\st 1}$, (vi)   \empt from $\p_{\st 2}$, and (vii)  \empt from $\re$. It outputs an empty string $\empt$ to  \se and $\tp$, $z$ to $\p_{\st 1}$, $\empt$ to $\p_{\st 2}$, and outputs to \re $s$-th message from $v$-th pair in $\bm{m}$, i.e., $m_{\st s, v}$. Formally, we define the functionality as: $\mathcal{F}_{\scriptscriptstyle\duqothf}:\big([(m_{\st 0, 0},m_{\st 1, 0}),...,$ $ (m_{\st 0, z-1},$ $m_{\st 1, z-1})], (v, s, z), \empt, \empt, \empt\big) \rightarrow ( \empt, \empt, z, \empt, m_{\st s, v})$, where $v\in\{0,..., z-1\}$. Next, we present a formal definition of \duqothf.


\begin{definition}[\duqothf]\label{def::DUQ-OT-HF-sec-def} Let $\mathcal{F}_{\scriptscriptstyle\duqothf}$ be the functionality defined above. We assert that protocol $\Gamma$ realizes $\mathcal{F}_{\scriptscriptstyle\duqothf}$ in the presence of passive adversaries, if for  every non-uniform PPT adversary \adv
in the real model, there exists a non-uniform PPT simulator \simm  in
the ideal model, such that:
\begin{equation}\label{equ::DUQ-OT-HF-sender-sim-}
\begin{split}
\Big\{\simm_{\st\se}\big(\bm{m}, \empt\big)\Big\}_{\st \bm{m}, s}\stackrel{c}{\equiv}  \Big\{\view_{\st\se}^{\st \Gamma}\big(\bm{m}, (v, s, z), \empt,  \empt, \empt \big) \Big\}_{\st \bm{m}, s}
\end{split}
\end{equation}
\begin{equation}\label{equ::DUQ-OT-HF-server-sim-}
\begin{split}
\Big\{\simm_{\st\p_i}(\empt, out_{\st i})\Big\}_{\st \bm{m}, s}\stackrel{c}{\equiv}  \Big\{\view_{\st\p_i}^{ \st\Gamma}\big(\bm{m}, (v, s, z), \empt,  \empt, \empt \big) \Big\}_{\st \bm{m}, s}
\end{split}
\end{equation}
\begin{equation}\label{equ::DUQ-OT-HF-t-sim-}
\begin{split}
\Big\{\simm_{\st\tp}\big((v, s, z), \empt\big)\Big\}_{\st \bm{m}, s}\stackrel{c}{\equiv}  \Big\{\view_{\st\tp}^{ \st\Gamma}\big(\bm{m}, (v, s, z), \empt,  \empt, \empt \big) \Big\}_{\st \bm{m}, s}
\end{split}
\end{equation}
\begin{equation}\label{equ::DUQ-OT-HF-reciever-sim-}
\begin{split}
&\Big\{\simm_{\st\re}\Big(\empt, \mathcal{F}_{\scriptscriptstyle\duqothf}\big(\bm{m}, (v, s, z), \empt,  \empt, \empt \big)\big)\Big)\Big\}_{\st \bm{m}, s}\stackrel{c}{\equiv}\\  &\Big\{\view_{\st\re}^{ \st\Gamma}\big(\bm{m}, (v, s, z), \empt,  \empt, \empt \big) \Big\}_{\st \bm{m}, s}
\end{split}
\end{equation}
where $\bm{m}=[(m_{\st 0, 0},m_{\st 1, 0}),..., (m_{\st 0, z-1}, m_{\st 1, z-1})]$, $out_{\st 1}=z$, $out_{\st 2}=\empt$, and $\forall i$,  $i\in \{1,2\}$. 
\end{definition}

\subsubsection{Protocol}\label{sec::DUQ-OT-HF} We proceed to present \rduqothf that realizes \duqothf. We build \rduqothf upon protocol DUQ-OT (presented in Figure \ref{fig::DQ-OT-with-unknown-query}). \rduqothf mainly relies on Lemma \ref{lemma::two-pairs-indis-} and the following technique.


To fetch a record $m_{\st v}$ ``securely'' from a semi-honest \se that holds a database of the form $\bm{a}=[m_{\st 0}, m_{\st 1}, ..., m_{\st z-1}]^{\st T}$ where $T$ denotes transpose, 
%
%
without revealing which plaintext record we want to fetch, we can perform as follows: 

\begin{enumerate}
\item construct vector $\bm{b}=[b_{\st 0},..., b_{\st z-1}]$, where all $b_{\st i}$s are set to zero except for $v$-th element $b_{\st v}$ which is set to $1$.

\item encrypt each element of $\bm{b}$ using additively homomorphic encryption, e.g., Paillier encryption. Let $\bm{b}'$ be the vector of the encrypted elements. 

 \item send $\bm{b}'$ to the database holder which performs $\bm{b}'\times \bm{a}$ homomorphically, and sends us the single result $res$.
 
 \item decrypt $res$ to discover $m_{\st v}$.\footnote{Such a technique was previously used by Devet \textit{et al.} \cite{DevetGH12} in the ``private information retrieval'' research line.} 
 \end{enumerate}

 In \rduqothf,  $\bm{b}'$ is not sent for each query to  \se. Instead,  it is stored once in one of the servers, for example, $\p_{\st 1}$. Any time \se computes a vector of responses, say $\bm{a}$, to an OT query, it sends $\bm{a}$ to  $\p_{\st 1}$ which computes $\bm{b}'\times \bm{a}$ homomorphically and sends the result to \re which decrypts its and finds the message it was interested.   Thus, $\p_{\st 1}$ \emph{obliviously filters out} all other records of field elements that do not belong to $\re$ and sends to $\re$ only the messages that $\re$ is allowed to fetch. 
 Figures  \ref{fig::DQ-OT-with-unknown-query-and-HF-first-eight-phases} and \ref{fig::DQ-OT-with-unknown-query-and-HF-last-three-phases} present \rduqothf in detail.

%
\begin{figure}[!htbp]
\setlength{\fboxsep}{.9pt}
\begin{center}
    \begin{tcolorbox}[enhanced,width=86mm, 
    drop fuzzy shadow southwest,
    colframe=black,colback=white]

\begin{enumerate}[leftmargin=4.5mm]

\item\label{phase::s-init} \underline{\textit{$\se$-side Initialization:}} 
$\mathtt{Init}(1^{\st \lambda})\rightarrow pk$

It chooses a large random prime number $p$, random element
$C \stackrel{\st \$}\leftarrow \mathbb{Z}_p$, and generator $g$. It publishes $pk=(C, p, g)$.

\item\label{DUQOT-HT::gen-key} \underline{\textit{$\re$-side One-off Setup:}}
$\mathtt{\re.Setup}(1^{\st \lambda})\rightarrow (pk_{\st j}, sk_{\st j})$
%


It generates a key pair for the  homomorphic encryption, by calling $\keygen(1^{\st\lambda})\rightarrow(sk_{\st j}, pk_{\st j})$. It send $pk_{\st j}$  to \tp and $\se$.

\item  \underline{\textit{\tp-side One-off Setup:}}
$\mathtt{\tp.Setup}(z, pk_{\st j})\rightarrow \bm{w}_{\st j}$

\begin{enumerate}

\item initialize an empty vector $\bm{w}_{\st j}=[]$ of size $z$.

\item create a compressing vector, by setting $v$-th position of $\bm{w}_{\st j}$ to encrypted $1$ and setting the rest of $z-1$ positions to encrypted $0$. $ \forall t, 0\leq t\leq z-1:$

\begin{enumerate}

\item  set  $d=1$, if $t=v$; set $d=0$, otherwise.

\item  append $\enc(pk_{\st j}, d)$ to  $\bm{w}_{\st j}$.

\end{enumerate}

\item send $\bm{w}_{\st j}$ to $\p_{\st 1}$.

\end{enumerate}

\item \underline{\textit{$\re$-side Delegation:}}
$\mathtt{\re.Request}( pk)\hspace{-.6mm}\rightarrow\hspace{-.6mm} req=(req_{\st 1}, req_{\st 2})$

\begin{enumerate} 

\item pick random values: $r_{\st 1}, r_{\st 2} \stackrel{\st\$}\leftarrow\mathbb{Z}_{\st p}$.

\item send $req_{\st 1} = r_{\st 1}$ to $\p_{\st 1}$ and $req_{\st 2}= r_{\st 2}$ to $\p_{\st 2}$.

\end{enumerate}

\item \underline{\textit{$\tp$-side Query Generation:}}
$\mathtt{\tp.Request}(1^{\st \lambda}, s,$ $ pk)\hspace{-1mm}\rightarrow\hspace{-1mm} (req'_{\st 1}, $ $req'_{\st 2},$ $ sp_{\st \se})$

\begin{enumerate}

\item split  the private index $s$ into two shares $(s_{\st 1}, s_{\st 2})$ by calling  $\ses(1^{\st \lambda}, s, 2, 2)\rightarrow (s_{\st 1}, s_{\st 2})$.

\item pick a uniformly random value: $r_{\st 3} \stackrel{\st\$}\leftarrow\{0,1\}^{\st\lambda}$.

\item send $req'_{\st 1} =s_{\st 1}$ to  $\p_{\st 1}$, $req'_{\st 2}=s_{\st 2}$ to  $\p_{\st 2}$. It sends secret parameter $sp_{\st \se}=r_{\st 3}$ to \se and $sp_{\st \re} = (req'_{\st 2}, sp_{\st \se})$ to \re.

\end{enumerate}
\item \underline{\textit{$\p_{\st 2}$-side Query Generation:}}
$\mathtt{\p_{\st 2}.GenQuery}(req_{\st 2}, req'_{\st 2}, pk)\rightarrow q_{\st 2}$

\begin{enumerate}

\item compute queries:
  $\delta_{\st s_2}= g^{\st r_2},\ \  \delta_{\st 1-s_2} = \frac{C}{g^{\st r_{\st 2}}}$.
\item send $q_{\st 2}=(\delta_{\st 0}, \delta_{\st 1})$ to  $\p_{\st 1}$. 

\end{enumerate}

\item\label{DUQOT-HT::gen-P1-side-q-gen}\underline{\textit{$\p_{\st 1}$-side Query Generation:}}
$\mathtt{\p_{\st 1}.GenQuery}(req_{\st 1},req'_{\st 1}, q_{\st 2},pk)\hspace{-1.1mm}\rightarrow\hspace{-1.1mm} q_{\st 1}$

\begin{enumerate}

\item compute  queries as: 
%
%
$\beta_{\st s_{\st 1}}=\delta_{\st 0}\cdot g^{\st r_1}, \beta_{\st 1-s_1}=\frac{\delta_{\st 1}} {g^{\st r_1}}$ 

\item send $q_{\st 1}=(\beta_{\st 0}, \beta_{\st 1})$ to  $\se$.
\end{enumerate}


\item\label{DUQOT-HT::gen-res}\underline{\textit{\se-side Response Generation:}} 
$\mathtt{GenRes}(m_{\st 0,0}, m_{\st 1, 0},\ldots, m_{\st 0,z-1},$ $ m_{\st 1, z-1}, pk, q_{\st 1}, sp_{\st \se})\rightarrow res$

\begin{enumerate}

\item abort if  $C \neq \beta_{\st 0}\cdot \beta_{\st 1}$.

\item compute a response as follows. $\forall t, 0\leq t \leq z-1:$
\begin{enumerate}[leftmargin=1.2mm]
\item   pick two random values $y_{\st 0, t}, y_{\st 1, t}  \stackrel{\$}\leftarrow\mathbb{Z}_{\st p}$.  
  
\item  compute  response:

 $ e_{\st 0, t} := (e_{\st 0, 0, t}, e_{\st 0, 1, t}) = (g^{\st y_{\st 0, t}}, \g(\beta_{\st 0}^{\st y_{0, t}}) \oplus (m_{\st 0, t}||r_{\st 3}))$
      
$e_{\st 1, t} := (e_{\st 1, 0, t}, e_{\st 1, 1, t}) = (g^{\st y_{1, t}}, \g(\beta_{\st 1}^{\st y_{\st 1, t}}) \oplus (m_{\st 1, t}||r_{\st 3})) $ 
   
\item   randomly permute the elements of each pair $(e_{\st 0, t}, e_{\st 1, t})$ as $\pi(e_{\st 0, t}, e_{\st 1, t})\rightarrow ({e}'_{\st 0, t}, {e}'_{\st 1, t})$.
   
\end{enumerate}

\item     send $res=(e'_{\st 0, 0}, e'_{\st 1, 0}),\ldots, (e'_{\st 0, z-1}, e'_{\st 1, z-1}) \text{ to } \p_{\st 1}$.
\end{enumerate}



    
 















\end{enumerate}
\vspace{-1mm}
\end{tcolorbox}
\end{center}
\vspace{-4mm}
    \caption{Phases \ref{phase::s-init}--\ref{DUQOT-HT::gen-res} of \rduqothf. 
    }
    \label{fig::DQ-OT-with-unknown-query-and-HF-first-eight-phases}
    \vspace{-3mm}
\end{figure}

%
\begin{figure}[!htbp]
\setlength{\fboxsep}{.9pt}
\begin{center}
    \begin{tcolorbox}[enhanced,right=1mm,  
    drop fuzzy shadow southwest,
    colframe=black,colback=white]

\begin{enumerate}[leftmargin=4mm,start=9]

\item\label{DUQOT-HT::oblivius-filter}\underline{\textit{$\p_{\st 1}$-side Oblivious Filtering:}} 
$\mathtt{OblFilter}(res, pk_{\st j}, \bm{w}_{\st j})\hspace{-.7mm}\rightarrow\hspace{-.6mm} res'$

\begin{enumerate}

\item compress \se's response using vector  $\bm{w}_{\st j}$ as follows. $\forall i,i', 0\leq i,i'\leq 1:$
    
$o_{\st i, i'}= (e'_{\st i, i',0}\hmul \bm{w}_{\st j}[0])\hadd...\hadd  (e'_{\st i,i', z-1}\hmul \bm{w}_{\st j}[z-1])$.
 
\item  send $res'=(o_{\st 0, 0}, o_{\st 0, 1}), (o_{\st 1, 0}, o_{\st 1, 1}) \text{ to } \re$.

\end{enumerate}

\item\label{DUQOT-HT::message-ext} \underline{\textit{\re-side Message Extraction:}}  
$\mathtt{Retrieve}(res', req,  sk_{\st j}, pk,$ $ sp_{\st \re})\hspace{-.6mm}\rightarrow\hspace{-.8mm} m_{\st s}$

\begin{enumerate}

\item decrypt the response from $\p_{\st 1}$ as follows: 

 $\forall i,i', 0\leq i,i'\leq 1:$
$ \dec(sk_{\st j}, o_{\st i,i'})\rightarrow o'_{\st i,i'}$.

\item set $x=r_{\st 2}+r_{\st 1}\cdot(-1)^{\st s_2}$.

\item retrieve message $m_{\st s, v}$ as follows. 
 $\forall i, 0\leq i\leq1:$

\begin{enumerate}

\item set $y=\g(({o}'_{\st i, 0})^{\st x})\oplus {o}'_{\st i, 1}$.

\item call $\parse(\gamma, y)\rightarrow (\ux, \uy)$.

\item  set $m_{\st s, v}=\ux$, if $\uy=r_{\st 3}$.

\end{enumerate}

\end{enumerate}


\end{enumerate}
\vspace{-1mm}
\end{tcolorbox}
\end{center}
\vspace{-3mm}
    \caption{Phases \ref{DUQOT-HT::oblivius-filter} and \ref{DUQOT-HT::message-ext} of \rduqothf.    
    }
    \label{fig::DQ-OT-with-unknown-query-and-HF-last-three-phases}
    \vspace{-3mm}
\end{figure}




\begin{theorem}\label{theo::DUQ-OTHF-2-sec}
Let $\mathcal{F}_{\scriptscriptstyle\duqothf}$ be the functionality defined in Section \ref{sec::DUQOT-HD-def}. If  
DL, CDH, and RO assumptions hold, and additive homomorphic encryption satisfies IND-CPA, then \rduqothf (presented in Figures \ref{fig::DQ-OT-with-unknown-query-and-HF-first-eight-phases} and \ref{fig::DQ-OT-with-unknown-query-and-HF-last-three-phases}) securely computes $\mathcal{F}_{\scriptscriptstyle\duqothf}$ in the presence of semi-honest adversaries, 
%
%
w.r.t. Definition \ref{def::DUQ-OT-HF-sec-def}. 
\end{theorem}
We refer readers to Appendix \ref{sec::proof-of-DUQ-OT-HF} for Theorem \ref{theo::DUQ-OTHF-2-sec}'s proof.

\section{A Compiler for Generic OT with Constant Size Response}\label{sec::the-compiler}

In this section, we present a compiler that transforms  \emph{any} $1$-out-of-$n$ OT that requires \re to receive $n$ messages into a $1$-out-of-$n$ OT that enables \re to receive only a \emph{constant} number of messages.

 The main technique we rely on is the encrypted binary vector that we used in Section \ref{sec::Delegated-Unknown-Query-OT-HF}.  The high-level idea is as follows.  During query computation, \re (along with its vector that encodes its index $s\in\{0, n-1\}$) computes a binary vector of size $n$, where all elements of the vector are set to $0$ except for the $s$-th element, which is set to $1$. \re encrypts each element of the vector and sends the result, as well as its query to \se. Subsequently, \se computes a response vector (in the same manner it does in regular OT), homomorphically multiplies each element of the response by the element of the encrypted vector (component-wise), and then homomorphically sums all the products. It sends the result (which is now constant with regard to $n$) to \re, which decrypts the response and retrieves the result $m_{\st s}$. 
 
 Next, we will present a generic OT's syntax, and introduce the generic compiler using the syntax. 
 

\subsection{Syntax of a Conventional OT}\label{sec::OT-syntax}


Since we aim to treat any  OT in a block-box manner, we first present the syntax of an OT. A conventional (or non-delegated) $1$-out-of-$n$ OT (\onenot) have the following algorithms:

\begin{itemize}[leftmargin=4.5mm]

\item[$\bullet$] $\mathtt{\se.Init}(1^{\st\lambda})\rightarrow pk$: a probabilistic algorithm run by \se. It takes as input security parameter $1^{\st\lambda}$ and returns a public key $pk$. 




\item[$\bullet$] $\mathtt{\re.GenQuery}(pk, n, s)\rightarrow (q, sp)$: a probabilistic algorithm run by \re. It takes as input $pk$,  $n$, and a secret index $s$. It returns a query (vector) $q$ and a secret parameter $sp$.  

 \item[$\bullet$] $\mathtt{\se.GenRes}(m_{\st 0},\ldots,m_{\st n-1}, pk, q)\rightarrow res$: a probabilistic algorithm run by \se. It takes as input $pk$ and $q$. It generates an encoded response (vector) $res$.

  \item[$\bullet$] $\mathtt{\re.Retrieve}(res, q, sp, pk, s)\rightarrow m_{\st s}$: a deterministic algorithm run by \se. It takes as input $res$,  $q$, $sp$,  $pk$, and $s$. It returns message $m_{\st s}$. 
 
\end{itemize}








 

The functionality that a $1$-out-of-$n$ OT  computes can be defined as: 
 $\mathcal{F}_{\scriptscriptstyle\onenot}:\big((m_{\st 0}, ..., m_{\st n-1}), s\big) \rightarrow (\empt, m_{\st s})$. Informally, the security of $1$-out-of-$n$ OT states that (1) \re's view can be simulated given its input query $s$ and output message $m_{\st s}$ and (2) \se's view can be simulated given its input messages $(m_{\st 0}, ..., m_{\st n-1})$. We refer readers to \cite{DBLP:books/cu/Goldreich2004} for further discussion on $1$-out-of-$n$ OT.

\subsection{The Compiler}\label{sec::compiler}
We present the compiler in detail in Figure \ref{fig::generic-short-res-HE}.  We highlight that in the case where each $e_i \in res$ contains more than one value, e.g., $e_{\st i}=[e_{\st 0,i},..., e_{\st w-1,i}]$ (due to a certain protocol design), then each element of $e_i$ is separately multiplied and added by the element of vector $\bm{b}'$, e.g., the $j$-st element of the response is $e_{\st j, 0}\hmul \bm{b}'[0]\hadd...\hadd e_{\st j, n-1}\hmul \bm{b}'[n-1]$, for all $j$, $0\leq j\leq w-1$. In this case, only $w$ elements are sent to \re.

\begin{figure}[!htbp]
\vspace{-2mm}
\setlength{\fboxsep}{.9pt}
\begin{center}
    \begin{tcolorbox}[enhanced,width=81mm, 
    drop fuzzy shadow southwest,
    colframe=black,colback=white]
\vspace{-1mm}
\begin{enumerate}[leftmargin=5.2mm]
\item\underline{\textit{$\se$-side Initialization:}}
$\mathtt{Init}(1^{\st\lambda})\rightarrow pk$

This phase involves \se. 
\begin{enumerate}
    \item call $\mathtt{\se.Init}(1^{\st\lambda})\rightarrow pk$. 
    \item publish $pk$.
\end{enumerate}

\item\underline{\re-side Setup:} 
$\mathtt{Setup}(1^{\st\lambda})\rightarrow (sk_{\st \re}, pk_{\st \re})$ 

This phase involves \re. 
\begin{enumerate}
 \item call 
 $\keygen(1^{\st\lambda})\rightarrow(sk_{\st \re}, pk_{\st \re})$. 

%
\item publish $pk_{\st \re}$.
 \end{enumerate}
\item\underline{\re-side Query Generation:} 
$\mathtt{GenQuery}(pk_{\st \re}, n, s)\rightarrow (q, sp, \bm{b}')$

This phase involves \re. 
\begin{enumerate}
\item call $\mathtt{\re.GenQuery}(pk,n, s)\rightarrow (q, sp)$. 
\item construct a vector $\bm{b}=[b_{\st 0},...,b_{\st n-1}]$, as: 

\begin{enumerate}
\item set every element $b_{\st i}$ to zero except for $s$-th element $b_{\st s}$ which is set to $1$.
\item encrypt each element of $\bm{b}$ using additive homomorphic encryption, $\forall 0\leq i\leq n-1: b'_{\st i}=\enc(pk_{\st \re}, b_{\st i})$. 
Let $\bm{b}'$ be the vector of the encrypted elements. 
\end{enumerate}
 \item send $q$ and $\bm{b}'$ to \se and locally store $sp$. 
\end{enumerate}
\item\underline{\se-side Response Generation:}
$\mathtt{GenRes}(m_{\st 0},\ldots,m_{\st n-1}, pk, pk_{\st \re}, q, \bm{b}')\rightarrow res$

This phase involves \se. 
\begin{enumerate}
\item call $\mathtt{\se.GenRes}(m_{\st 0},\ldots,m_{\st n-1}, pk, q)\rightarrow res$. Let $res=[e_{\st 0},..., e_{\st n-1}]$. 
\item compress the response using vector  $\bm{b}'$ as follows. $\forall i, 0\leq i\leq n-1:$ 
    %
$$e= (e_{0}\hmul \bm{b}'[0])\hadd...\hadd  (e_{\st n-1}\hmul \bm{b}'[n-1])$$
 
\item  send $res=e$  to \re.
\end{enumerate}

\item\underline{\re-side Message Extraction.} 
$\mathtt{Retrieve}(res, q, sp, pk,sk_{\st \re}, s)\rightarrow m_{\st s}$

This phase involves \re. 

\begin{enumerate}

\item call $\dec(sk_{\st \re}, res)\rightarrow res'$. 

\item call $\mathtt{\re.Retrieve}(res', q, sp, pk, s)\rightarrow m_{\st s}$. 
\end{enumerate}
\end{enumerate}
\end{tcolorbox}
\end{center}
\vspace{-3mm}
    \caption{A compiler that turns a $1$-out-of-$n$ OT with response size $O(n)$ to  a $1$-out-of-$n$ OT with response size $O(1)$.}
    \label{fig::generic-short-res-HE}
    \vspace{-2mm}
\end{figure}

\begin{theorem}\label{theo::one-out-of-n-OT}
Let $\mathcal{F}_{\scriptscriptstyle\onenot}$ be the functionality defined above. If  
$\onenot$ is secure and additive homomorphic encryption meets IND-CPA, generic OT with constant size response (presented in Figure \ref{fig::generic-short-res-HE}) (i) securely computes $\mathcal{F}_{\scriptscriptstyle\onenot}$ in the presence of semi-honest adversaries and (ii) offers $O(1)$ response size, regarding $n$. 
\end{theorem}

Appendix \ref{theo::compiler-sec--} presents the proof of Theorem \ref{theo::one-out-of-n-OT}. 
This compiler is most beneficial in the generic database setting, where each item is a record with multiple attributes. Let the sender hold $n$ records $\{m_i\}_{i=0}^{n-1}$, where $m_i=(m_{i,1},\ldots,m_{i,u})$ has $u$ attributes. A conventional $1$-out-of-$n$ OT requires the receiver to download all $n$ encrypted records, i.e., $O(n\cdot u)$ ciphertext per query. Our compiler instead returns only the selected record, reducing the receiver download to $O(u)$. For $u>1$, this also reduces overall communication from $O(|q|+n\cdot u)$ to $O(|q|+n+u)$, since the extra $O(n)$ upload (the encrypted one-hot selection vector) is independent of $u$.

\section{Supersonic OT}\label{sec::supersonice-OT}



In this section, we introduce a $1$-out-of-$2$ OT, called ``Supersonic OT'',  which (i)
operates at high speed by eliminating the need for public-key-based cryptography, (ii) delivers a response of size $O(1)$ to the recipient, \re, (iii) ensures information-theoretic security, making it post-quantum secure, and (iv) is simple (but elegant), thus facilitating a simple analysis of its security and implementation.

\subsection{Security Definition}\label{sec::Ultra-OT-definition}

Supersonic OT involves a sender \se, a receiver \re, and a server \p. Each party might be corrupted by an independent passive adversary. The functionality $\mathcal{F}_{\scriptscriptstyle\ot}$ that Supersonic OT computes is similar to that of conventional OT, with the difference that now an additional party $\p$ is introduced, having no input and receiving no output. The functionality is defined as:    $\mathcal{F}_{\scriptscriptstyle\ot}:\big((m_{\st 0}, m_{\st 1}), \empt, s\big) \rightarrow (\empt, \empt, m_{\st s})$. 
 Next, we present a formal definition of \ot.

\begin{definition}[\ot]\label{def::ultra-OT-sec-def} Let $\mathcal{F}_{\scriptscriptstyle\ot}$ be the   OT functionality defined above. We assert that protocol $\Gamma$ realizes $\mathcal{F}_{\scriptscriptstyle\ot}$ in the presence of passive adversaries, if for  every non-uniform PPT adversary \adv
in the real model, there is a non-uniform PPT  simulator \simm  in
the ideal model, where:
%
\begin{equation}\label{equ::ultra-ot-sender-sim-}
\begin{split}
\Big\{\simm_{\st\se}\big((m_{\st 0}, m_{\st 1}), \empt\big)\Big\}_{\st m_{0}, m_{1}, s}{\equiv} \Big\{\view_{\st\se}^{ \st\Gamma}\big((m_{\st 0}, m_{\st 1}), \empt,  s\big) \Big\}_{\st m_0, m_1, s}
\end{split}
\end{equation}
\begin{equation}\label{equ::ultra-ot-server-sim-}
\begin{split}
\Big\{\simm_{\st\p}(\empt, \empt)\Big\}_{\st m_0, m_1, s}{\equiv}  \Big\{\view_{\st\p}^{\st \Gamma}\big((m_{\st 0}, m_{\st 1}), \empt, s\big) \Big\}_{\st m_0, m_1, s}
\end{split}
\end{equation}
\begin{equation}\label{equ::ultra-ot-reciever-sim-}
\begin{split}
&\Big\{\simm_{\st\re}\Big(s, \mathcal{F}_{\scriptscriptstyle\ot}\big((m_{\st 0}, m_{\st 1}), \empt,  s\big)\Big)\Big\}_{\st m_0, m_1, s}{\equiv}\\  &\Big\{\view_{\st\re}^{\st \Gamma}\big((m_{\st 0}, m_{\st 1}), \empt,   s\big) \Big\}_{\st m_0, m_1, s}
\end{split}
\end{equation}

\end{definition}

\subsection{The Protocol}
Figure \ref{fig::Ultrasonic-OT} presents Supersonic OT in detail. 
At a high level, the protocol works as follows. Initially, \re and \se agree on a pair of keys. In the query generation phase, \re splits its private index into two binary shares. It sends one share to \se and the other to \p. Given the share/query, \se encrypts every message $m_{\st i}$ (using a one-time pad) under one of the keys it agreed with \re.  
\se permutes the encrypted messages using  $\cper$ and its share. It sends the resulting pair to \p, which permutes the received pair using  $\cper$ and its share. \p sends only the first element of the resulting pair (which is a ciphertext) to \re and discards the second element of the pair. Next, \re decrypts the ciphertext and learns the message it was interested in. %


\begin{figure}[!h]
\vspace{-3mm}
\setlength{\fboxsep}{.9pt}
\begin{center}
    \begin{tcolorbox}[enhanced,  right=1mm,
    drop fuzzy shadow southwest,
    colframe=black,colback=white]
\vspace{-2mm}
\begin{enumerate}[leftmargin=1.1mm]

\item \underline{\textit{\re-side Setup:}}
$\mathtt{Setup}(1^{\st\lambda})\rightarrow (k_{\st 0}, k_{\st 1})$ 

\begin{itemize}


\item  \re picks two random keys $(k_{\st 0}, k_{\st 1}) \stackrel{\st\$}\leftarrow\{0, 1\}^{\st\sigma}$ and sends them to \se.

\end{itemize}

\item \underline{\textit{$\re$-side Query Generation:}} 
$\mathtt{GenQuery}(1^{\st \lambda}, s)\rightarrow q=(q_{\st 1}, q_{\st 2})$

\begin{enumerate}

\item split  the private index $s$ into two shares $(s_{\st 1}, s_{\st 2})$ by calling  $\ses(1^{\st \lambda}, s, 2, 2)\rightarrow (s_{\st 1}, s_{\st 2})$. 

\item send $q_{\st 1}=s_{\st 1}$ to \se and $q_{\st 2}=s_{\st 2}$ to $\p$. 

\end{enumerate}

\item \underline{\textit{$\se$-side Response Generation:}} 
$\mathtt{GenRes}(m_{\st 0}, m_{\st 1}, k_{\st 1},k_{\st 2}, q_{\st 1})\rightarrow res$

\begin{enumerate}

\item encrypt each message as follows. 
\vspace{-1.4mm}
$$\forall i, 0\leq i\leq 1: m'_{\st i}=m_{\st i}\oplus k_{\st i}$$

Let $e=(m'_{\st 0}, m'_{\st 1})$ contain the encrypted messages. 

\item permute the elements of $e$ as: $\cper(s_{\st 1}, e)\rightarrow e'$.

\item send $res=e'$ to $\p$. 

\end{enumerate}

\item \underline{\textit{$\p$-side Oblivious Filtering:}} 
$\mathtt{OblFilter}(res, q_{\st 2})\rightarrow res'$

\begin{enumerate}

\item permute the elements of $e'$ as: $\cper(s_{\st 2}, e')\rightarrow e''$.

\item\label{ultra-ot::e-double-prime} send (always) the first element in $e''$, say $res'=e''_{\st 0}$, to \re and discard the second element in $e''$. 

\end{enumerate}

\item\underline{\textit{\re-side Message Extraction:}} 
$\mathtt{Retrieve}(res',  k_{\st s})\rightarrow m_{\st s}$

\begin{itemize}
\item retrieve the final related message $m_{\st s}$ by decrypting $e''_{\st 0}$ as: $m_{\st s}=e''_{\st 0}\oplus k_{\st s}$.
\end{itemize}

\end{enumerate}
\end{tcolorbox}
\end{center}
\vspace{-4mm}
\caption{Supersonic OT.} 
\vspace{-4mm}
\label{fig::Ultrasonic-OT}
\end{figure}



\begin{theorem}\label{theo::ultra-OT}
Let $\mathcal{F}_{\scriptscriptstyle\ot}$ be the functionality defined in Section \ref{sec::Ultra-OT-definition}.  Then, Supersonic OT (presented in Figure \ref{fig::Ultrasonic-OT}) securely computes $\mathcal{F}_{\scriptscriptstyle\ot}$ in the presence of semi-honest adversaries, 
%
%
w.r.t. Definition \ref{def::ultra-OT-sec-def}. 
\end{theorem}


%
\begin{proof}[Proof sketch]
If $\re$ is corrupted, its view contains $e''_{\st 0}=m_{\st s}\oplus k_s$ where $k_s\gets\{0,1\}^{\st \sigma}$ is uniform, so a simulator picks $k\gets\{0,1\}^{\st \sigma}$ and outputs $e=m_{\st s}\oplus k$. 
If $\se$ is corrupted, its view is $(s_{\st 1}, k_{\st 0}, k_{\st 1})$; by security of $(2, 2)$ secret sharing, $s_{\st 1}$ is indistinguishable from a uniform bit, and $(k_{\st 0}, k_{\st 1})$ are uniform, so the simulator samples fresh independent values.  
If $\p$ is corrupted, its view is $(s_{\st 2}, e')$; $s_{\st 2}$ is uniform, and $e'$ is a pair of one-time-pad ciphertexts, hence indistinguishable from two uniform $\sigma$-bit strings. Also, the ordering is hidden from $P$ because the permutation depends on the unknown share $s_{\st 1}$.
Therefore, the real and ideal views are indistinguishable.
\end{proof}

We refer readers to Appendices \ref{sec::ultrasonic-ot-proof} and \ref{sec::Ultra-OT-Proof-of-Correctness} for full proof of Theorem~\ref{theo::ultra-OT} and proof of Supersonic OT's correctness, respectively.   Note that Supersonic OT is independent of the other protocols proposed in this paper; however, it targets the same proxy-mediated thin-client theme.

\section{Evaluation}\label{sec::supersonic-OT}

We have implemented Supersonic OT in C++ and evaluated its concrete runtime. The source code for the implementation is publicly available in \cite{supersonic-code}. 
We used two platforms for the experiments: (1) a MacBook Pro with an Apple M3 Pro CPU and 36 GB of RAM,
and (2) a Raspberry Pi 4 with 4 GB RAM. 
%
%
We did not use parallelization or any other optimization. Each experiment was repeated 50 times and we report the mean. We utilized the GMP library \cite{gmp} for big-integer arithmetic. 

\subsubsection{Supersonic OT Runtime}

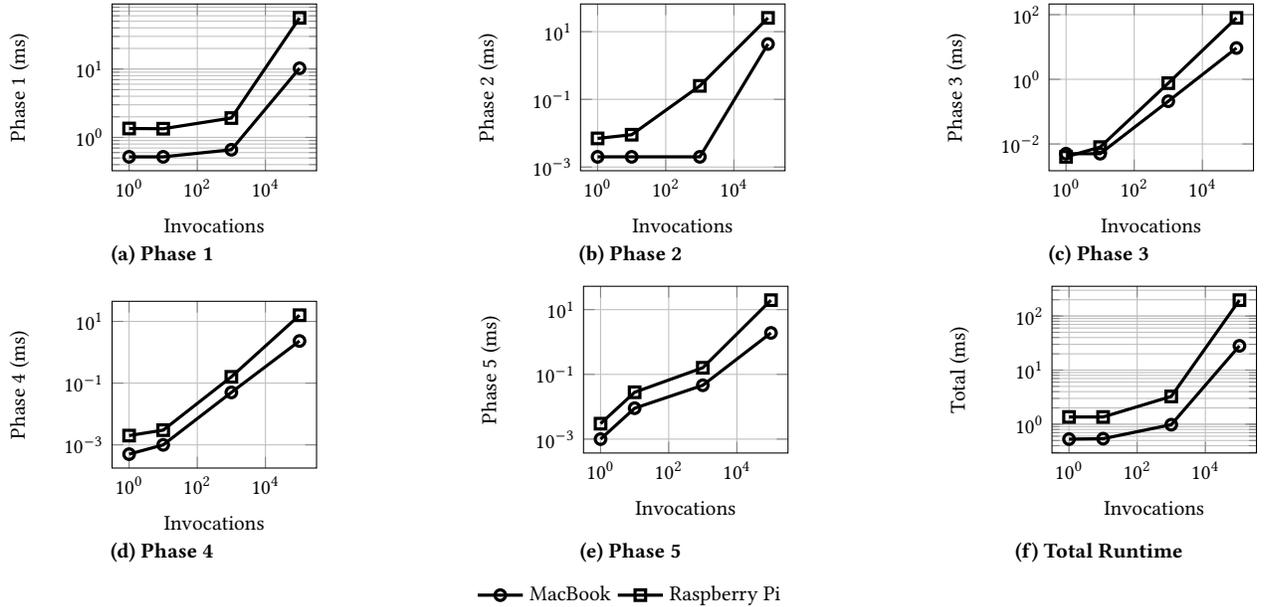
\begin{figure*}[!t]
\centering

\begin{subfigure}{0.30\textwidth}
\centering
\begin{tikzpicture}
\begin{loglogaxis}[
    width=4.3cm, height=3.8cm,
    xlabel={Invocations}, ylabel={Phase 1 (ms)},
    tick label style={font=\small},
    label style={font=\small},
    grid=both,
    legend to name=phasesLegend,
    legend columns=2,
    legend style={draw=none, font=\small},
]
\addplot[mark=o, very thick] coordinates{(1,0.52) (10,0.52) (1000,0.66) (100000,10.26)};
\addlegendentry{MacBook}
\addplot[mark=square, very thick] coordinates{(1,1.35) (10,1.34) (1000,1.92) (100000,55.96)};
\addlegendentry{Raspberry Pi}
\end{loglogaxis}
\end{tikzpicture}
\vspace{-2mm}
\caption{Phase 1}
\end{subfigure}
\hfill
\begin{subfigure}{0.30\textwidth}
\centering
\begin{tikzpicture}
\begin{loglogaxis}[
    width=4.3cm, height=3.8cm,
    xlabel={Invocations}, ylabel={Phase 2 (ms)},
    tick label style={font=\small},
    label style={font=\small},
    grid=both,
]
\addplot[mark=o, very thick] coordinates{(1,0.002) (10,0.002) (1000,0.002) (100000,4.32)};
\addplot[mark=square, very thick] coordinates{(1,0.007) (10,0.009) (1000,0.25) (100000,25.52)};
\end{loglogaxis}
\end{tikzpicture}
\vspace{-2mm}
\caption{Phase 2}
\end{subfigure}
\hfill
\begin{subfigure}{0.30\textwidth}
\centering
\begin{tikzpicture}
\begin{loglogaxis}[
    width=4.3cm, height=3.8cm,
    xlabel={Invocations}, ylabel={Phase 3 (ms)},
    tick label style={font=\small},
    label style={font=\small},
    grid=both,
]
\addplot[mark=o, very thick] coordinates{(1,0.005) (10,0.005) (1000,0.21) (100000,9.29)};
\addplot[mark=square, very thick] coordinates{(1,0.004) (10,0.008) (1000,0.76) (100000,79.43)};
\end{loglogaxis}
\end{tikzpicture}
\vspace{-2mm}
\caption{Phase 3}
\end{subfigure}

\vspace{2mm}

\begin{subfigure}{0.30\textwidth}
\centering
\begin{tikzpicture}
\begin{loglogaxis}[
    width=4.3cm, height=3.8cm,
    xlabel={Invocations}, ylabel={Phase 4 (ms)},
    tick label style={font=\small},
    label style={font=\small},
    grid=both,
]
\addplot[mark=o, very thick] coordinates{(1,0.0005) (10,0.001) (1000,0.05) (100000,2.32)};
\addplot[mark=square, very thick] coordinates{(1,0.002) (10,0.003) (1000,0.16) (100000,15.89)};
\end{loglogaxis}
\end{tikzpicture}
\vspace{-2mm}
\caption{Phase 4}
\end{subfigure}
\hfill
\begin{subfigure}{0.30\textwidth}
\centering
\begin{tikzpicture}
\begin{loglogaxis}[
    width=4.3cm, height=3.8cm,
    xlabel={Invocations}, ylabel={Phase 5 (ms)},
    tick label style={font=\small},
    label style={font=\small},
    grid=both,
]
\addplot[mark=o, very thick] coordinates{(1,0.001) (10,0.009) (1000,0.046) (100000,1.91)};
\addplot[mark=square, very thick] coordinates{(1,0.003) (10,0.028) (1000,0.16) (100000,19.58)};
\end{loglogaxis}
\end{tikzpicture}
\caption{Phase 5}
\end{subfigure}
\hfill
\begin{subfigure}{0.30\textwidth}
\centering
\begin{tikzpicture}
\begin{loglogaxis}[
    width=4.3cm, height=3.8cm,
    xlabel={Invocations}, ylabel={Total (ms)},
    tick label style={font=\small},
    label style={font=\small},
    grid=both,
]
\addplot[mark=o, very thick] coordinates{(1,0.53) (10,0.54) (1000,0.98) (100000,28.11)};
\addplot[mark=square, very thick] coordinates{(1,1.36) (10,1.36) (1000,3.26) (100000,196.39)};
\end{loglogaxis}
\end{tikzpicture}
\caption{Total Runtime}
\end{subfigure}

\vspace{1mm}
\ref{phasesLegend}

\vspace{-3mm}
\caption{Supersonic OT runtime breakdown across phases on MacBook and Raspberry Pi.}\label{fig::phases-Macbook-vs-Pi}
\vspace{-3mm}
\end{figure*}
We analyzed the runtime of various phases of Supersonic OT across different invocation frequencies ($1, 10, 10^{\st 3}$, and $10^{\st 5}$ times) and different platforms. We successfully ran the entire end-to-end Supersonic OT on the Raspberry Pi. 
Table~\ref{table-supersonic-runtime-phases} shows the performance of Supersonic OT. 


\begin{table}[!htbp]
\caption{Supersonic OT's runtime (in ms), categorized by various phases and execution environment.}\label{table-supersonic-runtime-phases}
 \vspace{-2mm}
\scalebox{1}{
\renewcommand{\arraystretch}{1}

\begin{tabular}{|c|c|c|c|c|c|c|c|c|c|c|c|} 
   \hline

&\multirow{2}{*} {  Phases}&

 \multicolumn{4}{c|}{  Number of OT Invocations}\\

            \cline{3-6}
 & &\cellcolor{white!20} $1$&\cellcolor{white!20}  $10$&\cellcolor{white!20} $10^{\st 3}$&\cellcolor{white!20} $10^{\st 5}$\\
            
    \cline{1-2}

    \hline

\multirow{6}{*} {  \rotatebox[origin=c]{90}{MacBook}} &\cellcolor{white!20}  Phase 1&
   0.52   &
      0.52  &
        0.66      &
           10.26              \\
    
       \cline{2-6}

&\cellcolor{white!20}  Phase 2&    
0.0020&    
0.0020&    
 0.0020&    
4.32       \\ 

   \cline{2-6}
     
&\cellcolor{white!20}  Phase 3&    
 0.0050&     
 0.0050&     
 0.21 &    
   9.29    
    \\ 

       \cline{    2-6}
          
&\cellcolor{white!20}  Phase 4&    
0.0005&    
0.0010&     
0.05&    
2.32 
\\ 

       \cline{    2-6}

&\cellcolor{white!20}  Phase 5&     
0.0010&     
0.0090&     
 0.046&    
  1.91     
   \\ 

       \cline{    2-6}
        \cline{    2-6}

&\cellcolor{white!20}  Total&    
0.53&    
0.54&    
 0.98&   
    28.11  
  \\ 

    \hline
    
    
    \multirow{6}{*} {  \rotatebox[origin=c]{90}{Raspberry Pi}} &\cellcolor{white!20}  Phase 1&
   1.35   &
     1.34   &
       1.92       &
        55.96                  \\
    
       \cline{2-6}

&\cellcolor{white!20}  Phase 2&    
 0.0070&    
 0.0090&    
 0.25&    
   25.52    \\ 

   \cline{2-6}
     
&\cellcolor{white!20}  Phase 3&    
 0.0040&     
 0.0080&     
 0.76  &    
   79.43       
    \\ 

       \cline{    2-6}
          
&\cellcolor{white!20}  Phase 4&    
 0.0020&    
 0.0030&     
 0.16&    
   15.89
\\ 

       \cline{    2-6}

&\cellcolor{white!20}  Phase 5&     
 0.0030&     
 0.028&     
  0.16&    
         19.58
   \\ 

       \cline{    2-6}
        \cline{    2-6}

&\cellcolor{white!20}  Total&    
1.36 &    
1.36 &    
 3.26 &   
   196.39    
  \\ 

    \hline

\end{tabular}
}
\vspace{-3mm}
\end{table}

A single invocation completes in 0.53 ms on the MacBook and 1.36 ms on the Raspberry Pi. Even under large batches of $10^{\st 5}$ invocations, the runtime remains modest: 28.11~ms on the MacBook and 196.39~ms on the Raspberry Pi. These results suggest that the protocol remains lightweight even on constrained hardware.

Figure~\ref{fig::phases-Macbook-vs-Pi}  shows the scaling behaviour of the five phases and the total runtime. The plots reveal a consistent pattern across both platforms. Across the measured invocation counts, the Raspberry Pi shows a slowdown of about $3\text{--}4\times$ compared to the MacBook; however, the relative behaviour of all phases remains consistent on both devices. 
The dominant cost on both platforms is Phase 1, which performs the key-related setup and share operations. The remaining phases incur minimal overhead, and their contribution stays small even at $10^{\st 5}$ invocations. 
%



\subsubsection{Runtime Comparison} We compared Supersonic OT's runtime to the runtimes of efficient: a base OT (i.e., Simplest OT)~\cite{ChouO15} and an OT extension (i.e., IKPN OT)~\cite{IshaiKNP03}. We report only the MacBook runtimes because we were unable to successfully build the LibOTe library~\cite{libote} (which contains the implementations of~\cite{ChouO15,IshaiKNP03}) on the Raspberry Pi despite extensive installation attempts. 
The source code implementaiton of ~\cite{IshaiKNP03} is available at~\cite{Literature-OT-code}.

\noindent\underline{\textit{Supersonic Versus Base OTs.}}  
For completeness, in addition to the runtime of Simplest OT, we include the reported runtimes of previous base OT protocols in~\cite{Efficient-OT-Naor,AsharovL0Z13}; these protocols'  source code was available but could not be successfully executed in our MacBook. For the runtime of STD--OT in  \cite{AsharovL0Z13}, STD--OT in \cite{Efficient-OT-Naor}, and RO--OT in \cite{Efficient-OT-Naor}, we derived the reported figures from \cite{AsharovL0Z13}, specifically from Table~3, where the GMP library was employed. 
Table \ref{fig::supersonic-vs-base-OTs} summarizes the results. 
As shown, Supersonic OT achieves a substantial speedup, running at least $92.8\times$  faster (than Simplest OT of  \cite{ChouO15}). 


\begin{table}[!h]
\vspace{-1mm}
\caption{ Comparing the runtime (in ms) of Supersonic OT with the base: standard OT (STD--OT)~\cite{AsharovL0Z13}, STD--OT~\cite{Efficient-OT-Naor}, and the random oracle OT (RO--OT)~\cite{Efficient-OT-Naor}. The runtime is based on 128 calls of each scheme.  The speedup factor is the improvement that Supersonic OT offers compared to
each scheme. The protocols tagged with \textsuperscript{\ddag} were run on the same machine. Other values are reported from the related original papers.}  \label{fig::supersonic-cost-breakdown-by-phases}  
\label{fig::supersonic-vs-base-OTs}
 \vspace{-1.5mm}
\begin{center}
\scalebox{1}{
\renewcommand{\arraystretch}{1}

\begin{tabular}{|c|c|c|c|c|c|c|c|c|c|c|c|} 
    \cline{1-3}

   Scheme&   Runtime (ms)&   Speedup Factor\\

       \hline

\cellcolor{white!20}  STD--OT in \cite{AsharovL0Z13}&   1,217  &     1,622\\
    
        \hline

\cellcolor{white!20}  STD--OT in \cite{Efficient-OT-Naor}&     1,681&      2,241\\ 

        \hline
     
\cellcolor{white!20}  RO--OT in \cite{Efficient-OT-Naor}&   288 &    384\\

       \hline

        Simplest OT\textsuperscript{\ddag}~\cite{ChouO15}&    69.6 &92.8   \\

       \hline

\cellcolor{white!20}  Supersonic OT\textsuperscript{\ddag}&     \textcolor{blue}{\textbf{0.75}}&      1\\ 

    \hline

\end{tabular}
}
\end{center}
 \vspace{-2mm}
\end{table}

\noindent\underline{\textit{Supersonic OT Versus OT Extensions.}} OT extensions are efficient when they are invoked many times. To determine how Supersonic OT performs compared to the  OT extensions, we compared its runtime with one of the most efficient OT extensions:  IKPN OT~\cite{IshaiKNP03}. We instantiated it using the Simplest OT as a base. We invoked the Supersonic OT and IKPN OT from 200 to 100,000 times. The source code implementaiton of~\cite{IshaiKNP03} is available at~\cite{Literature-OT-code}. 
Our implementation executes Supersonic OT sequentially without employing parallelization techniques. This contrasts with optimized libraries like libOTe, which leverage parallel execution (particularly for base OT operations) to achieve higher throughput. Thus, our sequential approach provides a conservative baseline for performance comparison. Table~\ref{table::supersonic-vs-IKNP-OT} and Figure~\ref{chart::supersonic-vs-IKPN} summarize the results.

Supersonic OT achieves notable speedups across all invocation counts, ranging from $106\times$ faster for 200 invocations to $2.6\times$ faster for 100,000 invocations. As shown in Figure~\ref{chart::supersonic-vs-IKPN}, the IKNP-based OT extension exhibits the expected near-constant runtime, yet Supersonic OT retains better performance across the entire tested range. The converging trend suggests that for a large number of invocations beyond those tested, the IKNP OT's minimal marginal cost may allow it to outperform our Supersonic OT implementation. Our experiment results suggest a crossover beyond $10^{\st 5}$ (i.e., $\approx5\times 10^{\st 5}$), after which IKNP OT outperforms. 


\begin{table}[!h]
\vspace{-1.5mm}
\caption{ End-to-end runtime of Supersonic OT versus the IKNP OT extension across different invocation counts. It shows the speedup achieved by Supersonic OT.}\label{table::supersonic-vs-IKNP-OT}
 \vspace{-2mm}
\begin{center}
\scalebox{1}{
\renewcommand{\arraystretch}{1}

\begin{tabular}{|c|c|c|c|c|c|c|c|c|c|c|c|} 
       \hline

\multirow{-1}{*}{Number of}&\multirow{-1}{*} {Supersonic} &\multirow{-1}{*}{IKNP OT} &\multirow{-1}{*}{Speedup}\\

Invocations& OT (ms)&  Extension (ms) &Factor \\

       \hline

200
     &  0.65
     &    70.05
     &\textcolor{blue}{\textbf{106}}
   \\
    
        \hline

   1,000	
     &  1.08
     &   70.46
     &65
     \\

       \hline

  4,500	
     &  2.24
     &   70.59
     &31
      \\

       \hline

20,000
     &  8.50
     &   71.89
     &8.4
      \\ 
      
       \hline
   100,000
          &  28.11
     &   74.41
     &2.6
      \\ 
      
             \hline
   500,000
          &  83.25 
     &   74.59 
     & -1.1
      \\

    \hline

\end{tabular}
}
\end{center}
 \vspace{-2mm}
\end{table}




\begin{figure}[t]
\centering
\begin{tikzpicture}

\begin{loglogaxis}[
    width=4.5cm,
    height=4cm,
    xlabel={Number of invocations},
    ylabel={Runtime (ms)},
    legend style={
        at={(0.5,1.15)},
        anchor=south,
        cells={anchor=west},
        draw=none,
        fill=none
    },
    tick label style={font=\small},
    label style={font=\small},
    grid=both,
]

\addplot[
    mark=o,
    very thick,
]
coordinates{
    (200, 0.658)
    (1000, 1.089)
    (4500, 2.247)
    (20000, 8.507)
    (100000, 28.119)
       (500000, 83.257)
};
\addlegendentry{Supersonic OT}

\addplot[
    mark=square,
    very thick,
]
coordinates{
    (200, 70.05)
    (1000, 70.4624)
    (4500, 70.5912)
    (20000, 71.89)
    (100000, 74.4135)
    (500000, 74.59)
};
\addlegendentry{IKNP OT extension~\cite{IshaiKNP03}}

\end{loglogaxis}
\end{tikzpicture}
\vspace{-3mm}
\caption{Runtime comparison between Supersonic OT and an IKNP-based OT extension across different invocation counts.}\label{chart::supersonic-vs-IKPN}
\vspace{-3mm}
\end{figure}
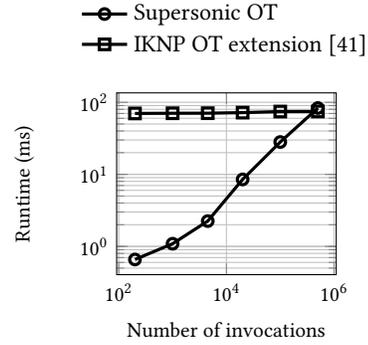

%
%
%
%


\subsubsection{Features Comparison} 

For the base OTs or OT extensions to achieve information-theoretical security, they typically require multiple replicas of the database, a noisy channel, or the involvement of a fully trusted initializer, all of which contribute to increased deployment costs or stronger trust assumptions. In contrast, Supersonic OT attains information-theoretic security without relying on such settings or assumptions. Unlike base OTs or OT extensions that typically only involve the sender and receiver, Supersonic OT involves a proxy  that might be corrupted by a semi-honest adversary who does not collude with other parties. Table~\ref{fig::features-comparison}  compares the features of Supersonic OT with several state-of-the-art OTs.


\begin{table}[!htp]
\caption{Feature comparison of 1-out-of-2 OTs.  I.T. stands for information-theoretic, and P.Q. stands for post-quantum.}  
\label{fig::features-comparison}
\vspace{-2mm}

\begin{center}
\scalebox{0.8}{
\renewcommand{\arraystretch}{1}

\begin{tabular}{|c|c|c|c|c|c|c|c|c|c|c|c|} 
    \hline


     \multirow{2}{*} {Protocol}&\multirow{-1}{*} {I.T.}&\multirow{-1}{*}{P.Q.}&\multirow{-1}{*}{No database}&\multirow{-1}{*}{No noisy}&\multirow{-1}{*}{No trusted} \\
     & security&security&replications&channel& party\\

\cline{3-4}

       \hline

\cellcolor{white!20}    \cite{AsharovL0Z13}
&  \textcolor{red}{$\times$}    
&    \textcolor{red}{$\times$}
&    \textcolor{  black}{$\checkmark$}
&    \textcolor{  black}{$\checkmark$}
&    \textcolor{  black}{$\checkmark$} \\
    
        \hline

 \cite{Efficient-OT-Naor}
 & \textcolor{red}{$\times$}    
 &   \textcolor{red}{$\times$}
 &   \textcolor{  black}{$\checkmark$}
 &   \textcolor{  black}{$\checkmark$}&   \textcolor{  black}{$\checkmark$}  \\ 

     

   \hline
   
 \cite{ChouO15}  
 & \textcolor{red}{$\times$}    
 &   \textcolor{red}{$\times$}
 &   \textcolor{  black}{$\checkmark$}
 &   \textcolor{  black}{$\checkmark$}
 &   \textcolor{  black}{$\checkmark$}\\ 
   
      \hline
   

    \cellcolor{white!20} \cite{NaorP00}   
    & \textcolor{  black}{$\checkmark$}    
    &   \textcolor{  black}{$\checkmark$}
    &   \textcolor{red}{$\times$} 
    &   \textcolor{  black}{$\checkmark$}
    &   \textcolor{  black}{$\checkmark$}\\ 

 \hline


        \cellcolor{white!20}   \cite{CrepeauK88} 
        & \textcolor{  black}{$\checkmark$}    
        &   \textcolor{  black}{$\checkmark$}
        &   \textcolor{  black}{$\checkmark$} 
        &   \textcolor{red}{$\times$}
        &   \textcolor{  black}{$\checkmark$}\\

    \hline

     \cellcolor{white!20}  \cite{CrepeauMW04}  
     &  \textcolor{  black}{$\checkmark$}    
     &    \textcolor{  black}{$\checkmark$}
     &    \textcolor{  black}{$\checkmark$} 
     &    \textcolor{red}{$\times$}
     &    \textcolor{  black}{$\checkmark$}\\

    \hline

     \cellcolor{white!20}\cite{IshaiKOPSW11}
     & \textcolor{  black}{$\checkmark$}    
     &   \textcolor{  black}{$\checkmark$}
     &   \textcolor{  black}{$\checkmark$} 
     &   \textcolor{red}{$\times$}
     &   \textcolor{  black}{$\checkmark$}\\

    \hline

         \cellcolor{white!20}\cite{rivest1999unconditionally}
         &  \textcolor{  black}{$\checkmark$}    
         &    \textcolor{  black}{$\checkmark$}
         &    \textcolor{  black}{$\checkmark$} 
         &    \textcolor{  black}{$\checkmark$}
         &    \textcolor{red}{$\times$}\\

    \hline





        \cellcolor{white!20}\cite{PeikertVW08}
        &  \textcolor{red}{$\times$}    
        &    \textcolor{  black}{$\checkmark$}
        &    \textcolor{  black}{$\checkmark$} 
        &    \textcolor{  black}{$\checkmark$}
        &    \textcolor{  black}{$\checkmark$}\\

      \hline
        \cellcolor{white!20}\cite{kundu20201}
        & \textcolor{red}{$\times$}    
        &   \textcolor{  black}{$\checkmark$}
        &   \textcolor{  black}{$\checkmark$} 
        &   \textcolor{  black}{$\checkmark$}
        &   \textcolor{  black}{$\checkmark$} \\

       \hline
       \cellcolor{white!20}\cite{DowsleyGMN12}
       &  \textcolor{red}{$\times$}    
       &    \textcolor{  black}{$\checkmark$}
       &    \textcolor{  black}{$\checkmark$} 
       &    \textcolor{  black}{$\checkmark$}
       &    \textcolor{  black}{$\checkmark$}\\

      \hline
        \cellcolor{white!20}\cite{BarretoOB18}
        & \textcolor{red}{$\times$}    
        &   \textcolor{  black}{$\checkmark$}
        &   \textcolor{  black}{$\checkmark$} 
        &   \textcolor{  black}{$\checkmark$}
        &   \textcolor{  black}{$\checkmark$}\\



    \hline
    
      Supersonic OT
      &  \textcolor{  black}{$\checkmark$}    
      &    \textcolor{  black}{$\checkmark$}
      &    \textcolor{  black}{$\checkmark$}
      &    \textcolor{  black}{$\checkmark$}
      &    \textcolor{  black}{$\checkmark$}\\ 


\hline

\end{tabular}
}
\end{center}
\vspace{-1mm}
\end{table}


\section{Conclusion and Future Work}

This work introduced Oblivis, a framework of Oblivious Transfer protocols that address gaps in delegated, metadata-hiding, and lightweight retrieval settings. We presented new OT variants that support secure query outsourcing, hidden-query execution, and multi-receiver privacy in merged cloud datasets, tasks that cannot be handled by classical OT. We further developed a generic compiler for achieving constant-size responses without imposing large local storage requirements. 
Our performance-oriented component, Supersonic OT, shows that information-theoretic OT can be both practical and fast. Our implementation achieves notable speedups over state-of-the-art base OTs and OT extensions. It operates smoothly even on constrained hardware such as a Raspberry Pi 4. Because it depends only on the GMP library, the implementation is lightweight and portable, which enables straightforward deployment across a wide range of environments.

Multiple directions remain open. Extending Supersonic OT to richer query classes (e.g., range predicates) is an important next step. Another direction is exploring browser-based deployments using lightweight client-side cryptography. Combining our delegated OT mechanisms with federated or multi-service architectures, or with edge-computing platforms, also appears promising.

\begin{acks}
Aydin Abadi was supported in part by (a) UKRI grant: EP/V011189/1 and Kuwait Foundation for the Advancement of Sciences (KFAS) under project code: PA24-6TE-2493. Yvo Desmedt thanks the Jonsson Endowment.  AI-based writing assistants were used only for language and presentation editing. 

\end{acks}

\bibliographystyle{ACM-Reference-Format}
\bibliography{ref}

\appendix



\section{Security Model}\label{sec::sec-model}

In this paper, we rely on the simulation-based model of secure multi-party computation \cite{GoldwasserMR85,DBLP:books/cu/Goldreich2004} to define and prove the proposed protocols. Below, we restate the formal security definition within this model.

 \subsubsection{Two-party Computation} A two-party protocol $\Gamma$ problem is captured by specifying a random process that maps pairs of inputs to pairs of outputs, one for each party. Such process is referred to as a functionality denoted by  $f:\{0,1\}^{\st  *}\times\{0,1\}^{\st  *}\rightarrow\{0,1\}^{ \st *}\times\{0,1\}^{ \st *}$, where $f:=(f_{\st  1},f_{\st  2})$. For every input pair $(x,y)$, the output pair is a random variable $(f_{\st  1} (x,y), f_{\st  2} (x,y))$, such that the party with input $x$ wishes to obtain $f_{\st  1} (x,y)$ while the party with input $y$ wishes to receive $f_{\st  2} (x,y)$. 
%
 In the setting where $f$ is asymmetric and only one party (say the first one) receives the result, $f$ is defined as $f:=(f_{\st  1}(x,y), \empt)$.

 \subsubsection{Security in the Presence of Passive Adversaries}  In the passive adversarial model, the party corrupted by such an adversary correctly follows the protocol specification. Nonetheless, the adversary obtains the internal state of the corrupted party, including the transcript of all the messages received, and tries to use this to learn information that should remain private. Loosely speaking, a protocol is secure if whatever can be computed by a party in the protocol can be computed using its input and output only. In the simulation-based model, it is required that a party’s view in a protocol's 
 execution can be simulated given only its input and output. This implies that the parties learn nothing from the protocol's execution. More formally, party $i$’s view (during the execution of $\Gamma$) on input pair  $(x, y)$ is denoted by $\mathsf{View}_{\st  i}^{\st  \Gamma}(x,y)$ and equals $(w, r_{\st  i}, m_{\st  1}^{\st  i}, ..., m_{\st  t}^{\st  i})$, where $w\in\{x,y\}$ is the input of $i^{\st  th}$ party, $r_{\st  i}$ is the outcome of this party's internal random coin tosses, and $m_{\st  j}^{\st  i}$ represents the $j^{\st  th}$ message this party receives.  The output of the $i^{\st  th}$ party during the execution of $\Gamma$ on $(x, y)$ is denoted by $\mathsf{Output}_{\st  i}^{\st  \Gamma}(x,y)$ and can be generated from its own view of the execution.  
\vspace{-1mm}
\begin{definition}
Let $f$ be the deterministic functionality defined above. Protocol $\Gamma$ securely computes $f$ in the presence of a  passive adversary if there exist polynomial-time algorithms $(\mathsf {Sim}_{\st  1}, \mathsf {Sim}_{\st  2})$ such that:
\end{definition}
\vspace{-1mm}
  \begin{equation*}
  \{\mathsf {Sim}_{\st 1}(x,f_{\st 1}(x,y))\}_{\st x,y}\stackrel{c}{\equiv} \{\mathsf{View}_{\st 1}^{\st \Gamma}(x,y) \}_{\st x,y}
  \end{equation*}
  \begin{equation*}
    \{\mathsf{Sim}_{\st 2}(y,f_{\st 2}(x,y))\}_{ \st x,y}\stackrel{c}{\equiv} \{\mathsf{View}_{\st 2}^{\st \Gamma}(x,y) \}_{\st x,y}
  \end{equation*}



\section{DQ-OT's Security Proof}\label{sec::DQ-OT-proof}
Below, we prove  DQ-OT's security, i.e., Theorem \ref{theo::DQ-OT-sec}.

\begin{proof}
We consider the case where each party is corrupt.

\subsubsection{Corrupt Receiver \re} In the real execution, \re's view is: 

$\view_{\re}^{\st DQ\text{-}OT}\big(m_{\st 0}, m_{\st 1}, \empt,$ $  \empt, s\big) = \{r_{\st\re}, g, C, p, e_{\st 0}, e_{\st 1}, m_{\st s} \}$, where $g$ is a random generator, $C=g^{\st a}$ is a random public parameter, $a$ is a random value, $p$ is a large random prime number, and $r_{\st \re}$ is the outcome of the internal random coin of  \re and is used to generate $(r_{\st 1}, r_{\st 2})$.  
 Below, we construct an idea-model simulator  $\simm_{\st\re}$ which receives $(s, m_{\st s})$ from \re. 
 
 \begin{enumerate}
 \item initiates an empty view and appends uniformly random coin $r'_{\st\re}$ to it, where $r'_{\st\re}$ will be used to generate \re-side randomness. 
It chooses a large random prime number $p$ and a random generator $g$.


 %
 \item sets $(e'_{\st 0}, e'_{\st 1})$ as follows: 
 \begin{itemize}
 \item splits $s$ into two shares: $\ses(1^{\st\lambda}, s, 2, 2)\rightarrow (s'_{\st 1}, s'_{\st 2})$.
 \item picks uniformly random values: $C', r'_{\st 1}, r'_{\st 2},y'_{\st 0}, y'_{\st 1}\stackrel{\st\$}\leftarrow\mathbb{Z}_{\st p}$.
 %
 %
 \item sets $\beta'_s=g^{\st x}$, where $x$ is set as follows: 
 \begin{itemize}
 \item[$*$] $x = r'_{\st 2} + r'_{\st 1}$, if $(s = s_{\st 1} = s_{\st 2} = 0)$ or $(s = s_{\st 1} = 1\wedge s_{\st 2} = 0)$.
 \item[$*$] $x = r'_{\st 2} - r'_{\st 1}$, if $(s =0 \wedge s_{\st 1} = s_{\st 2} = 1)$ or $(s = s_{\st 2} = 1\wedge s_{\st 1} = 0)$.
 \end{itemize}
 \item picks a uniformly random value $u\stackrel{\st\$}\leftarrow\mathbb{Z}_{p}$ and then sets $e'_{\st s} = (g^{\st y'_s}, \h(\beta'^{\st y'_s}_{\st s})\oplus m_{\st s})$ and $e'_{\st 1-s} = (g^{\st y'_{\st 1-s}}, u)$. 
 \end{itemize}
 \item  appends $(g, C', p, r'_{\st 1}, r'_{\st 2},  e'_{\st 0}, e'_{\st 1}, m_{\st s})$ to the view and outputs the view.
 \end{enumerate}
 
 Now we discuss why the two views in the ideal and real models are indistinguishable. Since we are in the semi-honest model, the adversary picks its randomness according to the protocol description; thus, $r_{\st \re}$ and $r'_{\st \re}$ model have identical distributions, so do values $(r_{\st 1}, r_{\st 2})$ in the real model and $(r'_{\st 1}, r'_{\st 2})$ in the ideal model. Also, $C$ and $C'$ have been picked uniformly at random and have identical distributions. The same applies to values $g$ and $p$ in the real and ideal models.

 Next, we argue that  $e_{\st 1-s}$ in the real model and $e'_{\st 1-s}$ in the ideal model are indistinguishable. In the real model, it holds that $e_{1-s}=(g^{\st y_{\st 1-s}}, \h(\beta^{y_{1-s}}_{\st 1-s})\oplus m_{\st 1-s})$, where $\beta^{\st y_{\st 1-s}}_{\st 1-s}=\frac{C}{g^{\st x}}=g^{\st a-x}$. Since $y_{\st 1-s}$ in the real model and $y'_{\st 1-s}$ in the ideal model have been picked uniformly at random and unknown to the adversary/distinguisher, $g^{\st y_{\st 1-s}}$ and $g^{\st y'_{\st 1-s}}$  have identical distributions.

 Furthermore, in the real model, given $C=g^{\st a}$, due to the DL problem, $a$ cannot be computed by a PPT adversary. Also, due to CDH assumption, \re cannot compute  $\beta^{y_{1-s}}_{1-s}$ (i.e., the input of $\h$), given $g^{\st y_{\st 1-s}}$ and $g^{\st a-x}$. We also know that $\h$ is modelled as a random oracle and its output is indistinguishable from a random value. Thus,  $\h(\beta^{\st y_{\st 1-s}}_{\st 1-s})\oplus m_{\st 1-s}$ in the real model and $u$ in the ideal model are indistinguishable. This means that $e_{\st 1-s}$ and $e'_{\st 1-s}$ are indistinguishable too, due to DL, CDH, and RO assumptions. Also, in the real and ideal models, $e_{\st s}$ and $e'_{\st s}$ have been defined over $\mathbb{Z}_{\st p}$ and their decryption always result in the same value $m_{\st s}$. Thus, $e_{\st s}$ and $e'_{\st s}$ have identical distributions too. 
Also, $m_{\st s}$ has an identical distribution in both models. 

We conclude that the two views are computationally indistinguishable, i.e., Relation \ref{equ::reciever-sim-} (in Section \ref{sec::sec-def}) holds.

\subsubsection{Corrupt Sender \se} In the real model, \se's view is: 

$\view_{\se}^{\st DQ\text{-}OT}\big((m_{\st 0}, m_{\st 1}), \empt,  \empt, $ $s\big)= \{r_{\st\se}, C, \beta_{\st 0}, \beta_{\st 1}\}$, where $r_{\st \se}$ is the outcome of the internal random coin of  \se.  Next, we construct an idea-model simulator  $\simm_{\st\se}$ which receives $(m_{\st 0}, m_{\st 1})$ from \se. 

\begin{enumerate}
\item constructs an empty view and appends uniformly random coin $r'_{\st\se}$ to it, where $r'_{\st\se}$ will be used to generate random values for \se.
%

\item picks random values $C', r'\stackrel{\st\$}\leftarrow\mathbb{Z}_{\st p}$.

\item sets $\beta'_{\st 0}=g^{\st r'}$ and $\beta'_{\st 1}=\frac{C'}{g^{\st r'}}$. 

\item appends $\beta'_{\st 0}$ and $\beta'_{\st 1}$ to the view and outputs the view. 

\end{enumerate}

Next, we explain why the two views in the ideal and real models are indistinguishable. Recall, in the real model, $(\beta_{\st s}, \beta_{\st 1-s})$ have the following form: $\beta_{s}=g^{\st x}$ and  $\beta_{\st 1-s}=g^{\st a-x}$, where $a=DL(C)$ and $C=g^{\st a}$. In this model, because $a$ and $x$ have been picked uniformly at random and unknown to the adversary, due to the DL assumption,  $\beta_{\st s}$ and $\beta_{\st 1-s}$ have identical distributions and are indistinguishable. In the ideal model, $r'$ has been picked uniformly at random and we know that $a'$ in $C' = g^{\st a'}$ is a uniformly random value, unknown to the adversary; therefore, due to DL assumption, $\beta'_{\st 0}$ and $\beta'_{\st 1}$ have identical distributions too. Moreover, values $\beta_{\st s}, \beta_{\st 1-s}, \beta'_{\st 0}$, and $\beta'_{\st 1}$  have been defined over the same field, $\mathbb{Z}_{\st p}$. Thus, they have identical distributions and are indistinguishable. 

Therefore, the two views are computationally indistinguishable, i.e., Relation \ref{equ::sender-sim-} (in Section \ref{sec::sec-def}) holds.

\subsubsection{Corrupt Server $\p_{\st 2}$}  
 
 In the real execution, $\p_{\st 2}$'s view is: 
 
$\view_{\p_2}^{\st DQ\text{-}OT}\big((m_{\st 0}, m_{\st 1}), \empt,$ $ \empt, s\big)=\{g, C, p, s_{\st 2}, r_{\st 2}\}$. Below, we show how an ideal-model simulator $\simm_{\st \p_{\st 2}}$ works. 

\begin{enumerate}
\item initiates an empty view. It selects a random generator $g$ and a large random prime number $p$.
\item picks two uniformly random values $s'_{\st 2}\stackrel{\st\$}\leftarrow\mathbb U$ and $C', r'_{\st 2}\stackrel{\st\$}\leftarrow\mathbb{Z}_{\st p}$, where $\mathbb{U}$ is the output range of $\ses$. 
\item appends $s'_{\st 2}, C'$ and $r'_{\st 2}$ to the view and outputs the view. 
\end{enumerate}

Next, we explain why the views in the ideal and real models are indistinguishable.

Since values $g$ and $p$ have been picked uniformly at random in both models, they have identical distributions in the real and ideal models. 
Since $r_{\st 2}$ and $r'_{\st 2}$ have been picked uniformly at random from  $\mathbb{Z}_{\st p}$, they have identical distributions. Also, due to the security of $\ses$ each share $s_{\st 2}$ is indistinguishable from a random value $s'_{\st 2}$, where $s'_{\st 2}\in \mathbb{U}$. Also, both $C$ and $C'$ have been picked uniformly at random from $\mathbb{Z}_{\st p}$; therefore, they have identical distribution.

Thus, the two views are computationally indistinguishable, i.e., Relation \ref{equ::server-sim-} w.r.t. $\p_{\st 2}$ (in Section \ref{sec::sec-def}) holds.

\subsubsection{Corrupt Server $\p_{\st 1}$}

In the real execution, $\p_{\st 1}$'s view is: 
 
$\view_{\st\p_1}^{\st DQ\text{-}OT}\big((m_{\st 0}, m_{\st 1}), \empt,$ $ \empt, s\big)=\{ g, C, p, s_{\st 1}, r_{\st 1}, \delta_{\st 0}, \delta_{\st 1}\}$. Ideal-model $\simm_{\st\p_{1}}$ works as follows.

\begin{enumerate}
\item initiates an empty view. It chooses a random generator $g$ and a large random prime number $p$. 
\item picks two random values $\delta'_{\st 0}, \delta'_{\st 1}\stackrel{\st\$}\leftarrow\mathbb{Z}_{\st p}$. 
\item picks two uniformly random values $s'_{\st 1}\stackrel{\st\$}\leftarrow\mathbb U$ and \\  $C', r'_{\st 1}\stackrel{\st\$}\leftarrow\mathbb{Z}_{\st p}$, where $\mathbb{U}$ is the output range of $\ses$. 
\item appends  $s'_{\st 1}, C', r'_{\st 1}, \delta'_{\st 0}, \delta'_{\st 1}$ to the view and outputs the view. 
\end{enumerate}

Now, we explain why the views in the ideal and real models are indistinguishable. Values $g$ and $p$ have been picked uniformly at random in both models. Hence, $g$ and $p$ in the real and ideal models have identical distributions (pair-wise).

Recall, in the real model, $\p_{\st 1}$ receives $\delta_{\st s_{\st 2}}=g^{\st r_{\st 2}}$ and  $\delta_{\st 1- s_2}=g^{\st a- r_2}$ from  $\p_{\st 2}$. 
Since $a$ and $r_{\st 2}$ have been picked uniformly at random and unknown to the adversary due to DL assumption,   $\delta_{\st s_2}$ and $\delta_{\st 1-{s_2}}$ (or $\delta_{\st 0}$ and $\delta_{\st 1}$) have identical distributions and are indistinguishable from random values (of the same field). 

In the ideal model, $\delta'_{\st 0}$ and  $\delta'_{\st 1}$ have been picked uniformly at random; therefore, they have identical distributions too. Moreover,  $\delta_{\st s}, \delta_{\st 1-s}, \delta'_{\st 0}$, and $\delta'_{\st 1}$  have been defined over the same field, $\mathbb{Z}_{\st p}$. So, they have identical distributions and are indistinguishable. 
Due to the security of $\ses$ each share $s_{\st 1}$ is indistinguishable from a random value $s'_{\st 1}$, where $s'_{\st 1}\in \mathbb{U}$. Also, $(r_{\st 1}, C)$ and $(r'_{\st 1}, C')$ have identical distributions, as they are picked uniformly at random from $\mathbb{Z}_{\st p}$.

Hence, the two views are computationally indistinguishable, i.e., Relation \ref{equ::server-sim-} w.r.t. $\p_{\st 1}$ (in Section \ref{sec::sec-def}) holds.
%
\end{proof}

\section{Proof of Correctness}\label{sec::DQ-OT-proof-of-correctness}

In this section, we discuss why the correctness of DQ-OT always holds. Recall, in the original OT of Naor and Pinkas \cite{Efficient-OT-Naor}, the random value $a$ (i.e., the discrete logarithm of random value $C$) is inserted by receiver \re into the query $\beta_{\st 1-s}$ while the other query $\beta_{\st s}$ is free from value $a$. As we will explain below, in our DQ-OT, the same applies to the final queries that are sent to \se.  
Briefly, in DQ-OT, when:

\begin{itemize}

\item[$\bullet$]  $s=s_{\st 1}\oplus s_{\st 2}=1$ (i.e., when $s_{\st 1}\neq s_{\st 2}$), then $a$ will always appear in $\beta_{\st 1-s}=\beta_{\st 0}$; however, $a$ will not appear in $\beta_{\st 1}$. 

\item[$\bullet$] $s=s_{\st 1}\oplus s_{\st 2} = 0$ (i.e., when $s_{\st 1}=s_{\st 2}$), then $a$ will always appear in $\beta_{\st 1-s}=\beta_{\st 1}$; but $a$ will not appear in $\beta_{\st 0}$.

\end{itemize}

This is evident in Table~\ref{tab:analysis} which shows what
$\delta_{\st i}$ and $\beta_{\st j}$ are for the different values of $s_{\st 1}$ and $s_{\st 2}$. Therefore, the query pair ($\beta_{\st 0}, \beta_{\st 1}$) has the same structure as it has in \cite{Efficient-OT-Naor}. 


\begin{table}[h]
\begin{center}
\scalebox{.76}{
  \begin{tabular}{|l||p{4.5cm}|p{4.5cm}|}
\hline
    &\multicolumn{1}{c|}{$s_{\st 2}=0$}&\multicolumn{1}{c|}{$s_{\st 2}=1$}\\
\hline
\hline
&\cellcolor{gray!20}$\delta_{\st 0} = g^{\st r_2}$, \hspace{3mm} $\delta_{\st 1} = g^{\st a-r_2}$&\cellcolor{gray!20} $\delta_{\st 0}=g^{\st a-r_2}$,\hspace{3mm} $\delta_{\st 1}=g^{\st r_2}$\\
\multirow{-2}{*}{$s_{\st 1}=0$}&\cellcolor{gray!20}$\beta_{\st 0}=g^{\st r_2+r_1}$, \hspace{3mm} $\beta_{\st 1}=g^{\st a-r_2-r_1}$&\cellcolor{gray!20}$\beta_{\st 0}=g^{\st a-r_2+r_1}$, \hspace{3mm} $\beta_{\st 1}=g^{\st r_2-r_1}$\\
\hline
&\cellcolor{gray!20}$\delta_{\st 0}=g^{\st r_2}$, \hspace{3mm} $\delta_{\st 1}=g^{\st a-r_2}$&\cellcolor{gray!20}$\delta_{\st 0}=g^{\st a-r_2}$, \hspace{3mm} $\delta_{\st 1}=g^{\st r_2}$\\
\multirow{-2}{*}{$s_{\st 1}=1$}&\cellcolor{gray!20}$\beta_{\st 0}=g^{\st a-r_2-r_1}$, \hspace{3mm} $\beta_{\st 1}=g^{\st r_2+r_1}$&\cellcolor{gray!20}$\beta_{\st 0}=g^{\st r_2-r_1}$, \hspace{3mm} $\beta_{\st 1}=g^{\st a-r_2+r_1}$\\
\hline
\end{tabular}
}
\end{center}
\caption{$\delta_{\st i}$ and $\beta_{\st j}$ are for the different values of $s_{\st 1}$ and $s_{\st 2}$.
We express each value as a power of $g$.}
\label{tab:analysis}
\end{table}

Next, we show why, in DQ-OT, \re can extract the correct message, i.e., $m_s$. Given \se's reply pair $(e_{\st 0}, e_{\st 1})$ and its original index $s$, \re knows which element to pick from the response pair, i.e., it picks $e_{\st s}$.

Moreover, given $g^{\st y_s}\in e_{\st s}$, \re can  recompute $\h(g^{\st y_s})^{\st x}$, as it knows the value of $s, s_{\st 1}$, and $s_{\st 2}$.  Specifically, as Table~\ref{tab:analysis} indicates, when: 

\begin{itemize}
\item[$\bullet$]  $\overbrace{(s= s_{\st 1}=s_{\st 2}=0)}^{\st \text{Case 
 1}}$ or $\overbrace{(s=s_{\st 1}=1 \wedge s_{\st 2}=0)}^{\st \text{Case 
 2}}$, then \re can set $x = r_{\st 2} + r_{\st 1}$.

\begin{itemize}

 \item In Case 1, it holds $\h((g^{\st y_0})^{\st x})= \h((g^{\st y_0})^{\st r_2 +r_1})=q$. Also, $e_{\st 0}=\h(\beta_{\st 0}^{\st y_0})\oplus m_{\st 0}=\h((g^{\st r_2+r_1})^{\st y_0})\oplus m_{\st 0}$. Thus, $q\oplus e_{\st 0}=m_{\st 0}$.
 
 \item In Case 2, it holds $\h((g^{\st y_1})^{\st x})= \h((g^{\st y_1})^{\st r_2 +r_1})=q$. Moreover, $e_{\st 1}=\h(\beta_{\st 1}^{\st y_1})\oplus m_{\st 1}=\h((g^{\st r_2+r_1})^{\st y_1})\oplus m_{\st 1}$. Hence, $q\oplus e_{\st 1}=m_{\st 1}$.

 \end{itemize}

 \vspace{.6mm} 

\item[$\bullet$]  $\overbrace{(s=0 \wedge s_{\st 1}=s_{\st 2}=1)}^{\st \text{Case 
 3}}$ or $\overbrace{(s=s_{\st 2}=1 \wedge s_{\st 1}=0)}^{\st \text{Case 
 4}}$, then \re can set $x = r_{\st 2} - r_{\st 1}$.

\begin{itemize}

 \item In Case 3, it holds $\h((g^{\st y_0})^{\st x})= \h((g^{\st y_0})^{\st r_2 - r_1})=q$. On the other hand, $e_{\st 0}=\h(\beta_{\st 0}^{\st y_0})\oplus m_{\st 0}=\h((g^{\st r_2 - r_1})^{\st y_0})\oplus m_{\st 0}$. Therefore, $q\oplus e_{\st 0}=m_{\st 0}$.
 
 \item In Case 4, it holds $\h((g^{\st y_1})^{\st x})= \h((g^{\st y_1})^{\st r_2 -r_1})=q$. Also, $e_{\st 1}=\h(\beta_{\st 1}^{\st y_1})\oplus m_{\st 1}=\h((g^{\st r_2-r_1})^{\st y_1})\oplus m_{\st 1}$. Hence, $q\oplus e_{\st 1}=m_{\st 1}$.
 
 \end{itemize}

\end{itemize}

We conclude that DQ-OT always allows honest \re to recover the message of its interest, i.e., $m_s$.

\section{DUQ-OT's Security Proof}\label{sec::DUQ-OT-Security-Proof}
Below, we prove  DUQ-OT's security theorem, i.e., Theorem \ref{theo::DUQ-OT-sec}.  Even though the proofs of DUQ-OT and DQ-OT have similarities, they have significant differences too; thus, for the sake of completeness, we present a complete proof for DUQ-OT.

\begin{proof}
We consider the case where each party is corrupt at a time.

\subsubsection{Corrupt Receiver \re} In the real execution, \re's view is: 

$$\view_{\re}^{\st DUQ\text{-}OT}\big(m_{\st 0}, m_{\st 1}, \empt,  \empt, \empt, s\big) = \{r_{\st\re}, g, C, p,  r_{\st 3}, s_{\st 2}, e'_{\st 0}, e'_{\st 1}, m_{\st s} \}$$

where $r_{\st \re}$ is the outcome of the internal random coin of  \re and is used to generate $(r_{\st 1}, r_{\st 2})$.  
 Below, we construct an idea-model simulator  $\simm_{\st\re}$ which receives $m_{\st s}$ from \re. 
 
 \begin{enumerate}
 \item initiate an empty view and appends uniformly random coin $r'_{\st\re}$ to it, where $r'_{\st\re}$ will be used to generate \re-side randomness, i.e., $(r'_{\st 1}, r'_{\st 2})$. 

 \item selects a random generator $g$ and a large random prime number $p$.
 \item sets response $(\bar{e}'_{\st 0}, \bar{e}'_{\st 1})$ as follows: 
 \begin{itemize}
 %
 %
 \item picks random values: $C', r'_{\st 1}, r'_{\st 2}, y'_{\st 0}, y'_{\st 1}\stackrel{\st\$}\leftarrow\mathbb{Z}_{\st p}$, $ r'_{\st 3}\stackrel{\st\$}\leftarrow\{0, 1\}^{\st\lambda}$, $s' \stackrel{\st\$}\leftarrow\{0,1\}$, and $u\stackrel{\st\$}\leftarrow\{0, 1\}^{\st\sigma+\lambda}$. 
 
 \item sets $x=r'_{\st 2}+r'_{\st1}\cdot (-1)^{\st s'}$ and $\beta'_{\st 0}=g^{\st x}$. 
 \item sets $\bar{e}_{\st 0} = (g^{\st y'_0}, \g(\beta'^{\st y'_0}_{\st 0})\oplus (m_{\st s}|| r'_{\st 3}))$ and $\bar{e}_{\st 1} = (g^{\st y'_{1}}, u)$. 

 \item randomly permutes the element of  pair $(\bar{e}_{\st 0}, \bar{e}_{\st 1})$. Let $(\bar{e}'_{\st 0}, \bar{e}'_{\st 1})$ be the result. 

 \end{itemize}
 \item  appends $(g, C', p, r'_{\st 3}, s', \bar{e}'_{\st 0}, \bar{e}'_{\st 1}, m_{\st s})$ to the view and outputs the view.
 \end{enumerate}
 
 Next, we argue that the views in the ideal and real models are indistinguishable. 
 
 As we are in the semi-honest model, the adversary picks its randomness according to the protocol description; therefore, $r_{\st \re}$ and $r'_{\st \re}$ model have identical distributions, the same holds for values $(r_{\st 3}, s_{\st 2})$ in the real model and $(r'_{\st 3}, s')$ in the ideal model, component-wise.  
Furthermore, because values $g$ and $p$ have been selected uniformly at random in both models, they have identical distributions in the real and ideal models.

 For the sake of simplicity, in the ideal mode let $\bar{e}'_{\st j}=\bar{e}_{\st 1} = (g^{\st y'_{\st 1}}, u)$ and in the real model let $e'_{\st i}=e_{\st 1-s}=(g^{\st y_{1-s}}, \g(\beta^{\st y_{1-s}}_{\st 1-s})\oplus (m_{\st 1-s}||r_{\st 3}))$, where $i, j\in\{0,1 \}$.  We will explain that  $e'_{\st i}$ in the real model and $\bar{e}'_{\st j}$ in the ideal model are indistinguishable. 
 
 In the real model, it holds that $e_{\st 1-s}=(g^{\st y_{1-s}}, \g(\beta^{\st y_{1-s}}_{\st 1-s})\oplus (m_{\st 1-s}||r_{\st 3}))$, where $\beta^{\st y_{1-s}}_{\st 1-s}=\frac{C}{g^{\st x}}=g^{\st a-x}$. Since $y_{\st 1-s}$ in the real model and $y'_{\st 1}$ in the ideal model have been picked uniformly at random and unknown to the adversary, $g^{\st y_{1-s}}$ and $g^{\st y'_{1}}$  have identical distributions.

 Moreover, in the real model, given $C=g^{\st a}$, because of the DL problem, $a$ cannot be computed by a PPT adversary. 
 Furthermore, due to CDH assumption, \re cannot compute  $\beta^{\st y_{1-s}}_{\st 1-s}$ (i.e., the input of $\g$), given $g^{\st y_{\st 1-s}}$ and $g^{\st a-x}$. We know that $\g$ is considered a random oracle and its output is indistinguishable from a random value. Therefore,  $\g(\beta^{\st y_{\st 1-s}}_{\st 1-s})\oplus (m_{\st 1-s} || r_{\st 3})$ in the real model and $u$ in the ideal model are indistinguishable. This means that $e_{\st 1-s}$ and $\bar{e}'_{\st j}$ are indistinguishable too, due to DL, CDH, and RO assumptions.

 Moreover, since (i) $y_{\st s}$ in the real model and $y'_{\st 0}$ in the ideal model have picked uniformly at random and (ii) the decryption of both $e'_{\st 1-i}$ and $\bar{e}'_{\st 1-j}$ contain $m_{\st s}$, $e'_{\st 1-i}$ and $\bar{e}'_{\st 1-j}$ have identical distributions. $m_{\st s}$ also has identical distribution in both models.  Both $C$ and $C'$ have also been picked uniformly at random from $\mathbb{Z}_{\st p}$; therefore, they have identical distribution. 
 
 In the ideal model, $\bar{e}_{\st 0}$ always contains an encryption of the actual message $m_{\st s}$ while $\bar{e}_1$ always contains a dummy value $u$. However, in the ideal model the elements of pair  $(\bar{e}_{\st 0}, \bar{e}_{\st 1})$ and in the real model the elements of  pair $(e_{\st 0}, e_{\st 1})$ have been randomly permuted, which result in $(\bar{e}'_{\st 0}, \bar{e}'_{\st 1})$ and $(e'_{\st 0}, e'_{\st 1})$ respectively. Therefore, the permuted pairs have identical distributions too. 

We conclude that the two views are computationally indistinguishable, i.e., Relation \ref{equ::DUQ-OT-reciever-sim-} (in Section \ref{sec::DUQ-OT-definition}) holds.

\subsubsection{Corrupt Sender \se} In the real model, \se's view is: 

$$\view_{\se}^{\st DUQ\text{-}OT}\big((m_{\st 0}, m_{\st 1}), \empt,  \empt ,\empt, s\big)= \{r_{\st\se}, C, r_{\st 3}, \beta_{\st 0}, \beta_{\st 1}\}$$

where $r_{\st \se}$ is the outcome of the internal random coin of  \se.  Next, we construct an idea-model simulator  $\simm_{\se}$ which receives $\{m_{\st 0}, m_{\st 1}\}$ from \se. 

\begin{enumerate}
\item constructs an empty view and appends uniformly random coin $r'_{\st\se}$ to it, where $r'_{\st\se}$ will be used to generate random values for \se.
\item picks  random values $C', r'\stackrel{\st\$}\leftarrow\mathbb{Z}_{\st p}, r'_{\st 3}\stackrel{\st \$}\leftarrow\{0, 1\}^{\st \lambda}$.

\item sets $\beta'_{\st 0}=g^{\st r'}$ and $\beta'_{1}=\frac{C'}{g^{\st r'}}$.

\item appends $C', r'_{\st 3}, \beta'_{\st 0}$, and $\beta'_{\st 1}$ to the view and outputs the view. 

\end{enumerate}

Next, we explain why the two views in the ideal and real models are indistinguishable. Recall, in the real model, $(\beta_{\st s}, \beta_{\st 1-s})$ have the following form: $\beta_{s}=g^{\st x}$ and  $\beta_{\st 1-s}=g^{\st a-x}$, where $a=DL(C)$ and $C=g^{\st a}$. 

In this ideal model, as $a$ and $x$ have been picked uniformly at random and unknown to the adversary, due to the DL assumption,   $\beta_{\st s}$ and $\beta_{\st 1-s}$ have identical distributions and are indistinguishable. 

In the ideal model, $r'$ and $C'$ have been picked uniformly at random and we know that $a'$ in $C'=g^{\st a'}$ is a uniformly random value, unknown to the adversary; thus, due to DL assumption, $\beta'_{\st 0}$ and $\beta'_{\st 1}$ have identical distributions too. The same holds for values $C$ and $C'$. 

Moreover, values $\beta_{\st s}, \beta_{\st 1-s}, \beta'_{\st 0}$, and $\beta'_{\st 1}$  have been defined over the same field, $\mathbb{Z}_{\st p}$. Thus, they have identical distributions and are indistinguishable.  The same holds for values $r_{\st 3}$ in the real model and $r'_{\st 3}$ in the ideal model. 

Therefore, the two views are computationally indistinguishable, i.e., Relation \ref{equ::DUQ-OT-sender-sim-} (in Section \ref{sec::DUQ-OT-definition}) holds.

\subsubsection{Corrupt Server $\p_{\st 2}$}  
 
 In the real execution, $\p_{\st 2}$'s view is: 
$$\view_{\st \p_2}^{\st DUQ\text{-}OT}\big((m_{\st 0}, m_{\st 1}), \empt, \empt, \empt, s\big)=\{ g, C, p, s_{\st 2}, r_{\st 2}\}$$ 

Below, we show how an ideal-model simulator $\simm_{\st \p_{\st 2}}$ works. 

\begin{enumerate}
\item initiates an empty view. It chooses a random generator $g$ and a large random prime number $p$.

\item picks two uniformly random values $s'_{\st 2}\stackrel{\st\$}\leftarrow\mathbb U$ and $C', r'_{\st 2}\stackrel{\st\$}\leftarrow\mathbb{Z}_{\st p}$, where $\mathbb{U}$ is the output range of $\ses$. 
\item appends $s'_{\st 2}, C'$ and $r'_{\st 2}$ to the view and outputs the view. 
\end{enumerate}

Next, we explain why the views in the ideal and real models are indistinguishable. 
 Values $g$ and $p$ have been selected uniformly at random in both models. Thus, they have identical distributions in the real and ideal models.

Since $r_{\st 2}$ and $r'_{\st 2}$ have been picked uniformly at random from  $\mathbb{Z}_{\st p-1}$, they have identical distributions. 

Also, due to the security of $\ses$ each share $s_{\st 2}$ is indistinguishable from a random value $s'_{\st 2}$, where $s'_{\st 2}\in \mathbb{U}$. Also, both $C$ and $C'$ have been picked uniformly at random from $\mathbb{Z}_{\st p}$; therefore, they have identical distributions.

Thus, the two views are computationally indistinguishable, i.e., Relation \ref{equ::DUQ-OT-server-sim-} w.r.t. $\p_{\st 2}$ (in Section \ref{sec::DUQ-OT-definition}) holds.

\subsubsection{Corrupt Server $\p_{\st 1}$}

In the real execution, $\p_{\st 1}$'s view is: 
$$\view_{\st\p_1}^{\st DUQ\text{-}OT}\big((m_{\st 0}, m_{\st 1}), \empt, \empt, \empt, s\big)=\{ g, C, p, s_{\st1}, r_{\st 1}, \delta_{\st 0}, \delta_{\st 1}\}$$ Ideal-model $\simm_{\st \p_{1}}$ works as follows.

\begin{enumerate}
\item initiates an empty view. It chooses a large random prime number $p$ and a random generator $g$.
\item picks two random values $\delta'_{\st 0}, \delta'_{\st 1}\stackrel{\st\$}\leftarrow\mathbb{Z}_{\st p}$. 
\item picks two uniformly random values $s'_{\st 1}\stackrel{\st\$}\leftarrow\mathbb U$ and $C', r'_{\st 1}\stackrel{\st\$}\leftarrow\mathbb{Z}_{\st p}$, where $\mathbb{U}$ is the output range of $\ses$. 
\item appends  $s'_{\st 1}, g, C', p, r'_{\st 1}, \delta'_{\st 0}, \delta'_{\st 1}$ to the view and outputs the view. 
\end{enumerate}

Now, we explain why the views in the ideal and real models are indistinguishable. Recall, in the real model, $\p_{\st 1}$ receives $\delta_{\st s_2}=g^{\st r_2}$ and  $\delta_{\st 1- s_2}=g^{\st a- r_2}$ from  $\p_{\st 2}$. 
Since $a$ and $r_{\st 2}$ have been picked uniformly at random and unknown to the adversary due to DL assumption,   $\delta_{\st s_2}$ and $\delta_{\st 1-{\st s_2}}$ (or $\delta_{\st 0}$ and $\delta_{\st 1}$) have identical distributions and are indistinguishable from random values (of the same field). 

In the ideal model, $\delta'_{\st 0}$ and  $\delta'_{\st 1}$ have been picked uniformly at random; therefore, they have identical distributions too. Moreover, values $\delta_{\st s}, \delta_{\st 1-s}, \delta'_{\st 0}$, and $\delta'_{\st 1}$  have been defined over the same field, $\mathbb{Z}_{\st p}$. So, they have identical distributions and are indistinguishable. 

Due to the security of $\ses$ each share $s_{\st 1}$ is indistinguishable from a random value $s'_{1}$, where $s'_{\st 1}\in \mathbb{U}$. Furthermore, $(r_{\st 1}, C)$ and $(r'_{\st 1}, C')$ have identical distributions, as they are picked uniformly at random from $\mathbb{Z}_{\st p}$.  Values $g$ and $p$ in the real and ideal models have identical distributions as they have been picked uniformly at random.

Hence, the two views are computationally indistinguishable, i.e., Relation \ref{equ::DUQ-OT-server-sim-} w.r.t. $\p_{\st 2}$ (in Section \ref{sec::DUQ-OT-definition})  holds.

\subsubsection{Corrupt \tp} T's view can be easily simulated.  It has an input $s$, but it receives no messages from its counterparts and receives no output from the protocol. Thus, its real-world view is defined as

 $$\view_{\st\tp}^{\st DUQ\text{-}OT}\big((m_{\st 0}, m_{\st 1}), \empt, \empt, \empt, s\big)=\{r_{\st \tp}, g, C,  p\}$$
 
  where  $r_{\st \tp}$ is the outcome of the internal random coin of  \tp and is used to generate random values.  

Ideal-model $\simm_{\st\tp}$ initiates an empty view, picks  $r'_{\st \tp}$, $g, C$, and $p$ uniformly at random, and adds them to the view. Since, in the real model, the adversary is passive, then it picks its randomness according to the protocol's description; thus, $r_{\st \tp}, g, C, p$ and $r'_{\st \tp}, g, C, p$ have identical distributions. 

Thus, the two views are computationally indistinguishable, i.e., Relation \ref{equ::DUQ-OT-t-sim-} (in Section \ref{sec::DUQ-OT-definition}) holds.
%
\end{proof}

\section{Proof of Lemma \ref{lemma::two-pairs-indis-}}\label{sec::proof-of-thorem}
Below, we prove Lemma \ref{lemma::two-pairs-indis-} presented in Section \ref{sec::rdqothf}.

\begin{proof}
First, we focus on the first element of   pairs $(g^{\st y_0}, \h(\beta_{\st 0}^{\st y_0}) \oplus m_{\st 0})$ and $(g^{\st y_1}, \h(\beta_{\st 1}^{\st y_1}) \oplus m_{\st 1})$. Since  $y_{\st 0}$ and $y_{\st 1}$ have been picked uniformly at random and unknown to the adversary, $g^{\st y_0}$ and $g^{\st y_1}$ are indistinguishable from random elements of group $\mathbb{G}$. 

Next,  we turn our attention to the second element of the pairs.  Given $C=g^{\st a}$, due to the DL problem, value $a$ cannot be extracted by a PPT adversary, except for a probability at most $\mu(\lambda)$. We also know that, due to CDH assumption, a PPT adversary cannot compute  $\beta^{\st y_{i}}_{\st i}$ (i.e., the input of $\h$), given $g^{\st y_{\st i}}, r_{\st 1}, C, g^{\st r_2}$, and $\frac{C}{g^{\st r_2}}$, where $i\in \{0, 1\}$, except for a probability at most $\mu(\lambda)$. 

We know that $\h$ has been considered as a random oracle and its output is indistinguishable from a random value. Therefore,  $\h(\beta^{\st y_{0}}_{\st 0})\oplus m_{\st 0}$ and  $\h(\beta^{\st y_1}_{\st 1})\oplus m_{\st 1}$ are indistinguishable from random elements of $\{0,1\}^{\st\delta}$, except for a negligible probability, $\mu(\lambda)$. 
\end{proof}

\section{\rdqothf in more Detail}\label{sec::DQ-HF-OT-detailed-protocol}
Figure \ref{fig::DQHT-OT} presents the \rdqothf that realizes \dqothf.


\begin{figure}[!h]
\setlength{\fboxsep}{.9pt}
\begin{center}
    \begin{tcolorbox}[enhanced,width=81mm,
    drop fuzzy shadow southwest,
    colframe=black,colback=white]
\begin{enumerate}[leftmargin=5.2mm]
\item \underline{\textit{$\se$-side Initialization:}} 
$\mathtt{Init}(1^{\st \lambda})\rightarrow pk$
\begin{enumerate}

\item chooses a sufficiently large prime number $p$.

\item selects random element
$C \stackrel{\st \$}\leftarrow \mathbb{Z}_p$ and generator $g$.
\item publish $pk=(C, p, g)$. 

\end{enumerate}

\item \underline{\textit{$\re$-side Delegation:}}
$\mathtt{\re.Request}( 1^{\st \lambda}, s, pk)\rightarrow req=(req_{\st 1}, req_{\st 2})$
\begin{enumerate}

\item split  the private index $s$ into two shares $(s_{\st 1}, s_{\st 2})$ by calling  $\ses(1^{\st \lambda}, s, 2, 2)\rightarrow (s_{\st 1}, s_{\st 2})$.

\item pick two uniformly random values: $r_{\st 1}, r_{\st 2} \stackrel{\st\$}\leftarrow\mathbb{Z}_{\st p}$.

\item send $req_{\st 1 }=(s_{\st 1}, r_{\st 1})$ to $\p_{\st1}$ and $req_{\st 2 }=(s_{\st 2}, r_{\st 2})$ to $\p_{\st 2}$.

\end{enumerate}

\item \underline{\textit{$\p_{\st 2}$-side Query Generation:}}
$\mathtt{\p_{\st 2}.GenQuery}(req_{\st 2}, ,pk)\rightarrow q_{\st 2}$

\begin{enumerate}

\item compute a pair of partial queries:
\vspace{-2mm}
  $$\delta_{\st s_2}= g^{\st r_2},\ \  \delta_{\st 1-s_2} = \frac{C}{g^{\st r_2}}$$
  
\item send $q_{\st 2}=(\delta_{\st 0}, \delta_{\st 1})$ to  $\p_{\st 1}$. 

\end{enumerate}

\item\underline{\textit{$\p_{\st 1}$-side Query Generation:}}
$\mathtt{\p_{\st 1}.GenQuery}(req_{\st 1}, q_{\st 2},pk)\rightarrow q_{\st 1}$

\begin{enumerate}

\item compute a pair of final queries as: 
\vspace{-2mm}
$$\beta_{\st s_1}=\delta_{0}\cdot g^{\st r_1},\ \ \beta_{\st 1-s_1}=\frac{\delta_{\st 1}} {g^{\st r_1}}$$

\item send $q_{\st 1}=(\beta_{\st 0}, \beta_{\st 1})$ to  $\se$.

\end{enumerate}

\item\underline{\textit{\se-side Response Generation:}} 
$\mathtt{GenRes}(m_{\st 0, 0}, m_{\st 1, 0},\ldots, m_{\st 0, z-1}, m_{\st 1, z-1},  pk, q_{\st 1})\rightarrow res$

\begin{enumerate}

\item abort if  $C \neq \beta_{\st 0}\cdot \beta_{\st 1}$.
\item compute a response as follows. $\forall t, 0\leq t\leq z-1:$

\begin{enumerate}[leftmargin=3.5mm]

\item  pick two random values $y_{\st 0, t}, y_{\st 1, t}  \stackrel{\$}\leftarrow\mathbb{Z}_{\st p}$.
\item compute response:
\vspace{-2mm}
   $$e_{\st 0, t} := (e_{\st 0, 0, t}, e_{\st 0, 1, t}) = (g^{\st y_{0, t}}, \h(\beta_{\st 0}^{\st y_{0, t}}) \oplus m_{\st 0, t})$$
$$e_{\st 1, t} := (e_{\st 1, 0, t}, e_{\st 1, 1, t}) = (g^{\st y_{1, t}}, \h(\beta_{\st 1}^{\st y_{1, t}}) \oplus m_{\st 1, t}) $$

\end{enumerate}

\item send $res=(e_{\st 0, 0}, e_{\st 1, 0}),..., (e_{\st 0, z-1}, e_{\st 1, z-1})$ to $\p_{\st 1}$. 

\end{enumerate}

\item\underline{\textit{$\p_{\st 1}$-side Oblivious Filtering:}} 
$\mathtt{OblFilter}(res)\rightarrow  res'$

\begin{itemize}
\item[$\bullet$] forward $res'=(e_{\st 0, v}, e_{\st 1, v})$ to $\re$ and discard the rest of the messages received from \se. 
\end{itemize}

\item\underline{\textit{$\re$-side Message Extraction:}} 
$\mathtt{Retrieve}(res', req,   pk)  \rightarrow m_{\st s}$

\begin{enumerate}

\item set $x=r_{\st 2}+r_{\st 1}\cdot(-1)^{\st s_2}$.

\item retrieve message $m_{\st s,v}$ by setting:  
 $$m_{\st s, v}=\h((e_{\st s, 0, v})^{\st x})\oplus e_{\st s, 1 v}$$

\end{enumerate}

\end{enumerate}
\end{tcolorbox}
\end{center}
\vspace{-2mm}
    \caption{\rdqothf: Our protocol that realizes \dqothf.}
    \label{fig::DQHT-OT}
    \vspace{-3mm}
\end{figure}

\begin{theorem}\label{theo::DQ-OTHF-sec}
Let $\mathcal{F}_{\scriptscriptstyle\dqothf}$ be the functionality defined in Section \ref{sec::DQOT-HF}. If  
DL, CDH, and RO assumptions hold, then \rdqothf (presented in Figure \ref{fig::DQHT-OT}) securely computes $\mathcal{F}_{\scriptscriptstyle\dqothf}$ in the presence of semi-honest adversaries, 
%
%
w.r.t. Definition \ref{def::DQ-OT-HF-sec-def}. 
\end{theorem}

Appendix \ref{sec::proof-of-DQ-OTHF-sec} presents the proof of Theorem \ref{theo::DQ-OTHF-sec}. 


\section{Proof of  Theorem \ref{theo::DQ-OTHF-sec}}\label{sec::proof-of-DQ-OTHF-sec}
Below, we prove the security of  \rdqothf, i.e., Theorem \ref{theo::DQ-OTHF-sec}.

\begin{proof}
To prove the above theorem, we consider the cases where each party is corrupt at a time.

\subsubsection{Corrupt \re} Recall that in \rdqothf, sender  \se holds a vector $\bm{m}$ of $z$ pairs of messages (as opposed to DQ-OT where \se holds only a single pair of messages). In the real execution, \re's view is:  
$\view_{\re}^{\st \rdqothf}\big(m_{\st 0}, m_{\st 1}, v,$ $  \empt, s\big) = \{r_{\st\re}, g, C, p, e_{\st 0,v}, e_{\st 1, v}, m_{\st s,v} \}$, where $C=g^{\st a}$ is a random value and public parameter, 
where $g$ is a random generator, $a$ is a random value, $p$ is a large random prime number, and $r_{\st \re}$ is the outcome of the internal random coin of  \re and is used to generate $(r_{\st 1}, r_{\st 2})$.

We will construct a simulator $\simm_{\st\re}$ that creates a view for \re such that (i) \re will see only a pair of messages (rather than $z$ pairs), and (ii) the view is indistinguishable from the view of corrupt \re in the real model. 
$\simm_{\st \re}$ which receives $(s, m_{\st s})$ from \re operates as follows.

 \begin{enumerate}
 \item initiates an empty view and appends uniformly random coin $r'_{\st\re}$ to it, where $r'_{\st\re}$ will be used to generate \re-side randomness. It chooses a large random prime number $p$ and a random generator $g$.

 \item sets $(e'_{\st 0}, e'_{\st 1})$ as follows: 
 \begin{itemize}
 \item splits $s$ into two shares: $\ses(1^{\st\lambda}, s, 2, 2)\rightarrow (s'_{\st 1}, s'_{\st 2})$.
 \item picks uniformly random values: $C', r'_{\st 1}, r'_{\st 2},y'_{\st 0}, y'_{\st 1}\stackrel{\st\$}\leftarrow\mathbb{Z}_{\st p}$.
 %
 %
 \item sets $\beta'_s=g^{\st x}$, where $x$ is set as follows: 
 \begin{itemize}
 \item[$*$] $x = r'_{\st 2} + r'_{\st 1}$, if $(s = s_{\st 1} = s_{\st 2} = 0)$ or $(s = s_{\st 1} = 1\wedge s_{\st 2} = 0)$.
 \item[$*$] $x = r'_{\st 2} - r'_{\st 1}$, if $(s =0 \wedge s_{\st 1} = s_{\st 2} = 1)$ or $(s = s_{\st 2} = 1\wedge s_{\st 1} = 0)$.
 \end{itemize}
 \item picks a uniformly random value $u\stackrel{\st\$}\leftarrow\mathbb{Z}_{p}$ and then sets $e'_s = (g^{\st y'_s}, \h(\beta'^{\st y'_s}_{\st s})\oplus m_{\st s})$ and $e'_{\st 1-s} = (g^{\st y'_{1-s}}, u)$. 
 \end{itemize}
 \item  appends $(g, C', p, r'_{\st 1}, r'_{\st 2},  e'_{\st 0}, e'_{\st 1}, m_{\st s})$ to the view and outputs the view.
 \end{enumerate}

The above simulator is identical to the simulator we constructed for DQ-OT. Thus, the same argument that we used (in the corrupt \re case in Section \ref{sec::DQ-OT-proof}) to argue why real model and ideal model views are indistinguishable, can be used in this case as well.  That means, even though \se holds $z$ pairs of messages and generates a response for all of them, \re's view is still identical to the case where \se holds only two pairs of messages. 
Hence, Relation \ref{equ::reciever-sim-DQ-OT-HF}  (in Section \ref{sec::DQOT-HF}) holds.


 \subsubsection{Corrupt \se}   This case is identical to the corrupt \se in the proof of DQ-OT (in Section \ref{sec::DQ-OT-proof}) with a minor difference. Specifically, the real-model view of \se in this case is identical to the real-model view of \se in DQ-OT; however, now $\simm_{\st\se}$ receives a vector $\bm{m}=[(m_{\st 0, 0},m_{\st 1, 0}),...,$ $ (m_{\st 0, z-1},$ $m_{\st 1, z-1})]$ from \se, instead of only a single pair that $\simm_{\st\se}$ receives in the proof of DQ-OT. $\simm_{\st\se}$ still operates the same way it does in the corrupt \se case in the proof of DQ-OT. Therefore, the same argument that we used (in Section \ref{sec::DQ-OT-proof}) to argue why real model and ideal model views are indistinguishable (when \se is corrupt), can be used in this case as well.
 
 Therefore, Relation \ref{equ::sender-sim-DQ-OT-HF}  (in Section \ref{sec::DQOT-HF}) holds.

\subsubsection{Corrupt $\p_{\st 2}$} This case is identical to the corrupt   $\p_{\st 2}$ case in the proof of DQ-OT. So, Relation \ref{equ::server2-sim-DQ-OT-HF}  (in Section \ref{sec::DQOT-HF}) holds.

\subsubsection{Corrupt $\p_{\st 1}$} In the real execution, $\p_{\st 1}$'s view is: 
 
$\view_{\st\p_1}^{\st \rdqothf}\big((m_{\st 0}, m_{\st 1}), v,$ $ \empt, s\big)=\{g, C, p, s_{\st 1}, r_{\st 1}, \delta_{\st 0}, \delta_{\st 1}, (e_{\st 0,0}, e_{\st 1,0}),$\\ $\ldots, (e_{\st 0, z-1},$ $ e_{\st 1, z-1})\}$. Ideal-model $\simm_{\st\p_{1}}$ that receives $v$ from $\p_{\st 1}$ operates as follows.

\begin{enumerate}
\item initiates an empty view. It selects a large random prime number $p$ and a random generator $g$.

\item picks two random values $\delta'_{\st 0}, \delta'_{\st 1}\stackrel{\st\$}\leftarrow\mathbb{Z}_{\st p}$. 
\item picks two uniformly random values $s'_{\st 1}\stackrel{\st\$}\leftarrow\mathbb U$ and $C', r'_{\st 1}\stackrel{\st\$}\leftarrow\mathbb{Z}_{\st p-1}$, where $\mathbb{U}$ is the output range of $\ses$. 

\item picks $z$ pairs of random values as follows  $(a_{\st 0,0}, a_{\st 1,0}),..., (a_{\st 0,z-1},$ $ a_{\st 1,z-1})\stackrel{\st\$}\leftarrow\mathbb{Z}_{\st p}$. 
\item appends  $s'_{\st 1}, g, C',p, r'_{\st 1}, \delta'_{\st 0}, \delta'_{\st 1}$ and pairs $(a_{\st 0,0}, a_{\st 1,0}),..., (a_{\st 0,z-1},$ $ a_{\st 1,z-1})$ to the view and outputs the view. 
\end{enumerate}

Now, we explain why the views in the ideal and real models are indistinguishable. The main difference between this case and the corrupt $\p_{\st 1}$ case in the proof of DQ-OT (in Section  \ref{sec::DQ-OT-proof}) is that now $\p_{\st 1}$ has $z$ additional pairs $(e_{\st 0,0}, a_{\st 1,0}),..., (e_{\st 0,z-1},$ $ a_{\st 1,z-1})$. Therefore, regarding the views in real and ideal models excluding the additional $z$ pairs, we can use the same argument we provided for the corrupt $\p_1$ case in the proof of DQ-OT to show that the two views are indistinguishable. Moreover, due to Lemma \ref{lemma::two-pairs-indis-}, the elements of each pair  $(e_{\st 0,i}, e_{\st 1,i})$ in the real model are indistinguishable from the elements of each pair  $(a_{\st 0,i}, a_{\st 1,i})$ in the ideal model, for all $i$, $0 \leq i \leq z-1$.  Hence, Relation \ref{equ::server1-sim-DQ-OT-HF} (in Section \ref{sec::DQOT-HF}) holds.
\end{proof}

\section{\rduqothf's Security Proof}\label{sec::proof-of-DUQ-OT-HF}


We prove the security of  \rduqothf, i.e., Theorem \ref{theo::DUQ-OTHF-2-sec}.

\begin{proof}
To prove the theorem, we consider the cases where each party is corrupt at a time.

\subsubsection{Corrupt \re} In the real execution, \re's view is: 

$\view_{\re}^{\st \rduqothf}\big(\bm{m}, (v, s, z)$ $  \empt, \empt, \empt\big) = \{r_{\st\re}, g, C, p, r_{\st 3}, s_{\st 2}, {o}_{\st 0}, {o}_{\st 1}, $ $m_{\st s, v} \}$, where 
$g$ is a random generator, $p$ is a large random prime number, 
${o}_{\st 0}:=({o}_{\st 0, 0}, {o}_{\st 0, 1})$, ${o}_{\st 1}:=({o}_{\st 1, 0}, {o}_{\st 1, 1})$, 
$C=g^{\st a}$ is a random value and public parameter, $a$ is a random value, and $r_{\st \re}$ is the outcome of the internal random coin of  \re that is used to (i) generate $(r_{\st 1}, r_{\st 2})$ and (ii) its public and private keys pair for additive homomorphic encryption.

We will construct a simulator $\simm_{\re}$ that creates a view for \re such that (i) \re will see only a pair of messages rather than $z$ pairs, and (ii) the view is indistinguishable from the view of corrupt \re in the real model. 
$\simm_{\st \re}$ which receives $m_{\st s, v}$ from \re performs as follows.

 \begin{enumerate}
 \item initiates an empty view and appends uniformly random coin $r'_{\st\re}$ to it, where $r'_{\st\re}$ will be used to generate \re-side randomness. It selects a large random prime number $p$ and a random generator $g$.

  \item sets response  as follows: 
  
  \begin{itemize}
 \item picks random values: $C', r'_{\st 1}, r'_{\st 2},y'_{\st 0}, y'_{\st 1}\stackrel{\st\$}\leftarrow\mathbb{Z}_{p}$, $ r'_{\st 3}\stackrel{\st\$}\leftarrow\{0, 1\}^{\st\lambda}$, $s'\stackrel{\st\$}\leftarrow\{0,1\}$, and $u\stackrel{\st\$}\leftarrow\{0, 1\}^{\st\sigma+\lambda}$. 
 
 \item sets $x=r'_{\st 2}+r'_{\st 1}\cdot (-1)^{\st s'}$ and $\beta'_{\st 0}=g^{\st x}$. 
 \item sets $\bar{e}_{\st 0} := \big(\bar{e}_{\st 0,0}=g^{\st y'_0}, \bar{e}_{\st 0, 1}=\g(\beta'^{\st y'_0}_{\st 0})\oplus (m_{\st s,v}|| r'_{\st 3})\big)$ and $\bar{e}_{\st 1} := (\bar{e}_{\st 1,0}=g^{\st y'_{\st 1}}, \bar{e}_{\st 1,1}=u)$. 

 \item encrypts the elements of the pair under $pk$ as follows. $\forall i,i', 0\leq i, i'\leq 1: \bar{o}_{\st i,i'}=\enc(pk, \bar{e}_{\st i,i'})$. Let $\bar{o}_{\st 0}:=(\bar{o}_{\st 0, 0}, \bar{o}_{\st 0, 1})$ and $\bar{o}_{\st 1}:=(\bar{o}_{\st 1, 0}, \bar{o}_{\st 1, 1})$.

 \item randomly permutes the element of  pair $(\bar{o}_{\st 0}, \bar{o}_{\st 1})$. Let $(\bar{o}'_{\st 0}, \bar{o}'_{\st 1})$ be the result.

 \end{itemize}
 \item  appends $(g, C', p, r'_{\st 3}, s', \bar{o}'_{\st 0}, \bar{o}'_{\st 1}, m_{\st s, v})$ to the view and outputs the view.
 
  %

 %
 \end{enumerate}

 Now, we argue that the views in the ideal and real models are indistinguishable. As we are in the semi-honest model, the adversary picks its randomness according to the protocol description; so, $r_{\st \re}$ and $r'_{\st \re}$ model have identical distributions, so do values $(r_{\st 3}, s_{\st 2})$ in the real model and $(r'_{\st 3}, s')$ in the ideal model, component-wise. Moreover,  values $g$ and $p$ in the real and ideal models, as they have been picked uniformly at random.  


For the sake of simplicity, in the ideal model let $\bar{e}'_{\st j}=\bar{e}_{\st 1} = (g^{y'_{\st 1}}, u)$ and in the real model let $e'_{\st i}=e_{\st 1-s}=(g^{\st y_{\st 1-s, v}}, \g(\beta^{\st y_{1-s, v}}_{\st 1-s})\oplus (m_{\st 1-s, v}||r_{\st 3}))$, where $i, j\in\{0,1 \}$.  Note that $\bar{e}'_{\st j}$ and $e'_{\st i}$ contain the elements that the adversary gets after decrypting the messages it receives from $\p_{\st 1}$ in the real model and from $\simm_{\st\re}$ in the ideal model.

We will explain that  $e'_{\st i}$ in the real model and $\bar{e}'_{\st j}$ in the ideal model are indistinguishable. In the real model, it holds that $e_{\st 1-s}=(g^{\st y_{\st 1-s, v}}, \g(\beta^{\st y_{\st 1-s, v}}_{\st 1-s})\oplus (m_{1-s, v}||r_3))$, where $\beta^{\st y_{\st 1-s, v}}_{\st 1-s}=\frac{C}{g^{\st x}}=g^{\st a-x}$. Since $y_{\st 1-s, v}$ in the real model and $y'_{\st 1}$ in the ideal model have been picked uniformly at random and unknown to the adversary, $g^{\st y_{\st 1-s, v}}$ and $g^{\st y'_{1}}$  have identical distributions. Moreover, in the real model, given $C=g^{\st a}$, due to the DL problem, $a$ cannot be computed by a PPT adversary.  Also, due to CDH assumption, \re cannot compute  $\beta^{\st y_{1-s, v}}_{\st 1-s}$, given $g^{\st y_{1-s}}$ and $g^{\st a-x}$. We know that $\g$ is considered a random oracle and its output is indistinguishable from a random value. Therefore,  $\g(\beta^{\st y_{1-s, v}}_{\st 1-s})\oplus (m_{\st 1-s, v} || r_{\st 3})$ in the real model and $u$ in the ideal model are indistinguishable. This means that $e'_{\st i}$ and $\bar{e}'_{\st j}$ are indistinguishable too, due to DL, CDH, and RO assumptions. 
 
Also, ciphertexts $\bar{o}_{\st 1,0}=\enc(pk, g^{y'_{\st 1}})$ and $\bar{o}_{\st 1,1}=\enc(pk, u)$ in the ideal model and ciphertexts ${o}_{\st 1-s,0}=\enc(pk, g^{\st y_{1-s, v}})$ and ${o}_{\st 1-s,1}=\enc(pk, \g(\beta^{\st y_{1-s, v}}_{1-s})\oplus (m_{\st 1-s, v}||r_{\st 3}))$ in the real model have identical distributions due to IND-CPA property of the additive homomorphic encryption.

 Further,  (i) $y_{\st s, v}$ in the real model and $y'_0$ in the ideal model have been picked uniformly at random and (ii) the decryption of both $e'_{\st 1-i}$ and $\bar{e}'_{\st 1-j}$ contain $m_{\st s, v}$; therefore, $e'_{\st 1-i}$ and $\bar{e}'_{\st 1-j}$ have identical distributions. Also, $m_{\st s, v}$ has an identical distribution in both models.  Both $C$ and $C'$ have also been picked uniformly at random from $\mathbb{Z}_{\st p-1}$; therefore, they have identical distributions.

 In the ideal model, $\bar{e}_{\st 0}$ always contains encryption of actual message $m_{\st s, v}$ while $\bar{e}_{\st 1}$ always contains a dummy value $u$. However, in the ideal model the encryption of the elements of pair  $(\bar{e}_{\st 0}, \bar{e}_{\st 1})$ and in the real model the encryption of the elements of pair $(e_{\st 0,v}, e_{\st 1,v})$ have been randomly permuted, which results in $(\bar{o}'_{\st 0}, \bar{o}'_{\st 1})$ and $(o_{\st 0}, o_{\st 1})$ respectively.

Moreover, ciphertexts $\bar{o}_{\st 0,0}=\enc(pk, g^{\st y'_{\st 0}})$ and $\bar{o}_{\st 0,1}=\enc(pk, $ $\g(\beta'^{\st y'_0}_{\st 0})\oplus (m_{\st s,v}|| r'_{\st 3})))$ in the ideal model and ciphertexts ${o}_{\st s,0}=\enc(pk, g^{\st y_{\st s, v}})$ and ${o}_{\st s,1}=\enc(pk, \g(\beta^{y_{\st s, v}}_{\st s})\oplus (m_{\st s, v}||r_{\st 3}))$ have identical distributions due to IND-CPA property of the additive homomorphic encryption. Thus, the permuted pairs have identical distributions, too.

We conclude that the two views are computationally indistinguishable, i.e., Relation \ref{equ::DUQ-OT-HF-reciever-sim-} (in Section \ref{sec::Delegated-Unknown-Query-OT-HF}) holds. That means, even though \se holds $z$ pairs of messages and generates a response for all of them, \re's view is still identical to the case where \se holds only two pairs of messages.

 \subsubsection{Corrupt \se}   This case is identical to the corrupt \se in the proof of DUQ-OT (in Appendix \ref{sec::DUQ-OT-Security-Proof}) with a minor difference. Specifically, the real-model view of \se in this case is identical to the real-model view of \se in DUQ-OT. Nevertheless, now $\simm_{\st\se}$ receives a vector $\bm{m}=[(m_{\st 0, 0},m_{\st 1, 0}),...,$ $ (m_{\st 0, z-1},$ $m_{\st 1, z-1})]$ from \se, instead of only a single pair that $\simm_{\st\se}$ receives in the proof of DUQ-OT. $\simm_{\st\se}$ still carries out the same way it does in the corrupt \se case in the proof of DUQ-OT. Therefore, the same argument that we used (in Appendix \ref{sec::DUQ-OT-Security-Proof}) to argue why real model and ideal model views are indistinguishable (when \se is corrupt), can be used in this case as well.
 
 Therefore, Relation \ref{equ::DUQ-OT-HF-sender-sim-}  (in Section \ref{sec::Delegated-Unknown-Query-OT-HF}) holds.

\subsubsection{Corrupt $\p_{\st 2}$} This case is identical to the corrupt   $\p_{\st 2}$ case in the proof of DUQ-OT. Thus, Relation \ref{equ::DUQ-OT-HF-server-sim-}  (in Section \ref{sec::Delegated-Unknown-Query-OT-HF}) holds.

\subsubsection{Corrupt $\p_1$} In the real execution, $\p_{1}$'s view is: 
 
$\view_{\p_1}^{\st \rduqothf}\big(\bm{m}, (v, s, z), $ $ \empt, \empt, \empt \big)=\{g, C, p, s_1, \bm{w}_{\st j}, r_{\st 1}, \delta_{\st 0}, \delta_{\st 1},$\\ $ (e'_{\st 0,0}, e'_{\st 1,0}),$ $ ...,$ $ (e'_{\st 0, z-1},$ $ e'_{\st 1, z-1})\}$. Ideal-model $\simm_{\st\p_{\st 1}}$ operates as follows.

\begin{enumerate}
\item initiates an empty view. It selects a large random prime number $p$ and a random generator $g$.
\item picks two random values $\delta'_{\st 0}, \delta'_{\st 1}\stackrel{\st\$}\leftarrow\mathbb{Z}_{p}$. 

\item constructs an empty vector $\bm{w}'$. It picks $z$ uniformly at random elements $w'_{\st 0},..., w'_{\st z}$  from the encryption (ciphertext) range and inserts the elements into $\bm{w}'$. 

\item picks two uniformly random values $s'_{\st 1}\stackrel{\st\$}\leftarrow\mathbb U$ and $C', r'_{\st 1}\stackrel{\st\$}\leftarrow\mathbb{Z}_{p}$, where $\mathbb{U}$ is the output range of $\ses$. 

\item picks $z$ pairs of random values as follows  $(a_{\st 0,0}, a_{\st 1,0}),..., (a_{\st 0,z-1},$ $ a_{\st 1,z-1})\stackrel{\$}\leftarrow\mathbb{Z}_{\st p}$. 
\item appends  $s'_{\st 1},g, C', p, r'_{\st 1}, \delta'_{\st 0}, \delta'_{\st 1}$ and pairs $(a_{\st 0,0}, a_{\st 1,0}),..., (a_{\st 0,z-1},$ $ a_{\st 1,z-1})$ to the view and outputs the view. 
\end{enumerate}

Next, we argue that the views in the ideal and real models are indistinguishable. The main difference between this case and the corrupt $\p_{\st 1}$ case in the proof of DUQ-OT (in Appendix  \ref{sec::DUQ-OT-Security-Proof}) is that now, in the real model, $\p_1$ has: (i) a vector $\bm{w}_{\st j}$ of ciphertexts and (ii) $z$ pairs $(e'_{\st 0,0}, e'_{\st 1,0}),..., (e'_{\st 0,z-1},$ $ e'_{\st 1,z-1})$.  Therefore, we can reuse the same argument we provided for the corrupt $\p_1$ case in the proof of DUQ-OT to argue that the views (excluding $\bm{w}_{\st j}$ and $(e'_{\st 0,0}, e'_{\st 1,0}),..., (e'_{\st 0,z-1},$ $ e'_{\st 1,z-1})$) have identical distributions.

Due to Lemma \ref{lemma::two-pairs-indis-}, the elements of each pair  $(e'_{\st 0, i}, e'_{1, i})$ in the real model are indistinguishable from the elements of each pair  $(a_{\st 0, i}, a_{\st 1, i})$ in the ideal model, for all $i$, $0 \leq i \leq z-1$.  Also, due to the IND-CPA property of the additive homomorphic encryption scheme, the elements of $\bm{w}_{\st j}$ in the real model are indistinguishable from the elements of $\bm{w}'$ in the ideal model.

Hence, Relation \ref{equ::DUQ-OT-HF-server-sim-} (in Section \ref{sec::Delegated-Unknown-Query-OT-HF}) holds.

\subsubsection{Corrupt $\tp$} This case is identical to the corrupt \tp in the proof of DUQ-OT, with a minor difference; namely, in this case, \tp also has input $z$, which is the total number of message pairs that \se holds. Thus, we can reuse the same argument provided for the corrupt \tp in the proof of DUQ-OT to show that the real and ideal models are indistinguishable. Thus, Relation \ref{equ::DUQ-OT-HF-t-sim-} (in Section \ref{sec::Delegated-Unknown-Query-OT-HF}) holds.
\end{proof}


\section{Proof of Theorem \ref{theo::one-out-of-n-OT}}\label{theo::compiler-sec--}
\begin{proof}[Proof sketch]
Compared to an original $1$-out-of-$n$ OT, the only extra information that \se learns in the real model is a vector of $n$ encrypted binary elements. Since the elements have been encrypted and the encryption satisfies IND-CPA, each ciphertext in the vector is indistinguishable from an element picked uniformly at random from the ciphertext (or encryption) range. Therefore, it would suffice for a simulator to pick $n$ random values and add them to the view. As long as the view of \se in the original $1$-out-of-$n$ OT can be simulated, the view of \se in the new $1$-out-of-$n$ OT can be simulated too (given the above changes). 

Interestingly, in the real model, \re learns less information than it learns in the original $1$-out-of-$n$ OT because it only learns the encryption of the final message $m_s$. The simulator (given $m_s$ and $s$) encrypts $m_s$ the same way as it does in the ideal model in the $1$-out-of-$n$ OT. After that, it encrypts the result again (using the additive homomorphic encryption) and sends the ciphertext to \re. Since in both models, \re receives the same number of values in response, the values have been encrypted twice, and \re can decrypt them using the same approaches, the two models have identical distributions. 

Moreover, the response size is $O(1)$, because the response is the result of (1) multiplying two vectors of size $n$ component-wise and (2) then summing up the products which results in a single value in the case where each element of the response contains a single value (or $w$ values if each element of the response contains $w$ values). 
\end{proof}

\section{Proof of Theorem \ref{theo::ultra-OT}}\label{sec::ultrasonic-ot-proof}

\begin{proof}
We consider the case where each party is corrupt at a time.

\subsubsection{Corrupt Receiver \re} In the real execution, \re's view is: 

$\view_{\st\re}^{\st Supersonic\text{-}OT}\big((m_{\st 0}, m_{\st 1}), \empt,  s\big) = \{r_{\st\re}, e''_{\st 0}, m_{\st s} \}$, where $r_{\st \re}$ is the outcome of the internal random coin of  \re and is used to generate $(s_{\st 1}, s_{\st 2}, k_{\st 0}, k_{\st 1})$.  
 Below, we construct an idea-model simulator $\simm_{\re}$ which receives $(s, m_{\st s})$ from \re.

\begin{enumerate}
\item constructs an empty view and appends a uniformly random coin $r'_{\st \re}$ to the view. 

\item picks a random key $k \stackrel{\st\$}\leftarrow\{0, 1\}^{\st\sigma}$, using $r'_{\st \re}$. 

\item encrypts  message $m_s$ as follows $e=m_s\oplus k$. 

\item appends $e$ to the view and outputs the view. 

\end{enumerate}

Since we are in the passive adversarial model, the adversary picks its random coin $r_{\st \re}$  (in the real models) according to the protocol. Therefore, $r_{\st \re}$ and $r'_{\st \re}$ have identical distributions. Moreover, $e''_{0}$ in the real model and $e$ in the ideal model have identical distributions as both are the result of XORing message $m_s$ with a fresh uniformly random value. Also, $m_s$ is the same in both models so it has identical distribution in the real and ideal models. We conclude that Relation \ref{equ::ultra-ot-reciever-sim-} (in Section \ref{sec::Ultra-OT-definition}) holds.

\subsubsection{Corrupt Sender \se}
In the real execution, \se's view is: 

$\view_{\se}^{\st Supersonic\text{-}OT}\big((m_{\st 0}, m_{\st 1}), \empt,   s\big) = \{r_{\st\se}, s_{\st 1}, k_{\st 0}, k_{\st 1} \}$, where $r_{\st \se}$ is the outcome of the internal random coin of  \se and is used to generate its random values.   Next, we construct an idea-model simulator $\simm_{\st\se}$ which receives $(m_{\st 0}, m_{\st 1})$ from \se.

\begin{enumerate}
\item constructs an empty view and appends a uniformly random coin $r'_{\st \se}$ to the view. 

\item picks a binary random value $s' \stackrel{\st\$}\leftarrow\{0, 1\}$.

\item picks two uniformly random keys $(k'_{\st 0}, k'_{\st 1}) \stackrel{\st\$}\leftarrow\{0, 1\}^{\st\sigma}$.

\item  appends $s', k'_{\st 0}, k'_{\st 1}$ to the view and outputs the view. 
\end{enumerate}
Next, we explain why the two views are indistinguishable. 
The random coins $r_{\st \se}$ and $r'_{\st \se}$ in the real and ideal models have identical distribution as they have been picked according to the protocol's description (as we consider the passive adversarial model). Moreover, $s_{\st 1}$ in the real model and $s'$ in the ideal model are indistinguishable, as due to the security of the secret sharing scheme, binary share $s_{\st 1}$ is indistinguishable from a random binary value $s'$. Also, the elements of pair $(k_{\st 0}, k_{\st 1})$ in the real model and the elements of pair $(k'_{\st 0}, k'_{\st 1})$ in the ideal model have identical distributions as they have been picked uniformly at random from the same domain. Hence,  Relation \ref{equ::ultra-ot-sender-sim-} (in Section \ref{sec::Ultra-OT-definition}) holds.

\subsubsection{Corrupt Server \p}
In the real execution, \p's view is: 

$\view_{\st\p}^{\st Supersonic\text{-}OT}\big((m_{\st 0}, m_{\st 1}), \empt,   s\big) = \{r_{\st\p}, s_{\st 2}, e' \}$,   where $r_{\st \p}$ is the outcome of the internal random coin of  \p and is used to generate its random values and  $e'$ is a pair $(e'_{\st 0}, e'_{\st 1})$ and is an output of $\cper$.    Next, we construct an idea-model simulator $\simm_{\st\p}$.

\begin{enumerate}
\item constructs an empty view and appends a uniformly random coin $r'_{\st \p}$ to the view. 

\item picks a binary random value $s' \stackrel{\st\$}\leftarrow\{0, 1\}$.

\item constructs a pair $v$ of  two uniformly random values $(v_{\st0}, v_{\st1}) \stackrel{\st\$}\leftarrow\{0, 1\}^{\st\sigma}$.

\item appends $s', v$ to the view and outputs the view. 

\end{enumerate}

Since we consider the passive adversarial model, the adversary picks its random coins $r_{\st \p}$ and $r'_{\st \p}$ (in the real and ideal models, respectively) according to the protocol. So, they have identical distributions. Also, $s_{\st 2}$ in the real model and $s'$ in the ideal model are indistinguishable, as due to the security of the secret sharing scheme, binary share $s_{\st 2}$ is indistinguishable from a random binary value $s'$. 

In the real model, the elements of $e'$, which are $e'_{\st 0}$ and $e'_{\st 1}$ have been encrypted/padded with two fresh uniformly random values. In the ideal model, the elements of $v$ which are $v_{\st0}$ and $v_{\st 1}$ have been picked uniformly at random. Due to the security of a one-time pad, $e'_{\st i}$ ($\forall i, 0\leq i\leq 1$) is indistinguishable from a uniformly random value, including $v_{\st 0}$ and $v_{\st 1}$.  

Also, in the real model, the pair $e'$ that is given to \p  is always permuted based on the value of \se's share (i.e., $s_{\st 1}\in\{0, 1\}$) which is not known to \p; whereas, in the ideal model, the pair $v$ is not permuted. However, given the permuted pair $e'$ and not the permuted pair $v$, a distinguisher cannot tell where each pair has been permuted with a probability greater than $\frac{1}{2}$.   Therefore,  Relation \ref{equ::ultra-ot-server-sim-} (in Section \ref{sec::Ultra-OT-definition}) holds. 
 \end{proof}


\section{Proof of Correctness}\label{sec::Ultra-OT-Proof-of-Correctness}


In this section, we demonstrate that \re always receives the message $m_{\st s}$ corresponding to its query $s$. To accomplish this, we will show that (in step \ref{ultra-ot::e-double-prime}) the first element of pair $e''$ always equals the encryption of $m_s$. This outcome is guaranteed by the following two facts: (a) $s=s_{\st 1}\oplus s_{\st 2}$ and (b)  \se and \tp permute their pairs based on the value of their share, i.e., $s_{\st 1}$ and $s_{\st 2}$ respectively. 


\begin{table}[!htb]
\begin{center}
\scalebox{1}{
  \begin{tabular}{|c|c|c|c|c|}   
\hline
    $s$&$s_{\st 1}$&$s_{\st 2}$\\   
\hline

\hline
&\cellcolor{gray!20}1&\cellcolor{gray!20} 1\\

\cline{2-3}

\multirow{-2}{*}{$0$}&\cellcolor{gray!20}0&\cellcolor{gray!20}0\\

\hline

&\cellcolor{gray!20}1&\cellcolor{gray!20}0\\

\cline{2-3}

\multirow{-2}{*}{$1$}&\cellcolor{gray!20}0&\cellcolor{gray!20}1\\

\hline

\end{tabular}
}
\end{center}
\caption{Relation between query $s$ and behaviour of permutation $\cper$ from the perspective of $\se$ and $\p$. When $s_{\st i}=1$, $\cper$ swaps the elements of its input pairs and when $s_{\st i}=0$, $\cper$ does not swap the elements of the input pairs.}
\label{tab:ultra-ot-correctness}
\end{table}


As Table \ref{tab:ultra-ot-correctness} indicates, when $s=0$, then (i) either both $\se$ and \p permute their pairs or (ii) neither does. In the former case, since both swap the elements of their pair, then the final permuted pair $e''$ will have the same order as the original pair $e$ (before it was permuted). In the latter case, again $e''$ will have the same order as the original pair $e$ because neither party has permuted it. Thus, in both of the above cases (when $s=0$), the first element of $e''$ will be the encryption of $m_{\st 0}$. 
Moreover, as Table \ref{tab:ultra-ot-correctness} indicates, when $s=1$, then only one of the parties \se and \p will permute their input pair. This means that the first element of the final permuted pair $e''$  will always equal the encryption of $m_{\st 1}$. 

\clearpage

\end{document}